\documentclass{article}

\usepackage{arxiv}

\usepackage[utf8]{inputenc} 
\usepackage[T1]{fontenc}    
\usepackage{hyperref}       
\usepackage{url}            
\usepackage{booktabs}       
\usepackage{amsfonts}       
\usepackage{nicefrac}       
\usepackage{microtype}      
\usepackage{lipsum}		
\usepackage{graphicx}
\usepackage{natbib}
\usepackage{doi}

\usepackage{authblk}
\usepackage{hyperref}

\usepackage{algorithm}
\usepackage{algorithmic}
\usepackage{amsmath}
\usepackage{amsthm}
\usepackage{amssymb}
\usepackage{wrapfig}
\usepackage{dsfont}
\usepackage{bbm}
\usepackage[inline]{enumitem}
\usepackage{fancyvrb}
\usepackage{xcolor}
\usepackage{verbatim}

\usepackage{ifthen}
\newboolean{showcomments}
\setboolean{showcomments}{true} 

\newcommand{\yu}[1]{\ifthenelse{\boolean{showcomments}}
{ \textcolor{orange}{(Yu says:  #1)}}{}}
\newcommand{\guocheng}[1]{\ifthenelse{\boolean{showcomments}}
{ \textcolor{red}{(Guocheng says:  #1)}}{}}
\newcommand{\adam}[1]{\ifthenelse{\boolean{showcomments}}
{ \textcolor{blue}{(Adam says:  #1)}}{}}
\newcommand{\juba}[1]{\ifthenelse{\boolean{showcomments}}
{ \textcolor{purple}{(Juba says:  #1)}}{}}
\newcommand{\jianwei}[1]{\ifthenelse{\boolean{showcomments}}
{ \textcolor{red}{(Jianwei says: #1)}}{}}
\newcommand{\addref}[0]{\ifthenelse{\boolean{showcomments}}
{ \textcolor{red}{(add ref)}}{}}
\newcommand{\todo}[1]{\ifthenelse{\boolean{showcomments}}
{ \textcolor{red}{(To do:  #1)}}{}}
\newcommand{\note}[1]{\ifthenelse{\boolean{showcomments}}
{ \textcolor{red}{#1}}{}}

\newcommand{\agentSet}{\mathcal{S}}
\newcommand{\agentNum}{s}
\newcommand{\parSet}{\mathcal{N}}

\newcommand{\groupParSet}[1]{\mathcal{N}_#1}

\newcommand{\groupNum}{I}
\newcommand{\groupProb}[1]{q_{#1}}
\newcommand{\groupProbVec}{\mathbf{q}}
\newcommand{\group}[1]{i(#1)}

\newcommand{\cost}{c}
\newcommand{\reportedCost}{\tilde{c}}
\newcommand{\data}{x}

\newcommand{\budget}{B}

\newcommand{\agent}{k}
\newcommand{\decision}{\mathbbm{1} (\agent \in \groupParSet{i})}

\newcommand{\jointD}{\mathcal{D}}
\newcommand{\costD}{f}
\newcommand{\tD}{\mathcal{T}(\jointD, \allocation)}

\newcommand{\ratio}{\bar{\theta}}
\newcommand{\aveRatio}{\bar{\theta}}
\newcommand{\ratioVec}{\boldsymbol{\theta}}
\newcommand{\intraRatio}[1]{\bar{\theta}_{#1}}
\newcommand{\interRatio}[1]{\bar{\theta}_{-#1}}
\newcommand{\ratioVecone}{\boldsymbol{\theta}_f}
\newcommand{\ratioVecequ}{\boldsymbol{\theta}^*}
\newcommand{\ratioequ}{\bar{\theta}^*}
\newcommand{\ratioequg}{\bar{\theta}^*}

\newcommand{\gCorrelation}{\alpha}
\newcommand{\gIntraCor}[1]{\alpha_{#1}}
\newcommand{\gInterCor}[1]{\alpha_{-#1}}
\newcommand{\gCorrelationVec}{\boldsymbol{\alpha}}

\newcommand{\correlationVec}{\boldsymbol{\alpha}}
\newcommand{\correlation}{\alpha}
\newcommand{\intraCor}[1]{\alpha_{#1}}
\newcommand{\interCor}[1]{\alpha_{-#1}}

\newcommand{\benefitFnc}{w(\cdot)}
\newcommand{\benefit}[1]{w(#1)}

\newcommand{\g}{g}
\newcommand{\h}{h}

\newcommand{\gPrivacyFnc}{\g(\cdot)}
\newcommand{\gppl}[2]{\g(#1, \ratioVec_{#2}; \gCorrelationVec_#2)}
\newcommand{\gpplequ}[2]{\g(#1, \ratioVecequ_#2; \gCorrelationVec_#2)}

\newcommand{\gintraloss}[2]{{\g}(#1, \ratio_{#2}; \gCorrelation_{#2})}
\newcommand{\ginterloss}[2]{{\g}(#1, \ratio_{-#2}; \gCorrelation_{-#2})}

\newcommand{\privacyFnc}{\h(\cdot)}
\newcommand{\ppl}[2]{\h(#1, \ratioVec_#2; \gCorrelationVec_#2)}
\newcommand{\pplequ}[2]{\h(#1, \ratioVecequ_#2; \gCorrelationVec_#2)}

\newcommand{\intraloss}[2]{{\h}(#1, \ratio_{#2}; \correlation_{#2})}
\newcommand{\interloss}[2]{{\h}(#1, \ratio_{-#2}; \correlation_{-#2})}

\newcommand{\EFnc}{\mathbb{E}}
\newcommand{\E}[1]{\mathbb{E}_{#1}}

\newcommand{\allocation}{A}

\newcommand{\parameter}{\mu}
\newcommand{\estimator}{\hat{\parameter}}

\newcommand{\bias}{\mathbb{B}(\estimator; \jointD, \allocation)}

\newcommand{\variance}{\mathbb{V}(\estimator; \jointD, \allocation)}

\newcommand{\weight}{\gamma}

\newcommand{\utilityP}[3]{u(#1, #2, \ratioVec_#3; \correlationVec_#3)}
\newcommand{\utility}[3]{u_#3(#1 | #2)}

\newcommand{\expectedUtility}[3]{\bar{u}_{#3} (#1 | #2)}

\newcommand{\payFnc}{P(\cdot)}
\newcommand{\paymentP}[2]{P(#1,\ratioVec_#2; \correlationVec_#2)}
\newcommand{\payment}[2]{P_{#2}(#1)}

\newcommand{\selectionProbFnc}{A(\cdot)}
\newcommand{\selectionProbP}[2]{A(#1,\ratioVec_#2; \correlationVec_#2)}
\newcommand{\selectionProb}[2]{A_#2(#1)}

\newcommand{\constantpartequ}[1]{\tau(\ratioVecequ_#1,\gCorrelationVec_#1)}
\newcommand{\constantpart}[1]{\tau(\ratioVec_#1,\gCorrelationVec_#1)}

\newcommand{\costthreshold}[1]{\hat{c}_#1}

\newcommand{\virtualcost}[1]{\omega_#1}
\newcommand{\virtualcostall}{\omega}

\newtheorem{remark}{Remark}
\newtheorem{assumption}{Assumption}

\newcommand{\SelectionProbability}{A}

\newcommand{\pay}{P}
\newcommand{\Cost}{c}

\newcommand{\ParticipantSet}{\mathcal{N}}

\newtheorem{Def}{Definition}
\newtheorem{The}{Theorem}
\newtheorem{Pro}{Proposition}
\newtheorem{Cor}{Corollary}
\newtheorem{Lem}{Lemma}

\newtheorem{Prop}{Property}
\newtheorem{Cla}{Claim}

\makeatletter
\newcommand{\printfnsymbol}[1]{%
  \textsuperscript{\@fnsymbol{#1}}%
}
\makeatother

\title{The Privacy Paradox and Optimal Bias-Variance Trade-offs in Data Acquisition}

\author[\space\space 1,4,5]{Guocheng Liao\space \thanks{These two authors contributed equally.}}
\author[2]{Yu Su\space\printfnsymbol{1}}
\author[\space\space 3]{Juba Ziani}
\author[\space\space 2]{Adam Wierman}
\author[\space\space 4,5]{Jianwei Huang}

\affil[1]{Department of Information Engineering, The Chinese University of Hong Kong}
\affil[2]{Department of Computing and Mathematical Sciences, California Institute of Technology}
\affil[3]{Department of Computer and Information Science, University of Pennsylvania}
\affil[4]{Shenzhen Institute of Artificial Intelligence and Robotics for Society}
\affil[5]{School of Science and Engineering, The Chinese University of Hong Kong, Shenzhen}
\date{}

\begin{document}
\pagenumbering{arabic}
\maketitle


\begin{abstract}
While users claim to be concerned about privacy, often they do little to protect their privacy in their online actions.  One prominent explanation for this ``privacy paradox'' is that when an individual shares her data, it is not just her privacy that is compromised; the privacy of other individuals with correlated data is also compromised. This information leakage encourages oversharing of data and significantly impacts the incentives of individuals in online platforms.  In this paper, we study the design of mechanisms for data acquisition in settings with information leakage and verifiable data.  We design an incentive compatible mechanism that optimizes the worst-case trade-off between bias and variance of the estimation subject to a budget constraint, where the worst-case is over the unknown correlation between costs and data. Additionally, we characterize the structure of the optimal mechanism in closed form and study monotonicity and non-monotonicity properties of the marketplace. 
\end{abstract}


\section{Introduction}

There is a fundamental discrepancy between privacy attitudes and the  behaviors of users online: users claim to be concerned about their privacy but do little to protect privacy in their actions. More specifically, users express their concerns about privacy, including the ambiguous distribution of data and its use by third party \cite{Informationprivacyresearch, acquisti2017nudges, goldfarb2012shifts}; however, when choosing services, users mainly focus on the popularity, convenience, price, etc, despite the potential risk of data misuse \cite{Informationexamining,examination}. This phenomenon is known as the \emph{privacy paradox} \cite{privacyparadox} and understanding the reasons behind this paradox and its consequences for the design of online platforms is an important goal for both computer scientists and economists.  

The privacy paradox is at the root of the behavior of individuals in modern online data marketplaces.  Online platforms gather data on billions of individuals in order to personalize advertising and customize other aspects of their systems.  However, such usage tends to provide little direct benefits for the users, a fact that is used as indirect evidence to argue that users provide a small value on privacy \cite{athey2017digital}.  Such an argument ignores the impact that an individual's participation decision has on others.  In particular, when an individual shares her data, it is not just her privacy that is compromised; the privacy of other individuals whose data is correlated with hers is also compromised. Thus, these other individuals are more likely to share their own data given that some has already been leaked \cite{acemoglu_privacy}. This simple, but often overlooked issue is at the root of the privacy paradox.  Information leakage due to correlation has been shown to lead to oversharing since each individual overlooks their own privacy concerns as a result of the negative externalities created by others' revelation decisions.  Thus, information leakage leads to the potential for significant economic and social inefficiency in data marketplaces.

In this paper, \emph{we study the impact of privacy concerns and information leakage on the design of data markets.} Specifically, we study the task of designing mechanisms for obtaining verifiable data from a population for a statistical estimation task, such as estimating the expected value of some function of the underlying data.

The goal of designing mechanisms for optimal data acquisition is a core piece of the emerging literature on data marketplaces.  A common motivating example is a setting where a healthcare platform is doing statistical analysis on its population of users.  While some data is measured accurately from their smart devices, other desired data may be about characteristics users do not wish to provide or may vary over time. Thus the healthcare platform has to conduct a survey among the users to obtain such information accurately  (e.g. giving the individual a smart device or having the individual fill out a form); however when administering such a survey the responses are likely biased. For example, if weight is the target, then the respondents may be biased towards low-weight samples.  Thus, the task of designing mechanisms to limit the bias and reduce the variance of estimates obtained from such surveys is crucial.  However, such a task is challenging due to the fact that the analyst does not know the distribution of the data and has a limited budget.   

A growing line of work has focused on the design of such optimal data acquisition mechanisms, e.g., \cite{aaron_grant, dataacquisition, shuran_yiling, abernethy2015low}.  Initiated by \cite{abernethy2015low, aaron_grant}, this line of work has led to the design of mechanisms for unbiased estimation with minimal variance in a variety of settings. However, in this literature it is assumed that all individuals will participate, thus unbiased estimation is possible.  The trade-off between bias and variance has been ignored to this point with the exception of \citet{shuran_yiling}, which still assumes all individuals will participate and does not consider privacy concerns.  Further, this line of work has not considered the issues created  by information leakage due to correlation between the participants. Information leakage creates significant incentives for increased data sharing and thus mechanisms that do not consider it directly will suffer from undetected bias and increased variance in the obtained estimates.  Modeling the incentives created by leakage potentially provides the analyst the opportunity to obtain an estimator of the same quality using a smaller budget, due to the externalities created by data correlation.

\textbf{Contributions.} In this paper, we provide the first characterization of an optimal mechanism for data acquisition in the setting where agents are concerned about privacy and their data is correlated with each other. As a result, information leakage due to data correlation not only contributes to an agent's privacy cost, but also to the privacy costs of others with correlated data. Additionally, the mechanism allows, for the first time, a trade-off between the bias and variance of the estimator when privacy cost is considered. This offers the analyst freedom to tailor towards an emphasis on either bias or variance of the estimator depending on the contextual goals.

Specifically, we propose a novel model for data acquisition. The novelty of our model is a result of three important components. 
First, we introduce the privacy cost to model impacts of data correlation. Unlike modeling data correlation on an individual level in \cite{acemoglu_privacy}, 
we divide the agents into different groups and assume that agents within the same group share a same correlation strength. This gives us the power to work with any granularity of choice with regard to data correlation. For example, if every group has a relatively small size of agents, then our model of data correlation shifts towards a near individual level. In addition, an agent suffers a larger privacy cost if she joins the platform than that if she does not join. The choice of our privacy cost function enables us to model all these desired properties. Second, in reality, not every agent always decides to join the platform. Thus, we introduce the notion of participation rate as the ratio of the number of agents who join the platform to the number of total agents. This further allows us to study equilibrium with respect to participation rate, which is crucial since the mechanism impacts the participation rate, which in turn impacts the bias and variance. Third, given that not every agent always joins the platform, it is not always realistic to aim for an unbiased estimator. Instead, we minimize a linear combination of bias and variance of the estimator. Via a choice of constant weights for bias and variance, we are able to balance between these two metrics of the estimator as desired.

Our main theoretical results provide a closed form solution of payment and allocation rules under a choice of equilibrium participation rate in order to achieve a truthful mechanism (Theorem  \ref{pro_payment} $\&$  Theorem \ref{the_opt_pro}). More specifically, we aim to minimize a linear combination of bias and variance subject to budget and truthfulness constraints. By considering a linear combination, we are able to emphasize either bias or variance of the estimator as desired. Moreover, we provide conditions for the optimality of an unbiased estimator in the case when it is possible to achieve a full participation rate, i.e. every agent decides to join the platform.

Our results offer some interesting insights about mechanisms for data acquisition. First, an unbiased estimator is possible even if the budget is relatively small because we can meet the expected budget constraint using a small selection probability. However, an unbiased estimator is not always realistic in practice. As a result, it is important to optimize the bias-variance trade-off. Second, incorporation of privacy cost due to information leakage and sharing makes it possible for the analyst to underpay the agents to acquire the same data set. This can potentially lead to a relatively small payment for data, something that is frequently seen in practice. Last but not least, the privacy cost from leakage encourages more agents to join the platform, which coincides with the data oversharing phenomenon frequently observed in platforms today.

The design of mechanisms for data acquisition is known to be challenging even if we focus on an unbiased estimator and ignore privacy cost due to data correlation. However, obtaining our results requires overcoming additional challenges.  There are two technical innovations that enable our analysis. First, we need to introduce and characterize an equilibrium with respect to the participation rate, as the participation rate is an endogenous property of the mechanism design problem. This equilibrium adds considerable complexity to the analysis, but also provides insights about how the data acquisition mechanism depends on the popularity of the platform. Second, we introduce the notion of data correlation strength to characterize the privacy cost due to information leakage. This allows us to capture the impact of data correlation on the platform. Further, depending on a choice of group sizes, we are able to model information leakage at different granularity levels as desired.

\textbf{Related Literature.} The design of data markets has attracted a significant amount of interest in recent years. There is growing body of work studying a variety of aspects of data markets, including monetizing information via either dynamic sales or optimal mechanisms, e.g., \cite{babaioff2012optimal, horner2016selling}, exploiting personal information to improve allocation of resources in online markets, e.g., \cite{goldfarb2011online, bergemann2015selling, montes2019value}, optimal acquisition of information, e.g., \cite{aaron_grant, dataacquisition, shuran_yiling}. For a recent survey see \cite{bergemann_data_markets} and the references therein. Our work broadly falls into the last line of work, and focuses on the design of a mechanism for optimal data acquisition. The study of the optimal data acquisition has attracted a growing amount of attention across both economics and computer science. In this paper, we focus on designing a truthful mechanism in order to perform a statistical estimation task. This goal has received considerable attention. However, prior work does not consider the privacy cost of the participants.  Our work aims to fill the gap by considering privacy cost as an important factor for individuals' decision-making. 

In more detail, the prior work on optimal data acquisition can be divided into two categories depending on whether data is verifiable or not. In the first category, individuals' utilities often directly depend on the outcome of the statistical inference and they thus have an incentive to misreport their data, e.g., \cite{liu2016learning, liu2017sequential, liu2018surrogate, dekel2010incentive, meir2011strategyproof, meir2012algorithms, perote2003impossibility}. This is possible since there is no ground truth to verify the data. In the second category, individuals are assumed to report their data truthfully due to the ability of the analyst to verify the data, e.g., \cite{cummings2015accuracy,cai2015optimum,abernethy2015low, shuran_yiling, aaron_grant, dataacquisition}. In this paper, we consider the setting where data is verifiable and so focus our discussion on prior work in that category below.

The task in this context is to purchase data from individuals whose private costs are subject to an (expected) budget constraint. This model was introduced in \cite{abernethy2015low, aaron_grant}. The model assumes that data cannot be fabricated and that private costs of the participants are correlated with the data. Moreover, the participants do not derive utility or disutility from the estimation outcome. \citet{aaron_grant} minimizes a bound of the worst case variance while achieving an unbiased estimator while \citet{abernethy2015low} considers general supervised learning. Following these initial papers, a closed form result that directly minimizes the worst case variance is given in \cite{dataacquisition}. However, unbiased estimators are not always possible in realistic settings. When a biased estimator is considered, \citet{shuran_yiling} proposes a slightly different model, in which agents arrive in an online fashion and cost distribution is not known a priori. Under this model, \citet{shuran_yiling} studies a trade-off between bias and variance of the estimator rather than only focus on the unbiased estimators. This trade-off is also a core component of our work.

The prior work discussed above mostly focuses on data acquisition without considering the privacy concerns of the participants. Even when the privacy concerns are considered in the prior work, privacy cost is often interpreted as a simple cost, which does not capture the information leakage due to data correlation.  In contrast, privacy cost due to information leakage as a result of data correlation is a crucial concern in the model we consider in this paper. There is a recent line of work that examines data acquisition through the lens of differential privacy \cite{ghosh_selling_privacy,cummings2015accuracy,fleischer2012approximately,buyingwithoutverification,nissim_redrawing}. As agents might not be willing to report their data due to their privacy concerns, monetary incentives are given to encourage individuals to participate in order to balance between privacy of individuals and accuracy of estimation. However, in this line of work, the  impact of information leakage is not considered and thus the practical impact of the privacy paradox is not considered.

The information leakage that results from correlation between individuals' is at the root of the privacy paradox and is a crucial factor for data markets, as has been recognized by recent work in economics, e.g.,\cite{liao2018social,acemoglu_privacy}.  More specifically, \citet{acemoglu_privacy} recently introduced the concept of information leakage to account for privacy cost due to other individuals' data sharing.
Inspired by this, we incorporate the privacy cost due to heterogeneous data correlation in our model and aim to design a truthful mechanism to balance between bias and variance subject to an expected budget constraint.  The privacy concern of individuals, specifically the privacy cost due to data correlation, is essential in the design of such a mechanism. This is a critical feature of our model that is not present in any prior work on data acquisition.  The setting and model we consider differ considerably from \cite{acemoglu_privacy}.  We consider the problem of optimal data acquisition and the incentives created by information leakage, which have not been considered in this context.


\section{System Model}

We consider an online platform consisting of an analyst and many agents.  At a high level, the analyst aims to design a pricing mechanism to purchase private data from agents in order to perform a statistical estimation task, e.g., estimate the mean of the agents' data. Ideally, the analyst would like to purchase all the private data to obtain an unbiased estimator. However, given a limited budget, the analyst has to design a pricing mechanism to wisely select the data in order to balance between the bias and variance of the estimator.

To this end, we consider a family of pricing mechanisms that presents a menu to the agents. The menu consists of pairs of payments and probabilities of having agents' data selected for use in the estimation task. Given the menu, the agents report their costs and decide individually if they would like to join the platform or not. Once the agents make the decisions, the analyst selects data from an agent on the platform to purchase with the given probability selected from the menu. The analyst's goal in designing the mechanism is to determine the menu of payments and selection probabilities in order to perform a statistical estimation subject to an expected budget constraint. The form of this mechanism is classical, and adopted from, e.g., \cite{aaron_grant,dataacquisition,abernethy2015low}.

Considering privacy cost is critical to the design of such a mechanism. For the agents on the platform, they obtain benefit from participation but reveal some personal information through their interactions with the platform, which leads to the privacy cost. Further, when an agent's data is used, it negatively impacts not only the privacy of the agent revealing the data, but also the privacy of other agents whose data is correlated. Thus, even though some agents do not join the platform and do not directly share their data, they may still suffer a privacy cost through their peers' interactions with the platform due to data correlation. This privacy cost via \emph{information leakage} is an important and novel feature of the model presented here. 

We describe the full details of the model and the family of mechanisms we consider in the remainder of this section.  Figure \ref{fig:model} provides an overview of the model.

\begin{figure}[t]
    \centering
    \includegraphics[width=0.95\textwidth]{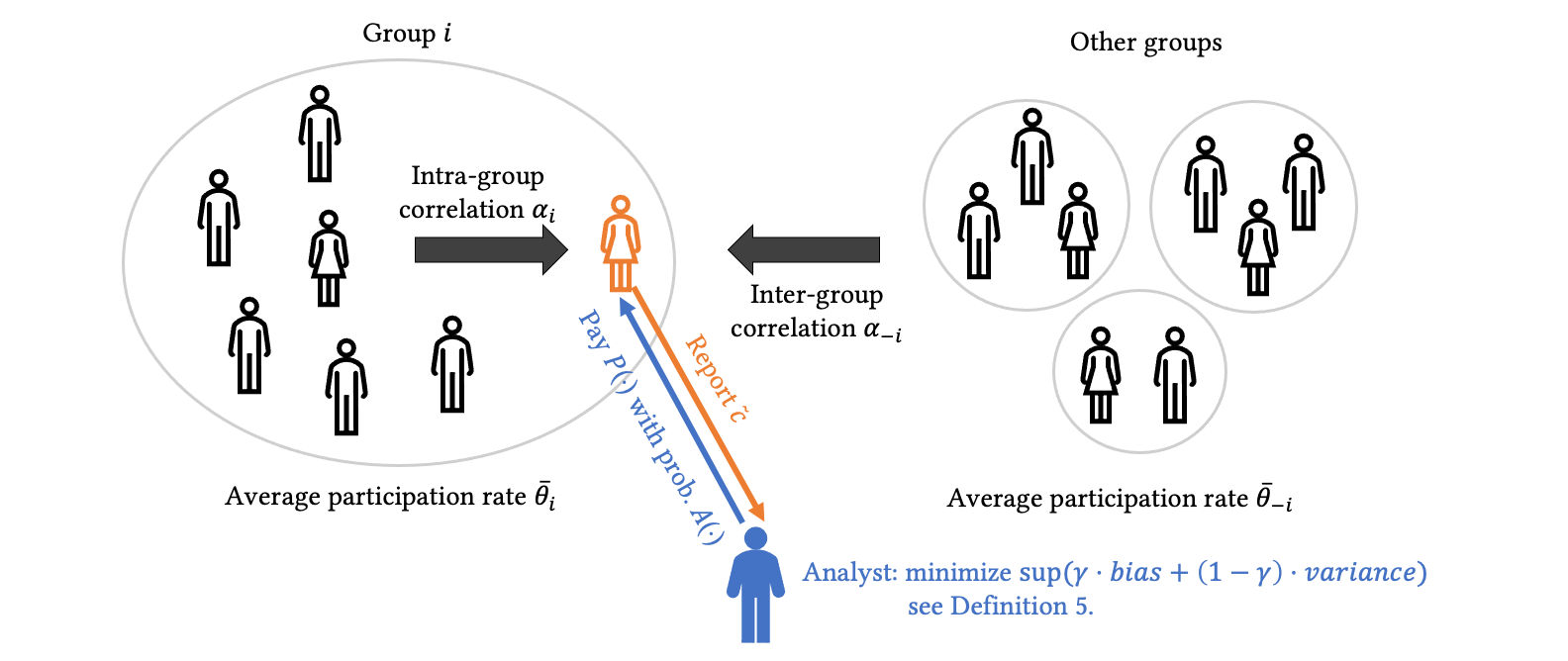}
    \caption{An illustration of the model from the perspective of a participating agent in group $i$.}
    \label{fig:model}
\end{figure}

\subsection{Agent model} \label{sec:agent}


We consider $\agentNum$ agents that hold data of interest to the analyst.  The set of agents is denoted by $\agentSet$. Every agent owns a data point. By reporting her data to the platform's survey, the agent incurs an overall cost $\cost$, which is known to her but not to the analyst. The overall cost consists of a combination of reporting cost and privacy cost, where the reporting cost results from the act of reporting the data while the privacy cost comes from both data sharing and data correlation. We discuss these costs further in Section \ref{sec:privacy}. 

We now present how the mechanism works. First, each agent is presented a price menu by the analyst that consists of a payment rule $\payFnc$ and selection probability $\selectionProbFnc$, both of which depend on the reported cost $\reportedCost$ of the agent. We interchangeably use selection probability or allocation rule, as convenient. Second, given the menu, an agent decides if she would like to join the platform. An agent who decides to join the platform is asked to report her cost, which determines the payment and selection probability. The payment $\payFnc$ is given to the agent if she joins the platform \emph{and} her data is selected (used) by the analyst.   More specifically, her data is selected with probability $\selectionProbFnc$, and she receives the payment only if her data is selected.


Whether an agent decides to join the platform relies on weighing the benefit of participation $\benefitFnc$, plus a potential payment, against the privacy cost that occurs as a result of her own or her peers' interactions with the platform. More specifically, for an agent on the platform, her privacy cost incurred includes both the cost from the agent herself sharing her data and the privacy cost due to her friends' sharing of possibly correlated data. We use $\privacyFnc$ to denote this combination.  In contrast, if an agent does not join the platform, she still suffers a  privacy cost $\gPrivacyFnc$  due to information leakage as a result of correlation between her data and the data of those agents who do join.
We use $\parSet$ to denote the set of agents who join the platform.
For agent $\agent$, her utility is as follows:
\begin{equation}\label{equ_uti_parameterized}
    \utility{\reportedCost}{\cost}{i} = \begin{cases}
        -\gPrivacyFnc, &\text{if } \agent \not\in \parSet, \\
        -\privacyFnc+\benefitFnc, &\text{w/ prob. } 1-\selectionProbFnc, \text{ if } \agent \in \parSet,\\ 
        \payFnc -c + \benefitFnc, &\text{w/ prob. }  \selectionProbFnc, \text{ if } \agent \in \parSet.
    \end{cases}
\end{equation}

Up to this point, we have described the model in a general way, without giving details on the form of the benefit of participation ($\benefitFnc$) and privacy cost ($\gPrivacyFnc$ and $\privacyFnc$).  In the remainder of this section, we introduce relevant models of these functions, which illustrate a variety of parameterized forms motivated by different potential settings. 



\subsubsection{Participation Benefit}

The participation benefit is the non-negative value received by agents who take advantage of the service provided by the platform. Intuitively, the more participants the platform attracts, the more valuable the platform's service becomes. Take Facebook as an example, the more friends of a user use Facebook, the more valuable Facebook as a social network is to this user. However, in light of the limited budget in practice, it is almost impossible to have every agent join the platform. As a result, we define the average participation rate of the population $\ratio$ to be the ratio of number of agents on the platform to the population of agents. The average participation rate $\ratio$ can be viewed as a measure of the popularity of the platform.\footnote{The participation rate is determined endogenously via the equilibrium, which we discuss in Section \ref{sec:mechanism}.} Moreover, let $\benefit{\ratio}$ denote the participation benefit as a continuous function of the average participation rate. The participation benefit is assumed to be non-decreasing in the average participation rate of the population $\ratio$. 

\subsubsection{Correlation Strength.}
Whether an agent decides to join the platform or not, an important source for her privacy cost comes from the information leakage due to data correlation. A stronger correlation naturally leads to more leakage and induces a larger privacy cost. Intuitively, if an agent's data is highly correlated with the rest of the agents, then she has a relatively large correlation strength. 
Moreover, some agents might share a stronger correlation with each other within a group of agents than with others outside the group. For instance, on a healthcare related platform, users who carry a common disease might share a similar pattern and thus their data is possibly highly correlated with each other. In order to capture the inter-group versus intra-group difference, we divide $\agentNum$ agents into $\groupNum$ groups, and agents within the same group $i$ share a common correlation strength $\correlationVec_i$. The correlation strength vector $\correlationVec_i$ is further defined as $\correlationVec_i \triangleq (\intraCor{i}, \interCor{i})$, where $\gIntraCor{i}$ and $\gInterCor{i}$ are used to respectively denote the correlation strength induced by agents inside group $i$ and those outside group $i$. Note that vectors are bold while scalars are not.

\subsubsection{Privacy cost} \label{sec:privacy}

Privacy cost is critical to an agent's decision about whether to join the platform. For an agent who does not join the platform, her privacy cost $\gPrivacyFnc$ comes entirely from information leakage through her peers' data sharing on the platform due to data correlation. In contrast, for an agent who joins the platform but does not report her data, not only her peers' actions but also her own interactions with the platform result in her privacy cost, which is denoted by $\privacyFnc$. Next, we introduce the parameterized form of privacy costs in detail.

\emph{Intra-group $\&$ Inter-group Privacy Loss.}
To differentiate the privacy cost due to correlation inside and outside the group, we further decompose privacy cost into a sum of intra-group cost and inter-group cost. Let the participation rate $\ratioVec_i$ denote a vector of participation rate within the group $\intraRatio{i}$ and that of the rest groups $\interRatio{i}$, i.e. $\ratioVec_i \triangleq  [\intraRatio{i},\interRatio{i}]$. For convenience, we later use $\aveRatio$ to denote the average participation rate of the overall population. For an agent in group $i$, her privacy cost is equal to a sum of intra-group cost and inter-group cost: $\gppl{\cost}{i} = \gintraloss{\cost}{i} + \ginterloss{\cost}{i}$ if the agent does not join the platform; otherwise her privacy cost of joining the platform is $\ppl{\cost}{i} = \intraloss{\cost}{i} + \interloss{\cost}{i}$.

Up to this point, we have introduced a parameterized form of privacy functions $\gppl{\cost}{i}$ and $\ppl{\cost}{i}$. 
{To derive more clear engineering insights,} we make two assumptions regarding these privacy functions. These assumptions are intuitive and consistent with the applications we consider.

\begin{assumption} \label{combine}
[Monotonicity $\&$ Boundedness] Both $\gppl{\cost}{i}$ and $\ppl{\cost}{i}$ are continuously  non-decreasing {in} cost $\cost$, participation rate $\intraRatio{i}$ and $\interRatio{i}$, and correlation strength $\intraCor{i}$ and $\interCor{i}$. Further, the difference $\ppl{\cost}{i}- \gppl{\cost}{i}$ is continuously increasing {in} cost $\cost$, {and  $\gppl{\cost}{i}$ is no greater than  $\ppl{\cost}{i}$. Both of them are bounded by $\cost$}:  $\gppl{\cost}{i} \le \ppl{\cost}{i} \leq \cost$.
\end{assumption}

To motivate the monotonicity assumption, first recall that cost is composed of reporting costs and privacy costs, as mentioned in Section \ref{sec:agent}. Thus, an agent with a high cost tends to value her privacy more, which is captured by the monotonicity with respect to costs. Second, intuitively, privacy costs increase as more agents join the platform. Indeed, the more agents join the platform and share their data, the more information that the analyst can infer about agents, including  those who do not join the platform. Third, the correlation strength characterizes how each agent's data is correlated with other agents' in the population. For an agent, a stronger correlation strength indicates that data sharing by other agents of the population could potentially cause more leakage, thus more privacy losses.

To understand the assumption $\gppl{\cost}{i}\leq\ppl{\cost}{i}\leq c$, we note that every agent suffers a privacy cost induced by her peers' activities on the platform due to data correlation even if she does not join the platform. In addition, for those who indeed join the platform, they tend to leak more information about themselves through their own activities on the platform. As a result, it is reasonable to assume that the privacy cost $\gppl{\cost}{i}$ of an agent, if she does not join the platform, is smaller than her privacy cost $\ppl{\cost}{i}$ if she joins but does not report her data. In other words,  sharing on the platform potentially leads to more privacy cost. The difference of these two parts models the privacy cost due to individual sharing. Naturally, for any agent, the privacy cost due to individual sharing tends to be higher if the agent has a larger overall cost. 
Since the overall cost consists of reporting cost and privacy cost, any privacy cost function is modeled as a fraction of the overall cost $\cost$, i.e. privacy cost is upper bounded by $\cost$.

\begin{assumption}
\label{Ass_linearincost}
The privacy cost  $\ppl{\cost}{i}$ is linear in cost $\Cost$ with parameter $b(\ratioVec_i;\gCorrelationVec_i)$, i.e., 
\begin{equation}
     \ppl{\Cost}{i} = \cost \cdot b(\ratioVec_i;\gCorrelationVec_i).
 \end{equation}
\end{assumption}

We assume that the privacy cost function of participation $\ppl{\cost}{i}$ is linear in cost $\cost$.   This assumption is  a consequence of the work in~\citet{ghosh_selling_privacy}, which shows that cost functions of the form $c_i \varepsilon$ are appropriate for settings using differential privacy, where $\varepsilon$ quantifies the amount of privacy leaked. Note that the parameter $b(\ratioVec_i;\gCorrelationVec_i)$ inherits the monotonicity with respect to the participation rate and correlation strength from that of the privacy cost function $\ppl{\cost}{i}$. 

\subsubsection{Agent's Utility}
After introducing the key components of agent's utility, we summarize the utility function as follows:

\begin{enumerate}[label=(\roman*)]
    \item If the agent does not join the platform, she only experiences a privacy cost $\gppl{\cost}{i}$ induced by information leakage due to data correlation.
    \item If the agent joins the platform but is not selected to report her data, her utility is $\ppl{\cost}{i}+\benefit{\aveRatio}$, where $\benefit{\aveRatio}$ is the participation benefit.
    \item If the agent joins the platform and is selected to report her data, she incurs her overall cost $\cost$, which includes her privacy cost. Based on the agent's reported cost $\reportedCost$, she receives a payment $\paymentP{\reportedCost}{i}$ and thus her utility is $\paymentP{\reportedCost}{i} -c + \benefit{\aveRatio}$.
\end{enumerate}


As the agent's privacy cost function is parameterized by the correlation strength parameter, so should the payment function $\paymentP{\reportedCost}{i}$, selection probability $\selectionProbP{\reportedCost}{i}$, utility $\utilityP{\reportedCost}{\cost}{i}$. However, for simplicity, we later use $\payment{\reportedCost}{i}$, $\selectionProb{\reportedCost}{i}$, and $\utility{\reportedCost}{\cost}{i}$, 
respectively to denote payment, selection probability and utility function. 
For agent $\agent$, her utility is as follows:
\begin{equation}\label{equ_uti}
    \utility{\reportedCost}{\cost}{i} = \begin{cases}
        -\gppl{\cost}{i}, &\text{if } \agent \not\in \parSet, \\
        -\ppl{\cost}{i}+\benefit{\aveRatio}, &\text{w/ prob. } 1-\selectionProb{\reportedCost}{i}, \text{ if } \agent \in \parSet,\\ 
        \payment{\reportedCost}{i} -c + \benefit{\aveRatio}, &\text{w/ prob. }  \selectionProb{\reportedCost}{i}, \text{ if } \agent \in \parSet.
    \end{cases}
\end{equation}
 Let the expected utility for an agent (reporting $\reportedCost$ while having a cost $\cost$) in group $i$ who joins the platform be denoted by $\expectedUtility{\reportedCost}{\cost}{i}$. Consequently,
\begin{equation}\label{expected_utility}
\begin{aligned}
\expectedUtility{\reportedCost}{\cost}{i} =& \selectionProb{\reportedCost}{i} \left[ \payment{\reportedCost}{i} -c + \benefit{\aveRatio} \right] + (1-\selectionProb{\reportedCost}{i}) \left[ - \ppl{\cost}{i} + \benefit{\aveRatio} \right].\\
=&\selectionProb{\reportedCost}{i}\left[\payment{\reportedCost}{i} -c + \ppl{\cost}{i}\right]- \ppl{\cost}{i}+ \benefit{\aveRatio}.
\end{aligned}
\end{equation}

\section{Mechanism Design} 


\label{sec:mechanism}

The analyst aims to perform an estimation task by designing the menu for the mechanism, which is composed of payment function and allocation rule for agents. To study this problem, we first introduce two standard desirable properties of the mechanism: \emph{truthfulness} (Definition \ref{def_truthful}) and \emph{budget feasibility} (Definition \ref{def_budget}), and then define the overall objective (a trade-off between bias and variance of the estimator) (Definition \ref{def_obj}) subject to these two properties.  Following that setup, we introduce the equilibrium participation rate (Definition \ref{def_equilibrium}), which is  a situation where no agent wants to alter her participation decision given the current participation rate. We then derive  structural results on the payment function and allocation rule that induce participants' truthful reporting in the equilibrium in Section \ref{sec_payment} (Theorem \ref{pro_payment}). Finally, we present monotonicity properties of the mechanism in Section \ref{sec_properties}. 

\subsection{Problem Statement}

To define the analyst's mechanism design problem, we first describe two classical, desirable properties, truthfulness and an expected budget constraint. 
First, we require the mechanism to be truthful, which is common, e.g. \cite{dataacquisition, shuran_yiling}.

\begin{Def}[Truthfulness]\label{def_truthful}
A  mechanism is truthful if for every participant with cost $\cost$, she can maximize her expected utility if 
she truthfully reports her cost, i.e.,
\begin{equation}\label{equ_truthfulness}
    \utility{\cost}{\cost}{i}\ge \utility{\reportedCost}{\cost}{i}, \quad
    \forall \reportedCost \neq \cost, \forall i.
\end{equation}
\end{Def}

Definition \ref{def_truthful} guarantees that rational agents on the platform will truthfully report their costs. Substituting the expected utility as in equation \eqref{expected_utility}, we get
\begin{equation}
\selectionProb{\cost}{i} \cdot \left[\payment{\cost}{i} -c + \ppl{\cost}{i}\right] \geq \selectionProb{\reportedCost}{i} \cdot \left[\payment{\reportedCost}{i} -c + \ppl{\cost}{i}\right], \quad
    \forall \reportedCost \neq \cost, \forall i.
\end{equation}




 Second, we constrain the mechanism to use a limited budget $\budget$, which limits the payments of the analyst to the agents for their data. Before presenting the budget constraint in details, we introduce some notation. Every agent owns a data point $\data$ and cost $\cost$.  These pairs $(\data, \cost)$, owned by agents in group $i$ ($i\in [I]$), follow a joint distribution $\jointD_i$. Define  $\jointD \triangleq \{ \jointD_i \}_{i\in [I]}$. Let $\costD_i$ denote the marginal distribution of the cost. We assume that $\costD \triangleq \{ \costD_i \}_{i\in [I]}$ is known to the analyst while $\jointD \triangleq \{ \jointD_i \}_{i\in [I]}$ is unknown.\footnote{This prior $\{\costD_i\}$ could be constructed from previous interactions between agents and data buyers. Knowledge of the cost (or valuation) distribution is a standard assumption in Bayesian mechanism design (e.g.~\citet{myerson}), and is assumed in the works of \citet{aaron_grant} and \citet{dataacquisition}.}

\begin{Def}[Expected Budget Constraint]\label{def_budget}
A mechanism with payment function $\pay$ and selection probability $\SelectionProbability$ satisfies the expected budget constraint $\budget$ if and only if
\begin{equation}
    \sum_{\agent: i = \group{\agent}} \E{c \sim \costD_{i}} [\payment{\cost}{i} \selectionProb{\cost}{i} \cdot \decision] \le \budget.\\
\end{equation}
\end{Def}

We enforce the budget constraint on the expected payment in  Definition \ref{def_budget}, which is a common approach in the prior work, e.g. \cite{aaron_grant, dataacquisition, shuran_yiling}. Note that the expectation is taken with respect to the marginal distribution of agent's cost $\costD$, and the agents who join the platform as only they potentially receive a payment by the allocation rule. For convenience, we use $\group{\agent}$ to denote the group index of agent $\agent$.

We now formally define bias and variance of the  estimator of the analyst.  Recall the analyst seeks to optimize a tradeoff between bias and variance in our setting. Specifically, the analyst wishes to learn an underlying parameter $\mu$ of the whole population, and he obtains an estimator $\estimator$ based on participants' data.
The estimator $\estimator$ is viewed as a random variable, and the randomness of the estimator comes from the joint distribution $\jointD$ and the allocation rule $\allocation$. We view the estimator as drawn from a distribution $\tD$.


\begin{Def}[Bias]\label{def_bias}
Given an allocation rule $\allocation$ and an instance of the true distribution $\jointD$, the bias of an estimator $\estimator$ is defined as follows:
\begin{equation}
    \bias = \left|\E{\estimator \sim \tD} \left[ \estimator - \parameter \right]\right|.
\end{equation}
\end{Def}

\begin{Def}[Variance]\label{def_variance}
Given an allocation rule $\allocation$ and an instance of the true distribution $\jointD$, the variance of an estimator $\estimator$ is defined as follows:
\begin{equation}
    \variance = \E{\estimator \sim \tD} \left[ (\estimator - \EFnc [\estimator])^2 \right].
\end{equation}
\end{Def}



Since the analyst does not know the joint distribution $\jointD$, he cannot directly optimize over bias and variance of the estimator. Instead, we consider the goal of minimizing the worst-case linear combination of bias (Definition \ref{def_bias}) and variance (Definition \ref{def_variance}) over all instantiations of $\jointD$ that are consistent with the marginal cost distribution $\costD$.

\begin{Def}[Worst-Case Bias-Variance Trade-off] \label{def_obj}
 \textcolor{black}{Given an allocation rule $\allocation$, an instance of the true distribution $\jointD$, and a combination parameter $\weight$, the worst-case bias-variance trade-off is the supremum of the linear combination of bias and variance:}  
 \begin{equation}
     \sup_{\costD \text{ consistent with } \jointD}\quad   \weight \cdot \variance+ (1-\weight)\bias.
 \end{equation}
\end{Def}

Using the above, we formally define the mechanism design problem as follows.

\begin{Def}[Mechanism Design Task]\label{def_mech_design}
Given an estimator $\estimator$, cost distribution $\costD$, correlation strength parameters, and fixed parameter $\weight$, the analyst aims to minimize a worst-case bias-variance trade-off by designing payment rule $\pay$ and allocation rule $\allocation$ subject to truthfulness and budgetary constraints:
\begin{equation}\label{equ_mech_design}
\begin{aligned}
\inf_{\allocation, \pay} \sup_{\costD \text{ consistent with } \jointD}\quad &   \weight \cdot \variance+ (1-\weight)\bias\\
\textrm{s.t.} \quad & \sum_{\agent: i = \group{\agent}} \E{c \sim \costD_{i}} [\payment{\cost}{i} \selectionProb{\cost}{i} \cdot \decision] \le \budget\\
& \selectionProb{\cost}{i}\left[\payment{\cost}{i} -c + \ppl{\cost}{i}\right] \geq \selectionProb{\reportedCost}{i}\left[\payment{\reportedCost}{i} -c + \ppl{\cost}{i}\right], \quad
    \forall \reportedCost \neq \cost, \forall i\\
\end{aligned}
\end{equation}
\end{Def}

Note that each agent suffers a privacy cost due to data correlation even if she does not join the platform. This negative utility offers the analyst freedom to compensate the agent a partial cost rather than the whole cost while motivating the agent to join the platform (individual rationality). As a result, we do not have to constrain the mechanism to satisfy a positive value of participation, which is a common requirement in the prior work in this area.  This is a novel, significant consequence of considering privacy leakage.



\subsection{Equilibrium Characterization}\label{sec_equilibrium}

To characterize the equilibrium that results under a given mechanism design, we consider the agents' decisions given a fixed participation rate. Recall that the participation rate is an aggregation of agents' decisions on whether to participate, and that an agent's utility depends on the participation rate. We emphasize that an agent's decision does not depend on other agents' individual decisions, rather only on the participation rate as an aggregate. This is natural for large platforms where a single agent's decision has little impact on others'. 

Before defining the equilibrium concept, we first introduce some notation. Recall that there are $I$ groups of agents parameterized by data correlation strength. We use $\groupProb{i}$ to represent the likelihood of a random agent coming from group $i$ in the population. For agents in group $i$, the cost follows a continuous distribution $\costD_i$  with  a support set $\mathcal{C} \triangleq [\Cost_{\min},\Cost_{\max}]$. We use $\ratio_i$ to denote the participation rate for group $i$. Let $\ratioVec \triangleq [\ratio_i]_{1\le i \le I}$, an array of average participation rate for all the groups, denote the participation rate profile. Naturally we have the average participation rate in the population satisfying $\aveRatio = \sum_{1\le i\le I} \groupProb{i}\cdot \aveRatio_i$.

For a given participation rate profile, an agent makes decisions by weighing her utility under different choices. For an agent with cost $\cost$ in group $i$, we use a binary variable $d_i(c | \ratioVec)$ to denote her decision. If the agent joins the platform, $d_i(c | \ratioVec) = 1$; otherwise $d_i(c | \ratioVec) = 0$.  We consider a non-atomic model where a single agent's decision has no effect on the aggregate participation rate when the agent set is large enough \cite{marketdemand}. This is a realistic assumption for a large platform. Every agent is assumed to be rational, i.e., she decides to join the platform only if her utility of non-participation is not lower than that of participation. This translates to the following mathematical expression:
\begin{equation}\label{equ_brparticipaiton}
d_i(c|\ratioVec) = \begin{cases}
1, & \text{if}\ \ \max_{\reportedCost}\expectedUtility{\reportedCost}{\cost}{i}\ge -\gppl{\cost}{i},\\
0, &  \text{otherwise}.
\end{cases}
\end{equation}
Since both payment and selection probability depend on the reported cost, an agent tends to report the cost that maximizes her expected utility of participation $\expectedUtility{\reportedCost}{\cost}{i}$.

For an agent in group $i, 1\le i\le I$, her decision on whether to participate is closely related to the participation rate  $\ratio_i$  of group $i$ and that of the rest groups $\ratio_{-i}$. Meanwhile, we notice that there are many possible participation rate profiles $\ratioVec$ that corresponds to one specific average participation rate $\ratio (= \sum_{1\le i\le I} \groupProb{i}\cdot \aveRatio_i)$. Thus, the average participation rate $\ratio$ alone is not enough to capture the agents' decisions. This motivates us to leverage the participation rate profile $\ratioVec$, instead of average participation rate $\ratio$ of all the groups, as the equilibrium concept.

We now define the equilibrium concept formally. This notion of equilibrium guarantees that agents' decisions are consistent with the participation rate profile. Note that we consider a strictly positive participation rate at equilibrium, i.e., $\ratioequ_i > 0, \forall i$, with an adequate budget; otherwise a zero participation rate for certain groups essentially leads to a biased estimator towards the rest of the groups with positive participation rates. 

\begin{Def}[Equilibrium]\label{def_equilibrium}
A participation rate profile  $\ratioVecequ=[\ratioequ_i]_{1\le i\le I}$ is an equilibrium, if for each group $i, 1\le i\le I$, the fraction of participating agents is exactly $\ratioequ_i$, i.e., 
\begin{equation}\label{equ_equilibrium}
   \int_{\cost_{\min}}^{\cost_{\max}}\mathbbm{1}\left[d_i(c|\ratioVecequ)=1\right]\cdot \costD_i(\cost)d\cost = \ratioequ_i, \quad \forall i.
\end{equation}
\end{Def}

\subsection{Payment Function}\label{sec_payment}

We now analyze the structure of the payment function in the mechanism associated with the equilibrium $\ratioVecequ$. The analyst's mechanism design problem in Definition \ref{def_mech_design} involves a  truthfulness constraint and a budgetary constraint. We first present the payment function that satisfies the truthfulness constraint for a given non-increasing allocation rule. Later, in Section \ref{sec_opt_tradeoff}, we further optimize over allocation rule to solve for the general optimization problem.

We present the payment function for group $i, 1\le i\le I$, and then characterize the requirements for a truthful mechanism. We begin with some definitions. For a fixed desired participation rate profile $\ratioVec$, we introduce  $\costthreshold{i}$  of group $i$, which satisfies the equality in \eqref{equ_ratio_threshold}. Later in this section, we explain in Corollary \ref{cor_threshold} that $\costthreshold{i}$ can be viewed as a cost threshold for agents in group $i$, and agents whose cost is no greater than this value would like to participate. 

\begin{equation}\label{equ_ratio_threshold}
	 \ratio_{i} = \int_{\Cost_{\min}}^{\costthreshold{i}} \costD_i(\Cost) d\Cost, 1\le i\le I.
\end{equation}
Next, we give  the payment function of group $i$ as follows

\begin{equation}\label{equ_payment}
\begin{aligned}
\payment{\reportedCost}{i} =\reportedCost-\ppl{\reportedCost}{i} 
+\frac{1}{\selectionProb{\reportedCost}{i}}\left(\left(1-b(\ratioVec_i;\gCorrelationVec_i)\right)\int_{\reportedCost}^{\cost_{\max}}\selectionProb{z}{i}dz + \constantpart{i}\right),
\end{aligned}
\end{equation}
where
\begin{equation}\label{equ_constantpart}
\constantpart{i} = \ppl{\costthreshold{i}}{i}-\gppl{\costthreshold{i}}{i}-\left(1-b(\ratioVec_i;\gCorrelationVec_i)\right)\int_{\costthreshold{i}}^{\cost_{\max}}\selectionProb{z}{i}dz - \benefit{\ratio}.
\end{equation}

The characterization of the payment function relies on allocation rule $\allocation$ and the associated participation rate profile $\ratioVec$. According to Myerson's Lemma \cite{myerson}, the allocation rule should be non-increasing
{in} the reported cost $\reportedCost$ in order to induce agents' truthfulness. Further, we present the requirements for a truthful mechanism in Theorem  \ref{pro_payment}. Such a mechanism can uniquely induce agents' truthful reporting of costs.  If the selection probability is strictly decreasing, the mechanism can induce agents' strict truthfulness.\footnote{By strict truthfulness, we are saying that the agent can maximize her utility of participation if and only if she truthfully reports her cost, i.e., $\utility{\cost}{\cost}{i}> \utility{\reportedCost}{\cost}{i},$ for all $\reportedCost \neq \cost$. That is, by removing the equality of (\ref{equ_truthfulness}) in Definition \ref{def_truthful}, we get the definition of strict truthfulness.} We highlight this result in Theorem \ref{pro_payment} and put its proof in Appendix \ref{appendixsec_payment}. 
	\begin{The}\label{pro_payment}
	The mechanism is truthful (strictly truthful, respectively) and induces participation profile $\ratioVecequ$ at equilibrium if and only if for every group $i \in [I]$, both of the following statements hold:
	\begin{enumerate}	[label=(\roman*)]
	    \item allocation rule $\selectionProb{\reportedCost}{i}$ is a non-increasing (decreasing, respectively) function of the reported cost $\reportedCost$;
	    \item payment function is given as in Equation \eqref{equ_payment} (as a function of $A$), where $\ratioVec \triangleq \ratioVecequ$ is the desired participation rate profile.
	\end{enumerate}    
\end{The}

For a mechanism as described in Theorem  \ref{pro_payment}, the agent would like to truthfully report her overall cost, as truthful reporting can maximize her expected utility.  The monotonicity of the allocation rule indicates that an agent with a high reported cost in the same group is less likely to be selected and get a payment.  
Furthermore, we would like to emphasize that the strictly decreasing property of the allocation rule induces strict truthfulness of the mechanism. Suppose that the allocation rule is fixed in a certain range of reported cost (i.e., not decreasing).  By reporting any cost in this range, the agents in this range can get the highest utility, which does not satisfy strict truthfulness definition. Thus, a strictly decreasing allocation rule is necessary to induce strict truthfulness.

\begin{Cor}\label{cor_threshold}
Under the mechanism in Theorem \ref{pro_payment}, the participation decisions of agents in group $i$ satisfy $d_i(c | \ratioVecequ) = 1$ if $c\le \costthreshold{i}$ and $d_i(c | \ratioVecequ) = 0$ otherwise.
\end{Cor}
This result highlights that, under the mechanism in Theorem \ref{pro_payment}, the decisions of the agents in each group demonstrate a threshold structure with respect to overall cost $\Cost$. Agents in group $i$ will participate and truthfully report her overall cost, if her overall cost is no greater than the threshold $\costthreshold{i}$; otherwise she will not participate.  


\subsection{Properties of the mechanism}\label{sec_properties}

We now study the connections among three key components of the mechanism design: data correlation, group participation rate, and payment (see Properties \ref{pr_pay_ratio}-\ref{pr_ratio_correlation} later in this section).  Existing discussions (e.g., \citet{acemoglu_privacy}) of the connections among data correlation, number of users, and payment, which demonstrated some interesting properties of data trading, motivate our study. For example, \citet{acemoglu_privacy} shows that the total payment to users is non-monotone in the number of users in some scenarios. It is interesting to explore such connections under the settings of this work. Moreover, studying the connections helps uncover the impacts these issues have on the mechanism design under our settings when we incorporate information leakage.

Many of the properties below highlight the complexity of data marketplaces with information leakage.  For example, our first property highlights that, even though a larger average participation rate means more participants,  the total payment for the entire group does not necessarily increase. This is because the individual payment for these agents will decrease as well.

\begin{Prop}\label{pr_pay_ratio}
Given a fixed selection probability $\allocation$, the total expected payment for agents in group $i$ is non-monotonic
{in} the average participation rate of group $i$.
\end{Prop}


However, despite the non-monotonicity of the expected payment 
{in} the average participation rate, there are some intuitive monotonicity properties that do hold for the payment after imposing one additional natural assumption.

\begin{Prop}\label{pr_payment_correlation}
Suppose the difference $\ppl{\costthreshold{i}}{i}-\gppl{\costthreshold{i}}{i}$ is non-increasing
{in} correlation strength. The payment to the agents in group $i$ who join the platform is decreasing 
{in} both intra-group correlation strength $\intraCor{i}$ and inter-group correlation strength  $\interCor{i}$. 
\end{Prop}

In the property above, we further require the difference function, i.e., $\ppl{\costthreshold{i}}{i}-\gppl{\costthreshold{i}}{i}$, to be non-increasing
{in} both intra-group correlation strength and inter-group correlation strength. Intuitively, this difference can be viewed as the extra privacy cost for the agent with cost equal to $\costthreshold{i}$ due to her own activities on the platform, and thus it might be even independent of data correlation strength. As a consequence, the payment for the agents who join the platform in group $i$ is decreasing in both intra-group and inter-group correlation strength, i.e., $\intraCor{i}$ and  $\interCor{i}$.

Some intuition for the above property comes from the observation that, when there is strong correlation among the agents' data, an individual will suffer a large privacy leakage even if they do not join the platform. As a result, they are incentivized to join the platform even when the individual payment is relatively low. Naturally, given a fixed participation rate profile and allocation rule, the total expected payment for the entire group decreases as a result.

However, in contrast to the above intuition, we do not have monotonicity of the participation rate in this setting, as we show in the property below.  Note that for most of this paper we treat the participation rate as something fixed (and optimized) by the analyst.  However, understanding its behavior is important for performing such an optimization.

\begin{Prop}\label{pr_ratio_correlation}
In group $i$, the group participation rate is non-monotonic in both the intra-group correlation strength $\intraCor{i}$ and inter-group correlation strength  $\interCor{i}$.
\end{Prop}

To develop some intuition for the above, we note that the analyst's mechanism design is subject to a budget constraint. Although stronger data correlation indicates lower payment, the total expected payment is also affected by the selection probability. If the selection probability increases, the analyst needs to induce a lower group participation rate so that the total expected payment does not exceed the budget. Otherwise, the analyst could induce a higher participation rate.

Our last property investigates what happens as the budget increases.  In this case the analyst can exploit the additional budget to improve either bias or variance. On the one hand, increasing the selection probability helps reduce the uncertainty of collected data, thus decreases the variance. On the other hand, inducing an equilibrium with a higher participation rate helps cover a wider range of participants' data, thus decreases the bias. 

\begin{Prop}\label{pr_pay_biasvariance}
As budget increases, the estimator achieves a better bias-variance trade-off, i.e., the optimal objective of the overall optimization problem as in Definition \ref{def_mech_design} reduces. 
\end{Prop}

\section{Optimization of the Worst-Case Bias-Variance Trade-off}\label{sec_opt_tradeoff}

We now discuss the design of a mechanism that optimizes the worst-case bias-variance trade-off.  First, we introduce the analyst's choice of estimator in Section \ref{sec_estimator}. Second, we  characterize the worst-case bias-variance trade-off of the estimator, given a selection rule and a participation rate profile at equilibrium, in Section \ref{sec_tradeoff}. Finally, we provide our full mechanism in Section \ref{sec_solution}. 

From a mathematical perspective, our bias-variance optimization problem optimizes over two types of variables: (i) the allocation rule, which is a collection of  probabilities that participants are selected to report their data and get the payments, and (ii) the participation rate profile in equilibrium, which controls the fraction of agents in each group that join the platform in the first place. However, in practice, a platform is unlikely to be willing to sacrifice a high participation rate to improve the performance of any single statistical task. Thus, optimizing over the participation rate is not attractice.  Thus, here, we treat the participation rate as given so that our work can  apply to any platform at any growth stage, e.g., either a start-up platform with a low participation rate or a popular platform with a high participation rate.

In what follows, we first show that the optimal allocation rule under a fixed equilibrium participation rate profile is monotone in the agents' virtual costs, as defined in Definition \ref{def_vcost} in Section \ref{sec_solution}, and provide a closed-form solution for this allocation rule. Given that optimizing over the participation rate is unrealistic, we instead present structural results and conditions for a full participation rate profile. This provides a connection to previous work on data acquisition since, when there is a full participation rate profile, the analyst sees a representative sample of the population and obtains an unbiased estimator. 

\subsection{Estimator}\label{sec_estimator}

The objective of the analyst's mechanism design problem in Definition \ref{def_mech_design} is to minimize the worst-case bias-variance trade-off of a given estimator. In this section, we describe the estimator of our choice. We then characterize its worst-case bias-trade-off variance in Section \ref{sec_tradeoff}.

The analyst is interested in estimating a population statistic using the mechanism. In this paper, we focus on estimating a mean over the population in question. For example, the analyst wishes to understand the average income, average BMI, or average ratings for a new movie.

We use the Horvitz-Thompson estimator \cite{HTestimator}, which is the unique unbiased linear estimator for the settings where an analyst must sample at different rates from different sub-populations; in our case, agents in different groups or with different costs may be selected with different probabilities $A(\cdot)$, and in turn define several distinct sub-populations. In such a setting, the standard sample mean estimator is biased towards the sub-populations from which the analyst samples most (relatively to the sub-population size). The Horvitz-Thompson estimator eliminates this bias by re-weighting each data point by the inverse of the selection probability associated with its sub-population. In our setting, the Horvitz-Thompson estimator reweights the data of each agent $k$ by $1/A^k$, where $A^k$ is the selection probability associated with agent $k$.

We define the Horvitz-Thompson estimator formally below in Definition \ref{Def_estimator}. Recall that $\ParticipantSet$ denotes the set of participants (i.e., agents who join the platform) and is of size $N$. Let $\mathcal{O}\subset \mathcal{N}$ denote the set of participants who are selected to report their data. Let $x_k$ be the data of agent $k$.
Suppose that the agents' data is independently and identically distributed. The Horvitz-Thompson estimator is then an unbiased estimator for the mean of data of the participants in set $\ParticipantSet$, i.e., the expectation of the estimator is equal to the expected value of the participants' data.

\begin{Def}[Horvitz-Thompson estimator] \label{Def_estimator}
The Horvitz-Thompson estimator is given by
\begin{equation}
\estimator = \frac{1}{N}\sum_{k\in \mathcal{N}} \frac{x_k \cdot \mathbbm{1}\{k\in \mathcal{O}\}}{A^k},
\end{equation}
where $A^k$ is the selection probability for agent $k$.
\end{Def}

Although the estimator is unbiased with respect to the set of participants $\ParticipantSet$, it may be biased with respect to the overall agent set $\agentSet$, when the average participation rate is less than one. When some agents do not participate, their data are not fed to the estimator and the estimator fails to capture the information of the non-participants. On the other hand, if the average participation rate is equal to one, i.e., all the agents are willing to participate, the estimator is naturally unbiased.  In conclusion, the participation rate at equilibrium is an important factor in the issue of estimation bias. Since there is an intrinsic trade-off between bias and variance, the  analyst might prefer low variance instead of low or zero bias in some cases, e.g., when the linear weight associated with the variance in equation \eqref{equ_mech_design} is high and the budget is large enough to use a large selection probability. Later, in Section \ref{sec_solution}, we present conditions under which the analyst prefers to enforce a full participation rate to achieve an unbiased estimator.

\subsection{Characterizing the worst-case bias-variance trade-off}\label{sec_tradeoff}
We characterize the worst-case bias-variance trade-off of the estimator in this section. We start by introducing some key notation. For simplicity,  we assume that the agent's data is binary from the set of $\{0,1\}$. In Remark \ref{rem_tradeoff}, we highlight that our analysis can be generalized to continuous data in $[0,1]$. The probability of an agent being in group $i$ is denoted by $\groupProb{i}$. Let $\groupProbVec\triangleq [\groupProb{i}]_{1\le i\le I}$. Let $p_{i}(\cost) \triangleq Pr[x = 1|c, i]$ denote the probability of an agent $k$ with cost $c$ in group $i$ having data point $x_k = 1$. Recall that according to Corollary \ref{cor_threshold}, an agent in group $i$ will join in the platform if her cost $\cost$ is no greater than the participation threshold $\costthreshold{i}$. Let $A \triangleq \{\selectionProb{c}{i}\}_{1\le i\le I}$, consisting of the selection probabilities in all groups. Furthermore, we assume a positive correlation between data and cost, which is mathematically described as follows.
\begin{assumption}\label{a_positivecorrealtion}
For each group $i$, $p_{i}(\cost)$ is a non-decreasing function of cost $\cost$. 
\end{assumption}
The assumption is often reasonable in practice. If the agent's data is of a higher value, then her cost is more likely to be higher as well, and vice versa. For example, agents with a higher income or higher BMI might care more about their privacy, hence their reporting costs tend to be higher for them as well.

\begin{Lem}\label{lem_tradeoff}
Under Assumption \ref{a_positivecorrealtion}, fix allocation rule $A$ and participation rate profile $\ratioVecequ$. Conditional on $\mathcal{N}$, the \textcolor{black}{supremum}  of linear combination of bias and variance of the estimator $\estimator$ in Definition \ref{Def_estimator} is 
\begin{equation}\label{equ_tradeoff}
\begin{aligned}
T(A,\ratioVecequ) = &\sup_{\costD \text{ consistent with } \jointD}  \weight \cdot \variance+ (1-\weight)\bias\\& =\sup_{p_{i}(c)\in [0,1], c \le \costthreshold{i} } \frac{\weight}{\agentNum\left(\ratioequ\right)^2}\left(\sum_{i} \groupProb{i}\int_{\cost_{\min}}^{\costthreshold{i}}\frac{p_i(c)}{\selectionProb{c}{i}}f_i(\cost)d\cost-\frac{1}{\ratioequ}\left(\sum_{i} \groupProb{i}\int_{\cost_{\min}}^{\costthreshold{i}}p_{i}(\cost)f_i(\cost)d\cost\right)^2\right) \\ &+(1-\weight)(1-\ratioequ) \cdot \left(1-\frac{1}{\ratioequ}\sum_{i} \groupProb{i}\int_{\cost_{\min}}^{\costthreshold{i}}p_{i}(\cost)f_i(\cost)d\cost\right).
\end{aligned}
\end{equation}
\end{Lem}

\begin{remark}\label{rem_tradeoff}
While we focus on the binary data case, we note that if the agents' data points $x_k$ are taken from the interval $[0,1]$, our proof can immediately be adapted to show that
\begin{equation}
\begin{aligned}
T(A,\ratioVecequ) = &\sup_{\costD \text{ consistent with } \jointD}  \weight \cdot \variance+ (1-\weight)\bias\\ \le& \sup_{p_{i}(c)\in [0,1], c \le \costthreshold{i} } \frac{\weight}{\agentNum\left(\ratioequ\right)^2}\left(\sum_{i} \groupProb{i}\int_{\cost_{\min}}^{\costthreshold{i}}\frac{p_i(c)}{\selectionProb{c}{i}}f_i(\cost)d\cost-\frac{1}{\ratioequ}\left(\sum_{i} \groupProb{i}\int_{\cost_{\min}}^{\costthreshold{i}}p_{i}(\cost)f_i(\cost)d\cost\right)^2\right) \\ &+(1-\weight)(1-\ratioequ) \cdot \left(1-\frac{1}{\ratioequ}\sum_{i} \groupProb{i}\int_{\cost_{\min}}^{\costthreshold{i}}p_{i}(\cost)f_i(\cost)d\cost\right).
\end{aligned}
\end{equation}

\noindent That is, the expression that we optimize over is an over-estimate of the worst-case bias-variance trade-off; in turn, our approach still provides an upper bound and does not underestimate the bias-variance trade-off when the data is non-binary. 
\end{remark}

The proof of Lemma \ref{lem_tradeoff} and Remark \ref{rem_tradeoff} are in Appendix \ref{apendixsec_tradeoff}. The basic idea is to derive the variance and the worst-case bias separately, and they both are fully characterized by the distribution of participants' data (i.e., $p_i(\cost)$ for $c\le \costthreshold{i}$). Thus, optimizing a linear combination of bias and variance gives the worst-case linear combination of bias and variance.


Lemma \ref{lem_tradeoff} presents  an  analytical formulation of the objective function with an auxiliary variable $p_i(c)$. This  facilitates our later analysis of the optimization problem, which is the minimization of \eqref{equ_tradeoff} subject to truthfulness constraint and budget constraint.

To conclude the subsection, we provide some interpretation about how the allocation rule $A$ and participation rate profile $\ratioVecequ$ affect the trade-off between bias and variance.   If the value of allocation  rule $A_i(\cost)$ is higher, the variance $\variance$ is lower. However, due to the budget constraint, a higher value of allocation  rule indicates a lower participation rate, which means a possibly higher bias. Thus, due to the budget constraint,  the analyst needs to carefully trade off bias and variance. 

The analyst's objective is to minimize the above  worst-case bias-variance trade-off by choosing the allocation rule $A$ and the participation rate profile $\ratioVecequ$ in the equilibrium.  Next, we discuss the optimal allocation rule and participation rate profile in details.

\subsection{Optimal Allocation Rule and Conditions of Full Participation Rate}\label{sec_solution}

We now derive solutions to the worst-case bias-variance trade-off optimization problem under the budget and truthfulness constraints. First, we present the optimal allocation rule given the participation rate profile at equilibrium. Second, we identify sufficient conditions under which the optimal participation rate is one and the analyst obtains an unbiased estimator. 

\subsubsection{Optimal Allocation Rule}
We first study the design of the allocation rule, given a desired participation rate profile $\ratioVecequ$ at equilibrium. Recall that we focus on $\ratioVecequ$ in which $\ratioequ_i\ge \ratio_{\min}, \forall i$ for some positive value $\ratio_{\min}>0$ to avoid the trivial case of complete non-participation of a group. 

In this section, we first introduce the notion of virtual cost in Definition \ref{def_vcost}, which is analogous to the classical notion of virtual value in~\cite{myerson}. The virtual costs help characterize the payment function that, given a fixed allocation rule, induces agents to report their privacy costs truthfully. The characterization of payment function is in Theorem \ref{pro_payment}.  We note that the virtual cost of an agent in group $i$ depends on her privacy cost ($ \pplequ{\cost}{i}$) and the cost distribution ($F_i$ and $\costD_i$) in group $i$. 

\begin{Def}[Virtual Cost]\label{def_vcost}
 In group $i$, given participation rate profile $\ratioVecequ$, the virtual cost of an  agent with cost $\cost$ is 
 \begin{equation}\label{equ_vcost}
 \phi_i(\cost;\ratioVecequ) = \cost-\pplequ{\cost}{i} + (1-b(\ratioVecequ_i;\gCorrelationVec_i))\frac{F_i(c)}{f_i(c)}.
 \end{equation}
 Here, $F_i$ and $\costD_i$ are cdf and pdf of cost in group $i$, respectively. 
 \end{Def}

Recall that the vector $\ratioVecequ_i$  consists of $\ratio^*_i$ and  $\ratio^*_{-i}$. The virtual cost $\phi_i(\cost;\ratioVecequ)$  in group $i$ is decreasing in the within-group participation rate $\ratio^*_i$ and the outside-group rate $\ratio^*_{-i}$, as $h(\cdot)$ and $b(\cdot)$ are increasing in both  $\ratio^*_i$ and  $\ratio^*_{-i}$. Recall that a higher participation rate means more privacy cost. Furthermore, we assume that the virtual cost within a group is non-decreasing in cost $c$ as in \cite{dataacquisition,myerson} (Assumption \ref{a_regularity}). Note that a uniform distribution is one example satisfying the assumption.    

\begin{assumption}\label{a_regularity}[Regularity]
The virtual cost $\phi_i(\cost;\ratioVecequ)$  in group $i$ is non-decreasing. Furthermore, $\costD_i(\cost)$ is twice differentiable, in which case $F_i(\cost)\costD'_i(\cost)\le2(\costD_i(\cost))^2$.
\end{assumption}

 


Now we are ready to present the optimal allocation rule of each group.
\begin{The}\label{the_opt_pro}
	Under Assumptions \ref{combine}-\ref{a_regularity}, given 
	a desired participation rate  profile $\ratioVecequ$, the optimal allocation rule of group $i$ is

	\begin{align}\label{equ_general_opt_pro}
	    A_i(\cost) = \begin{cases}\chi, \ \ &\text{if} \ \  \phi_{i}(\cost;\ratioVecequ) \le \hat{\phi}, \\
	\frac{\eta}{\sqrt{\phi_{i}(\cost;\ratioVecequ)}}, \ \ &\text{if} \  \ \hat{\phi}<\phi_{i}(\cost;\ratioVecequ)\le \phi_{i}(\costthreshold{i};\ratioVecequ).
	\end{cases}
	\end{align}
The characterizations of the constants $\eta$, $\chi$, and $\hat{\phi}$ depends on the system parameters including the budget $B$, number of agents $s$, distribution of virtual cost. See \eqref{equ_continuousA} in Appendix \ref{appendixsec_proof_the_opt}.
\end{The}

\begin{wrapfigure}{r}{0.5\textwidth}
    \centering
    \includegraphics[width=.38\textwidth]{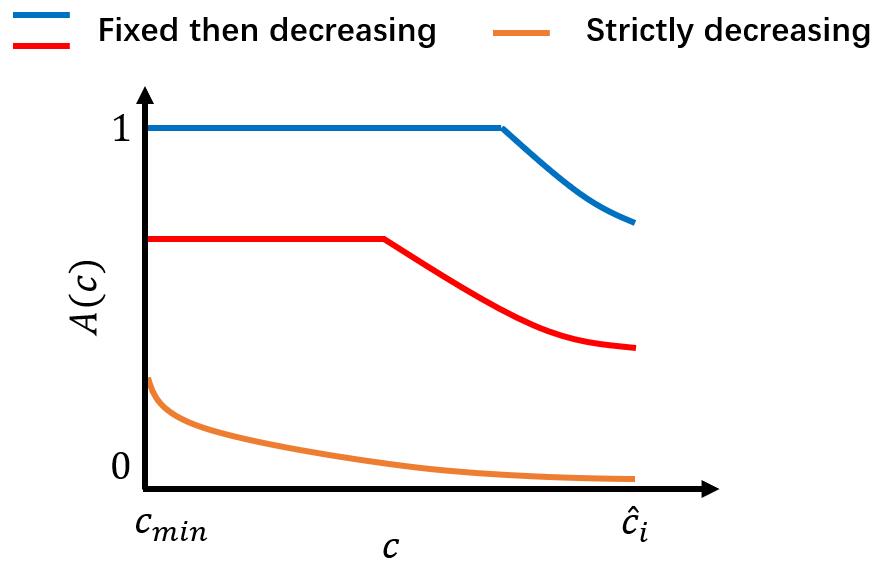}
    \caption{Illustration of two possible structures.}
    \label{fig_allocationrule}
\end{wrapfigure}
Theorem \ref{the_opt_pro} shows the optimal allocation rule of a group can have two possible structures: \emph{Fixed then Decreasing} (FtD) or  \emph{Strictly Decreasing} (SD), depending on system parameters such as budget $\budget$. Figure \ref{fig_allocationrule} provides an illustration. The FtD structure (the blue and red curves in Figure \ref{fig_allocationrule}) has an allocation rule  for the group that is firstly fixed in low-cost region (i.e., the region in which  $\phi_{i}(\cost;\ratioVecequ) \le \hat{\phi}$ holds) and then strictly decreasing in the high-cost region (inversely proportional to the square root of the virtual cost). Since the allocation rule is not strictly decreasing, the mechanism can only induce agents' weak (non-strict) truthfulness, according to Theorem \ref{pro_payment}. The other structure, SD (the orange curve in Figure \ref{fig_allocationrule}),  has an  allocation rule that is strictly decreasing (inversely proportional to the square root of the virtual cost) in the whole region. The strictly monotonic structure induces strict truthfulness.
{Note that the form of the optimal allocation rule parallels that of previous work~\cite{dataacquisition}, showing that a similar form is still optimal in a more general setting. As we have discussed, the generalization beyond~\cite{dataacquisition} to include the analyst's bias-variance trade-off, the agents' privacy costs, and data correlation add technical complexity and practicality. A priori, it is not clear that the optimal form would remain similar in the more general setting.}


Note that the optimal allocation rule presented in Theorem \ref{the_opt_pro} is defined for the  participating agents, whose cost are lower than threshold $\costthreshold{i}$ for group $i$. As in the case of non-participants, we only need to ensure monotonicity to induce truthfulness. A simple example is $A_i(\cost)=\epsilon$ for $c\ge \costthreshold{i}$, where $\epsilon$ is a small enough non-negative value. Note that this example does not make the total expected payment exceed the budget, as the non-participants do not get the payments.   In summary, the allocation rule of the cost lower than the threshold presented in Theorem \ref{the_opt_pro} and the allocation rule of the cost higher than the threshold make up a complete menu shown to the agents, which induces agents' truthfulness and optimizes the bias-variance trade-off.

We find that, when the budget $\budget$ is relatively small, the optimal allocation rule has a Strictly Decreasing structure. To state the result, first define 

\begin{equation}
 l(\ratioVecequ)\triangleq \sum_{i} q_i \ratioequg_i(\pplequ{\costthreshold{i}}{i}-\gpplequ{\costthreshold{i}}{i}-w(\ratioequ)).
\end{equation}
    \begin{Cor}\label{cor_opt_pro}
    Under Assumptions \ref{combine}-\ref{a_regularity}, the optimal allocation rule is 
    \begin{equation}
    A_i(\cost) = \frac{1}{\sqrt{\phi_{i}(\cost;\ratioVecequ)}} \cdot\frac{B/\agentNum-l(\ratioVecequ)}{\sum_{i}\groupProb{i}\int_{\cost_{\min}}^{\costthreshold{i}}\sqrt{\phi_{i}(\cost;\ratioVecequ)}\costD_i(\cost)d\cost},
    \end{equation}
    when the participation rate profile $\ratioVecequ$ satisfies the following inequality:
    \begin{equation}\label{equ_a}
        (B/s-l(\ratioVecequ))(2\weight+\ratioequ(1-\ratioequ)(1-\weight)\agentNum)<\weight\sqrt{\phi_{\min}(\ratioVecequ)} \cdot\sum_{i}\groupProb{i} \int_{\cost_{\min}}^{\costthreshold{i}}\sqrt{\phi_{i}(\cost;\ratioVecequ)}\costD_i(\cost)d\cost.
    \end{equation}
    Furthermore, the mechanism is strictly truthful.
    \end{Cor}

Corollary \ref{cor_opt_pro} is a special case of Theorem \ref{the_opt_pro}.  The inequality in (\ref{equ_a}) indicates that the budget is relatively low. Suppose the budget is large enough.  The analyst can keep the allocation rule in the low-cost region fixed, instead of strictly decreasing. In this case, the agents' reporting is weak truthful. However, if the budget is low as in (\ref{equ_a}), the optimal allocation rule is strictly decreasing, which induces agents' strict truthfulness. Strict truthfulness is a more desirable property than weak truthfulness.

\subsubsection{Sufficient Conditions for Unbiased Estimation}  

Next, we present simple, easy to verify, sufficient conditions for unbiased estimation. We focus on the low-budget regime, which induces agents' strict truthfulness, according to Corollary \ref{cor_opt_pro}. By substituting the allocation rule in Corollary \ref{cor_opt_pro} to the objective function (\ref{equ_tradeoff}), we arrive at an optimization problem  over $\ratioVecequ$ as follows:
\begin{equation}\label{equ_obj_ratio}
\max_{\ratioVecequ}  \ \  T^*(\ratioVecequ) = \frac{\weight}{\agentNum}\left(U(\ratioVecequ)-\frac{1}{\ratioequ}\right),
    \end{equation}
where 

\begin{equation}\label{equ_L}
        U(\ratioVecequ) \triangleq  \frac{1}{\left(\ratioequ\right)^2}\cdot \frac{\left(\sum_{i}\groupProb{i}\int_{\cost_{\min}}^{\costthreshold{i}}\sqrt{\phi_{i}(\cost;\ratioVecequ)}\costD_i(\cost)d\cost\right)^2}{B/s-l(\ratioVecequ)}.
    \end{equation}
The objective function in (\ref{equ_obj_ratio}) is complicated due to the complexity of $l(\ratioVecequ)$ in (\ref{equ_L}).  Furthermore, the complicated characterization of its derivative with respect to group $i$'s participation rate $\ratioequg_i$ adds to the complexity. Thus, we focus on identifying sufficient conditions under which the optimal participation rate is one, meaning that the analyst would like to have unbiased estimator. 

To begin, we introduce some notation and assumptions. We denote the full participation rate profile by $\ratioVecone\triangleq [1]^I\in \mathbb{R}^I $, each component of which is one. We define $\Delta_i \triangleq \ppl{\costthreshold{i}}{i}- \gppl{\costthreshold{i}}{i}$. It captures the additional privacy cost of an participant (compared with non-participation) whose cost is the cost threshold in group $i$. 

We now present sufficient conditions for achieving full participation and an unbiased estimator.


\begin{Pro}\label{Pro_fullpart_benefit}
Suppose that, for all $\ratioVecequ$ in which $\ratioequ_i\ge \ratio_{\min}>0, \forall i$,  (i) the inequality in (\ref{equ_a}) holds, and (ii) $w'(\ratioequ)\ge  D_i(\ratioVecequ,\budget,\weight,\agentNum,\groupProbVec)$ holds for the function $D_i(\ratioVecequ,\budget,\weight,\agentNum,\groupProbVec)$ defined in \eqref{equ_D} of Appendix \ref{appendixsec_fullpart_benefit},  where $w'$ is the derivative of function $w(\cdot)$.  Then, the optimal participation rate profile is $\ratioVecone$ and the analyst obtains an unbiased estimator.
       
        
    \end{Pro}

\begin{Pro}\label{pro_fullpart_loss}
Suppose that, for all $\ratioVecequ$ in which $\ratioequ_i\ge \ratio_{\min}>0, \forall i$, (i) the inequality in (\ref{equ_a})  holds, and (ii) $\frac{\partial \Delta_i}{\partial \ratioequ_i}\le  \delta_i(\ratioVecequ, \budget,\weight,\agentNum,\groupProbVec)$ holds for the function $\delta_i(\ratioVecequ, \budget,\weight,\agentNum,\groupProbVec)$ defined in \eqref{equ_delta} of Appendix \ref{appendixsec_fullpart_loss}. Then, the optimal participation rate profile is $\ratioVecone$ and the analyst obtains an unbiased estimator.
    

    \end{Pro}
    
    The proofs of Proposition \ref{Pro_fullpart_benefit}
and Proposition \ref{pro_fullpart_loss} are in Appendix \ref{appendixsec_fullpart_benefit} and Appendix \ref{appendixsec_fullpart_loss}, respectively.    To understand the conditions in these two propositions, note that Condition (i) in both propositions indicates that the budget is ``small'', i.e., lower than a certain threshold, which is related to the expectation of virtual cost in the population. It ensures that the allocation rule is strictly decreasing, according to Corollary \ref{cor_opt_pro}. Further, it holds if the budget itself is relatively low, or if the expectation of the virtual cost is relatively high. Recall that in the definition of virtual cost (Definition \ref{def_vcost}), the virtual cost is determined by both the cost and data correlation strength. In summary, Condition (i) indicates that the overall cost is relatively high, or the data correlation strength is relatively weak, or the budget is relatively low. 
    
Next we discuss Conditions (ii) in the above two propositions. Recall that the participation benefit  $\benefit{\ratio}$ is increasing in the participation rate. Condition (ii) in Proposition \ref{Pro_fullpart_benefit} implies that the increment of this benefit is high as average participation rate increases, which indicates that the participation benefit is large. Meanwhile, recall that  $\Delta_i= \ppl{\costthreshold{i}}{i}- \gppl{\costthreshold{i}}{i}$ captures the addition part of the privacy cost that  does not concern data correlation. Condition (ii) in Proposition \ref{pro_fullpart_loss} implies that the increment of this part of the privacy cost is small as the group participation rate increases, which indicates that the privacy cost of participation is closed to that of non-participation. Intuitively, if Conditions (ii) in the above two propositions holds,  the analyst does not need to pay a lot to each agent to incentivize her participation. This is because the benefit of participation is significant (Condition (ii) in Proposition \ref{Pro_fullpart_benefit}) or the negative impact of participation is small (Condition (ii) in Proposition \ref{pro_fullpart_loss}). As a result, the expected payment of each agent is relatively low. This enables the analyst to increase the participation rate and  incentivize more fraction of the population. This finally leads to  a full participation rate to induce zero bias.

\section{Concluding Remarks}

In this paper we study the design of an optimal mechanism for data acquisition in a setting where there is correlation among the data of participants that leads to information leakage.  Information leakage is a crucial and under-explored feature of data markets and our results represent the first to characterize an optimal mechanism for data acquisition in such a setting.  Additionally, our results provide the analyst the ability to optimally trade off bias and variance of the resulting estimator.  

Further, our results provide important perspectives about the consequences of correlation and information leakage for data marketplaces.  In particular, intuitively one might expect that these factors, which underlie the privacy paradox, would lead to inefficiencies where data marketplaces  exploit participants and obtain data for lower prices than otherwise.  The analysis in this paper shows that this is indeed the case.  Combined with the results of \cite{acemoglu_privacy}, there is a compelling argument that information leakage is a significant factor in market inefficiencies that lead to oversharing and reduced payments in data marketplaces.  

Thus, one may ask if it is possible to regulate data marketplaces to avoid such inefficiencies.  Unfortunately, our results highlight that typical suggestions, such as avoiding monopolistic platforms, may not be effective.  Since our results hold for any fixed participation rate, a platform can exploit information leakage regardless of its market size.  Further, differential privacy cannot eliminate the underlying issues in this case, which comes from inter-group and intra-group correlation.  Thus, an important open problem motivated by our work is: what approaches can be brought to bear to mitigate the impact of information leakage and the privacy paradox in data marketplaces?

\bibliographystyle{unsrtnat}
\bibliography{references}  

\appendix


\section{Proof of Theorem \ref{pro_payment}}\label{appendixsec_payment}

First, we show that our mechanism is truthful if and only if the payment rule is of the form 
\begin{equation}\tag{\ref{equ_payment}}
\begin{aligned}
\payment{\reportedCost}{i}=\reportedCost-\ppl{\reportedCost}{i} 
+\frac{1}{\selectionProb{\reportedCost}{i}}\left(\left(1-b(\ratioVec_i;\gCorrelationVec_i)\right)\int_{\reportedCost}^{\cost_{\max}}\selectionProb{z}{i}dz + \constantpart{i}\right),
\end{aligned}
\end{equation}
given by Theorem \ref{pro_payment}, assuming that $\ratioVec$ is the true induced equilibrium participation profile. Second, we show that indeed, the equilibrium participation profile $\ratioVecequ$ is the desired profile $\ratioVec$ when using this payment rule if and only if 
\begin{equation}\tag{\ref{equ_constantpart}}
\constantpart{i} = \ppl{\costthreshold{i}}{i}-g(\costthreshold{i},\ratioVec;\gCorrelationVec_i)-\left(1-b(\ratioVec_i;\gCorrelationVec_i)\right)\int_{\costthreshold{i}}^{\cost_{\max}}\selectionProb{z}{i}dz - \benefit{\ratio}
\end{equation}
where $\costthreshold{i}$ satisfies 
\begin{equation}\tag{\ref{equ_ratio_threshold}}
	 \ratio_{i} = \int_{\Cost_{\min}}^{\costthreshold{i}} \costD_i(\Cost) d\Cost, 1\le i\le I.
\end{equation}


\begin{enumerate}
\item We first argue that our mechanism is (strictly) truthful if and only if the payment rule has the form given in Equation~\eqref{equ_payment} and $A_i$ is (strictly) monotone for all $i$.

First, we show the ``if'' direction. Given a desired participation rate profile $\ratioVec$, we plug the payment function into the agent's utility from participation, and obtain the following expected utility $\expectedUtility{\reportedCost}{\cost}{i}$ when the agent with true cost $\cost$ reports $\reportedCost$:
    \begin{align*}
    \expectedUtility{\reportedCost}{\cost}{i}  & =\selectionProb{\reportedCost}{i}\left[\payment{\reportedCost}{i} -c + \ppl{\cost}{i}\right]- \ppl{\cost}{i}+ \benefit{\aveRatio}
    \\& = \selectionProb{\reportedCost}{i} \left(\reportedCost - \ppl{\reportedCost}{i} \right)
    + \left(1-b(\ratioVec_i;\gCorrelationVec_i)\right)\int_{\reportedCost}^{\cost_{\max}}\selectionProb{z}{i}dz 
    + \selectionProb{\reportedCost}{i}\left[-c + \ppl{\cost}{i}\right] + K
    \end{align*}
    where $K$ does not depend on $\reportedCost$  
    
    Now, note that that we are in the non-atomic setting in which each agent is infinitesimal. In turn, a single agent's cost reporting and participation decisions does not affect the participation ratio, and $\frac{\partial \ratio{i}}{\partial \reportedCost} = 0$. As such, the derivative of $\expectedUtility{\reportedCost}{\cost}{i}$ with respect to $\reportedCost$ is given by 
\begin{align*}
\begin{split}
&A'_i(\reportedCost)(\reportedCost-\ppl{\reportedCost}{i})
+ \selectionProb{\reportedCost}{i}\left(1-b(\ratioVec_i;\gCorrelationVec_i)\right)
- \left(1-b(\ratioVec_i;\gCorrelationVec_i)\right) \selectionProb{\reportedCost}{i}
- A'_i(\reportedCost) (c-\ppl{\cost}{i})
\\&= A'_i(\reportedCost)(\reportedCost-\ppl{\reportedCost}{i}-(c-\ppl{\cost}{i})).
\end{split}
\end{align*}
    In turn, when the allocation rule is decreasing with respect to cost, i.e., $A'_i(c)< 0$, we obtain that
\begin{equation}
\frac{\partial \expectedUtility{\reportedCost}{\cost}{i}}{\partial \reportedCost}=A'_i(\reportedCost)(\reportedCost-\ppl{\reportedCost}{i}-(c-\ppl{\cost}{i}))  \begin{cases} >0, \quad \text{if} \ \ \reportedCost < \cost,\\
=0, \quad \text{if} \ \ \reportedCost = \cost,\\
<0, \quad \text{if} \ \ \reportedCost > \cost.
\end{cases}
\end{equation}
Therefore, the agent maximizes her expected utility of participation if and only if she truthfully reports her cost, i.e., $\reportedCost = \cost$. This corresponds to strict truthfulness. On the other hand, if the allocation rule is non-increasing, i.e., $A'_i(c)\le 0$, we have
\begin{equation}
\frac{\partial \expectedUtility{\reportedCost}{\cost}{i}}{\partial \reportedCost}=A'_i(\reportedCost)(\reportedCost-\ppl{\reportedCost}{i}-(c-\ppl{\cost}{i}))  \begin{cases} \ge 0, \quad \text{if} \ \ \reportedCost \le \cost,\\
\le 0, \quad \text{if} \ \ \reportedCost>\cost.
\end{cases}
\end{equation}
Therefore, the agent maximizes her expected utility of participation if she truthfully reports her cost, i.e., $\reportedCost = \cost$. This corresponds to general truthfulness, which incorporates strict truthfulness as a special case.

We now show that ``only if'' direction: any (strictly) truthful payment function must have the form given in Equation~\eqref{equ_payment} and requires $A_i$ to be (strictly) monotone. By truthfulness, we have $\expectedUtility{\cost}{\cost}{i} = \max_{\reportedCost} \expectedUtility{\reportedCost}{\cost}{i}$. Applying the envelope theorem to $\max_{\reportedCost} \expectedUtility{\reportedCost}{\cost}{i}$ yields
\begin{equation}
\frac{\partial \expectedUtility{\cost}{\cost}{i} }{\partial \cost} = \left.\frac{\partial \expectedUtility{\reportedCost}{\cost}{i} }{\partial \cost}\right|_{\reportedCost = \cost} = -(1-b(\ratioVec_i;\gCorrelationVec_i)) \cdot \selectionProb{\cost}{i} -b(\ratioVec_i;\gCorrelationVec_i).
\end{equation}

Taking the integral from $\cost$ to $\cost_{\max}$, we further have 
\begin{equation}
\expectedUtility{\cost_{\max}}{\cost_{\max}}{i}-\expectedUtility{\cost}{\cost}{i}= (b(\ratioVec_i;\gCorrelationVec_i) - 1)\int_{\cost}^{\max}\selectionProb{z}{i}dz+\ppl{\cost}{i}-\ppl{\cost_{\max}}{i}.   
\end{equation}
Since $\expectedUtility{\cost}{\cost}{i} = \selectionProb{\cost}{i}\left[\payment{\cost}{i} -c + \ppl{\cost}{i}\right] - \ppl{\cost}{i} + \benefit{\aveRatio}$, we obtain the following equation:
\begin{align*}
&\expectedUtility{\cost_{\max}}{\cost_{\max}}{i} + (1-b(\ratioVec_i;\gCorrelationVec_i))\int_{\cost}^{\cost_{\max}}\selectionProb{z}{i}dz+ \ppl{\cost_{\max}}{i}-\ppl{\cost}{i}  
\\&= \selectionProb{\cost}{i}\left[\payment{\cost}{i} -c + \ppl{\cost}{i}\right] - \ppl{\cost}{i} + \benefit{\aveRatio}.
\end{align*}

Hence, we have 
\begin{equation}
\payment{\cost}{i}=\cost-\ppl{\cost}{i} 
+\frac{1}{\selectionProb{\cost}{i}}\left(\left(1-b(\ratioVec_i;\gCorrelationVec_i)\right)\int_{\cost}^{\cost_{\max}}\selectionProb{z}{i}dz + v\right),
\end{equation}
where
\begin{equation}
 v=  \expectedUtility{\cost_{\max}}{\cost_{\max}}{i} + \ppl{\cost_{\max}}{i}-\benefit{\aveRatio}.
\end{equation}
Now, remembering that 
\begin{equation}
\frac{\partial \expectedUtility{\reportedCost}{\cost}{i}}{\partial \reportedCost}=A'_i(\reportedCost)(\reportedCost-\ppl{\reportedCost}{i}-(c-\ppl{\cost}{i})),  
\end{equation}
(strict) truthfulness implies that for all $\cost$, there exists $\epsilon > 0$ (small) such that for any $\reportedCost \in (\cost-\epsilon,\cost)$, $\frac{\partial \expectedUtility{\reportedCost}{\cost}{i}}{\partial \reportedCost} \geq 0$ ($> 0$ for strict) and for any $\reportedCost \in (\cost,\cost+\epsilon)$, $\frac{\partial \expectedUtility{\reportedCost}{\cost}{i}}{\partial \reportedCost} \leq 0$ ($< 0$ for strict)
. Since $c - \ppl{\cost}{i} = c (1 - b(\ratioVec_i;\gCorrelationVec_i))$ is increasing, this in particular requires that $A_i'(\reportedCost) \leq 0$ ($< 0$ for strict truthfulness) on $(\cost-\epsilon,\cost)$ and $(\cost,\cost+\epsilon)$. Since this holds for all $\cost$, this in particular implies that for all $\reportedCost$, $A_i'(\reportedCost) \leq 0$ (resp $< 0$), which shows (strict) monotonicity of $A$.



\item 
It remains to show that our payment rule induces an equilibrium participation profile $\ratioVecequ$ equal to the desired participation profile $\ratioVec$ if and only if 
\begin{equation}\tag{\ref{equ_constantpart}}
\constantpart{i} = \ppl{\costthreshold{i}}{i}-g(\costthreshold{i},\ratioVec;\gCorrelationVec_i)-\left(1-b(\ratioVec_i;\gCorrelationVec_i)\right)\int_{\costthreshold{i}}^{\cost_{\max}}\selectionProb{z}{i}dz - \benefit{\ratio},
\end{equation}
where $\costthreshold{i}$ is defined as the solution to 
\begin{equation}\tag{\ref{equ_ratio_threshold}}
	 \ratio_{i} = \int_{\Cost_{\min}}^{\costthreshold{i}} \costD_i(\Cost) d\Cost, 1\le i\le I.
\end{equation}
I.e., $\costthreshold{i}$ is such that if all (and only the) agents in group $i$ with cost at most $\costthreshold{i}$ participate, then the participation rate in group $i$ is $\ratio_i$. As such, to show the result, we simply need to show that the participating agents in group $i$ are those with cost at most $\costthreshold{i}$ if and only if $v = \constantpart{i}$.
To do so, note that an agent in group $i$ with cost $c$ who truthfully reports his cost has utility
\begin{align*}
\expectedUtility{\cost}{\cost}{i}
& =\selectionProb{\cost}{i}\left[\payment{\cost}{i} -c + \ppl{\cost}{i}\right]- \ppl{\cost}{i}+ \benefit{\aveRatio}
\\& = \left(1-b(\ratioVec_i;\gCorrelationVec_i)\right)\int_{\cost}^{\cost_{\max}}\selectionProb{z}{i}dz + v - \ppl{\cost}{i}+ \benefit{\aveRatio}, 
\end{align*}
and has a utility of $-\gppl{\cost}{i}$ for non-participation. In turn, an agent in group $i$ and with cost $c$ participates if and only 
\[
\left(1-b(\ratioVec_i;\gCorrelationVec_i)\right)\int_{\cost}^{\cost_{\max}}\selectionProb{z}{i}dz + \gppl{\cost}{i} - \ppl{\cost}{i}+ \benefit{\aveRatio} + v \geq 0.
\]
Note that because $w$ is continuous, $\gppl{\cost}{i}-\ppl{\cost}{i}$ is continuous and increasing in $c$, and $A_i(z) \geq 0$ hence $\int_{\cost}^{\cost_{\max}}\selectionProb{z}{i}dz$ is continuous non-increasing in $\cost$, we have that $\expectedUtility{\cost}{\cost}{i}$ is continuous and decreasing in $\cost$. In turn an agent participates exactly when his cost satisfies $c \leq \costthreshold{i}$ if and only if 
\[
\left(1-b(\ratioVec_i;\gCorrelationVec_i)\right)\int_{\costthreshold{i}}^{\cost_{\max}}\selectionProb{z}{i}dz + \gppl{\costthreshold{i}}{i} - \ppl{\costthreshold{i}}{i}+ \benefit{\aveRatio} + v = 0,
\]
which yields 
\[
v = \constantpart{i}.
\]
This concludes the proof.

\end{enumerate}

\section{Proof of Lemma \ref{lem_tradeoff} and Remark \ref{rem_tradeoff}}\label{apendixsec_tradeoff}

We first derive the variance.  Recall that $\mathcal{O}\subset \mathcal{N}$ denotes the set of reporters selected in the participants. For simplicity of notations, we make the conditioning on $\mathcal{N}$ implicit in this proof. Let us fix the participation rate profile $\ratioVecequ$, the associated cost threshold $\costthreshold{i}$ under which agents are willing to participate, and the corresponding participation rate $\ratioequ$ under this profile in the population. Note that there are $N = \agentNum\ratioequ$ participants. The estimator is an average of $N = \agentNum\ratioequ$ i.i.d. random variables, where each variable has variance
\begin{align*}
\sigma^2 
&= \mathbbm{E}\left[\left(\frac{x_k \cdot \mathbbm{1}\{k\in \mathcal{O}\}}{A^k}\right)^2\right]-\mathbbm{E}\left[\frac{x_k \cdot \mathbbm{1}\{k\in \mathcal{O}\}}{A^k}\right]^2.
\end{align*}
Note that 
\begin{align*}
\mathbbm{E}\left[\left(\frac{x_k \cdot \mathbbm{1}\{k\in \mathcal{O}\}}{A^k}\right)^2\right]
& = \mathbbm{E} \left[\left(\frac{x_k}{A^k}\right)^2  \cdot \mathbbm{1}\{k\in \mathcal{O}\}\right]
\\& = \mathbbm{E}_{i,c} \left[ \mathbbm{E}_{x_k,\mathcal{O}}\left[\left(\frac{x_k}{A^k}\right)^2  \cdot \mathbbm{1}\{k\in \mathcal{O}\} \bigg| i,c\right] \right]
\\& = \mathbbm{E}_{i,c} \left[\frac{\mathbbm{E} \left[x_k^2|i,c\right]}{A_i(c)^2}  \cdot \Pr \left[k\in \mathcal{O} | i,c\right]\right]
\\& = \mathbbm{E}_{i,c} \left[\frac{\mathbbm{E} \left[x_k|i,c\right]}{A_i(c)} \right],
\end{align*}
where the second-to-last step uses the independence of $x$ and $\mathcal{O}$ conditional on $i,~c$, and the last step uses the facts that i) conditional on $k$ being a participant, $\Pr \left[k\in \mathcal{O} | i,c\right] = A_i(c)$ and ii) $x_k^2 = x_k$ as $x_k \in \{0,1\}$. When $x_k \in [0,1]$ instead, note that $\mathbbm{E} \left[x_k^2|i,c\right] \leq \mathbbm{E} \left[x_k|i,c\right]$, and we get an inequality instead, which proves Remark \ref{rem_tradeoff}. A similar calculation yields (both in the binary and non-binary cases)
\[
\mathbbm{E}\left[\frac{x_k \cdot \mathbbm{1}\{k\in \mathcal{O}\}}{A^k}\right] = \mathbbm{E}_{i,c} \left[\mathbbm{E} \left[x_k|i,c\right] \right].
\]

Further, using the fact that an agent in group $i$ participates if and only if $c \leq \costthreshold{i}$, note that the probability density of an agent being in group $i$ and having cost $c$ \emph{conditional on the agent participating} is given by 
\begin{equation*}
    \frac{\groupProb{i}f_i(\cost)}{\sum\limits_{1\le i\le I}\groupProb{i}\int_{\cost_{\min}}^{\costthreshold{i}}f_i(\cost)d\cost}~~\text{if}~~\cost_{\min} \le \cost \le \costthreshold{i},
\end{equation*}
and is equal to $0$ if $c > \costthreshold{i}$.  In turn, we obtain that
\begin{align*}
\sigma^2 
 &= \frac{\sum\limits_{1\le i\le I}\groupProb{i}\int_{\cost_{\min}}^{\costthreshold{i}}\frac{\mathbb{E}[x^2|i,\cost]}{A_{i}(\cost)}f_i(\cost)d\cost}{\sum\limits_{1\le i\le I}\groupProb{i}\int_{\cost_{\min}}^{\costthreshold{i}}f_i(\cost)d\cost}-\left(\frac{\sum\limits_{1\le i\le I}\groupProb{i}\int_{\cost_{\min}}^{\costthreshold{i}}\mathbbm{E}[x|i,c]f_i(\cost)d\cost}{\sum\limits_{1\le i\le I}\groupProb{i}\int_{\cost_{\min}}^{\costthreshold{i}}f_i(\cost)d\cost}\right)^2.
\end{align*}

Since $N = \frac{1}{s\ratioequ}$ and that $\ratioequ = \sum\limits_{1\le i\le I}\groupProb{i}\int_{\cost_{\min}}^{\costthreshold{i}}f_i(\cost)d\cost$, the variance of the Horvitz-Thompson estimator is then given by
\begin{equation}\label{equ_worstvariance}
 \frac{1}{s\ratioequ}\left(\frac{\sum\limits_{1\le i\le I}\groupProb{i}\int_{\cost_{\min}}^{\costthreshold{i}}\frac{\mathbb{E}[x^2|\cost, i]}{A_{i}(\cost)}f_i(\cost)d\cost}{\ratioequ}-\left(\frac{\sum\limits_{1\le i\le I}\groupProb{i}\int_{\cost_{\min}}^{\costthreshold{i}}\mathbbm{E}[x|c,i]f_i(\cost)d\cost}{\ratioequ}\right)^2  \right).
\end{equation}

Next, we characterize the worst-case bias (over the data of the non-participants). Recall that we assume positive correlation between data and costs. In turn, the worst-case bias corresponds to the case where the data of all the non-participants (whose costs hence average data point are higher than participants in each group) is one. To see this, let $\mu$ be true population mean. Let $\mathbbm{E}[x|d(c)=0]$, resp. $\mathbbm{E}[x|d(c)=1]$, be the expectations of the data of the non-participants (decision variable $d(c)=0$), resp. participants (decision variable $d(c)=1$). Recall that the estimator $\estimator$ is an unbiased estimator of participants' data, i.e., $\mathbbm{E}[\estimator] = \mathbbm{E}[x|d(c)=1]$. We have 
$$\agentNum\ratioequ\mathbbm{E}[\estimator]+\agentNum(1-\ratioequ) \mathbbm{E}[x|d(c)=0] = \agentNum\mu,$$
which leads to 
$$\mu - \mathbbm{E}[\estimator] = (1 - \ratioequ) (\mathbbm{E}[x|d(c)=0] - \mathbbm{E}[\estimator]).$$
Positive correlation between data and cost indicates that as the agents' costs increase, their data is more likely to increase as well than to decrease. Meanwhile, we show in Corollary \ref{cor_threshold} that  non-participants have higher costs. 
Thus, we know that the expectation of  non-participants' data is no less than that of participants' data, i.e.,

    $$
    \mathbbm{E}[x|d(c)=0] \geq \mathbbm{E}[x|d(c)=1] = \mathbbm{E}[\estimator].
    $$
Therefore, when taking absolute values, we have 
    \[
    \left\vert \mu - \mathbbm{E}[\estimator] \right\vert = (1 - \ratioequ) (\mathbbm{E}[x | d(c)=0] - \mathbbm{E}[\estimator]).
    \]
The bias is maximized when the expectation of non-participants' data is the maximum value, i.e., $\mathbbm{E}[x | d(c)=0] = 1$. Thus, we have  
    \[
    \left\vert \mu - \mathbbm{E}[\estimator] \right\vert = (1 - \ratioequ) (1 - \mathbbm{E}[\estimator]).
    \]
Plugging in (by a similar calculation as for the variance terms),
\[
\mathbbm{E}[\estimator] = \frac{1}{\ratioequ}\sum\limits_{1\le i\le I}\groupProb{i}\int_{\cost_{\min}}^{\costthreshold{i}}p_{i}(\cost)f_i(\cost)d\cost,
\]
we immediately obtain that the worst-case bias (where the worst-case is taken over the data of the non-participants) is 

\begin{equation}\label{equ_worstbias1}
|\mu - \mathbbm{E}[\estimator]| = (1-\ratioequ)\left(1-\frac{1}{\ratioequ} \sum\limits_{1\le i\le I}\groupProb{i}\int_{\cost_{\min}}^{\costthreshold{i}}p_{i}(\cost)f_i(\cost)d\cost\right).
\end{equation}

So far we have derived the variance and the worst-case bias separately. Notice that both variance and worst-case bias are characterized by the distribution of participants data, i.e., $p_i(c)$, for $\cost\le \costthreshold{i}$ and all $i$. Thus, the  supremum in $p_i(c)$ (i.e.,over the participants' data) of the linear combination of the variance and the worst-case bias gives the worst-case linear combination of variance and bias, i.e., the worst-case bias-variance trade-off, over the data of both participants and non-participants.

\section{Proof of Theorem \ref{the_opt_pro}}\label{appendixsec_proof_the_opt}
Let $\ratioVecequ$ be the desired equilibrium participation rate profile. The proof goes as follows. First, we write the problem of finding the optimal allocation rule as a minimax optimization problem over a discretization of the costs. Second, we interpret this optimization problem as a zero-sum game between the analyst who controls the allocation rule and aims to minimize the bias-variance trade-off, and an adversary who controls the correlation between data and costs and aims to maximize the bias-variance objective. Finally, we convert our solution in the discrete cost case to the original, continuous case.

\paragraph{Reformulating the optimization program} 

The optimization problem in Definition~\ref{def_mech_design} involves two constraints, a truthfulness constraint and a budget constraint. Because our payment rule uniquely induces truthfulness and the desired participation rate profile $\ratioVecequ$ by Theorem \ref{pro_payment}, we can directly plug in the closed-form expression for the payment into the budget constraint, and replace the truthfulness constraints by a monotonicity constraint on the selection rule $A_i$, without loss of generality. That is, the expected payment to a single participant is given by 
\begin{align*}
&\sum_i q_i \int_{\cost_{\min}}^{\costthreshold{i}} A_i(c) P_i(c) f_i(c) d\cost 
\\=&  \sum_i q_i \int_{\cost_{\min}}^{\costthreshold{i}} \left( A_i(c) \left(\cost-\pplequ{\cost}{i} \right)+
\left(1-b(\ratioVecequ_i;\gCorrelationVec_i)\right)\int_{\cost}^{\cost_{\max}}\selectionProb{z}{i}dz \right)  f_i(c) d\cost \\& +\sum_i q_i \int_{\cost_{\min}}^{\costthreshold{i}} \constantpartequ{i} f_i(c) d\cost
\\=& \sum_i q_i \int_{\cost_{\min}}^{\costthreshold{i}} A_i(c) \left(\cost-\pplequ{\cost}{i}\right) f_i(\cost) d\cost 
+ \sum_i \groupProb{i}
\left(1-b(\ratioVecequ_i;\gCorrelationVec_i)\right) \int_{\cost_{\min}}^{\costthreshold{i}} \left(\int_{\cost}^{\cost_{\max}} \selectionProb{z}{i}dz \right)f_i(c) d\cost 
\\& +\sum_i q_i \int_{\cost_{\min}}^{\costthreshold{i}} \constantpartequ{i} f_i(c) d\cost.
\end{align*}
First, let us simplify the double-integral term. We remark that 
\begin{align*}
\int_{\cost_{\min}}^{\costthreshold{i}} \left(\int_{\cost}^{\cost_{\max}} \selectionProb{z}{i}dz \right)f_i(c) d\cost 
&= \int_{\cost_{\min}}^{\costthreshold{i}} \left(\int_{\cost_{min}}^{\cost_{\max}} \selectionProb{z}{i}dz \right)f_i(c) d\cost 
- \int_{\cost_{\min}}^{\costthreshold{i}} \left(\int_{\cost_{min}}^{\cost} \selectionProb{z}{i}dz \right)f_i(c) d\cost 
\\& = \int_{\cost_{min}}^{\cost_{\max}} \left(\int_{\cost_{\min}}^{\costthreshold{i}} f_i(\cost) d\cost \right) \selectionProb{z}{i}dz
- \int_{\cost_{min}}^{\costthreshold{i}} \left(\int_{z}^{\costthreshold{i}} f_i(\cost) d\cost \right) \selectionProb{z}{i} dz
\\&= F_i(\costthreshold{i}) \int_{\cost_{min}}^{\cost_{\max}} \selectionProb{z}{i}dz
- F_i(\costthreshold{i}) \int_{\cost_{min}}^{\costthreshold{i}} \selectionProb{z}{i}dz
+ \int_{\cost_{min}}^{\costthreshold{i}} F_i(z) A_i(z) dz
\\& =F_i(\costthreshold{i})\int_{\costthreshold{i}}^{\cost_{\max}}A_i(c)dc+\int_{\cost_{\min}}^{\costthreshold{i}}F_i(\cost)A_i(\cost)d\cost.
\end{align*}

As $\constantpart{i}$ is independent of $\cost$, we also have that  
\begin{align*}
\int_{\cost_{\min}}^{\costthreshold{i}} \constantpartequ{i} f_i(c) d\cost
=&\ \  \constantpartequ{i} \int_{\cost_{\min}}^{\costthreshold{i}}f_i(c) d\cost 
\\=&\ \  F_i(\costthreshold{i}) \left(\pplequ{\costthreshold{i}}{i}-\gpplequ{\costthreshold{i}}{i} - \benefit{\ratioequ} \right)\\
& \ \  -F_i(\costthreshold{i})\left(\left(1-b(\ratioVecequ_i;\gCorrelationVec_i)\right)\int_{\costthreshold{i}}^{\cost_{\max}}\selectionProb{z}{i}dz\right). 
\end{align*}
Thus, 
\begin{align*}
&\sum_i q_i \int_{\cost_{\min}}^{\costthreshold{i}} A_i(c) P_i(c) f_i(c) d\cost
\\=& \sum_i q_i \int_{\cost_{\min}}^{\costthreshold{i}} A_i(c) \left(\cost-\pplequ{\cost}{i}\right) f_i(\cost) d\cost 
+ \sum_i q_i \left(1-b(\ratioVecequ_i;\gCorrelationVec_i)\right) \int_{\cost_{\min}}^{\costthreshold{i}}F_i(\cost)A_i(\cost)d\cost
\\&+ \sum_i q_i  \left(1-b(\ratioVecequ_i;\gCorrelationVec_i)\right) F_i(\costthreshold{i})\int_{\costthreshold{i}}^{\cost_{\max}}A_i(c)dc
- \sum_i q_i F_i(\costthreshold{i})\left(\left(1-b(\ratioVecequ_i;\gCorrelationVec_i)\right)\int_{\costthreshold{i}}^{\cost_{\max}}\selectionProb{z}{i}dz\right)
\\& + \sum_i q_i F_i(\costthreshold{i}) \left(\pplequ{\costthreshold{i}}{i}-\gpplequ{\costthreshold{i}}{i} - \benefit{\ratioequ} \right)
\\=&  \sum_i q_i \int_{\cost_{\min}}^{\costthreshold{i}} A_i(c) \left(\cost-\pplequ{\cost}{i} + \left(1-b(\ratioVecequ_i;\gCorrelationVec_i)\right) \frac{F_i(\cost)}{f_i(\cost)} \right) f_i(\cost) d\cost
\\&+\sum_i q_i F_i(\costthreshold{i})\left(\pplequ{\costthreshold{i}}{i}-\gpplequ{\costthreshold{i}}{i} - \benefit{\ratioequ} \right).
\end{align*}

Remembering that by definition $F_i(\costthreshold{i})= \ratioequ_i$, and that the virtual costs are given by 
 \begin{equation*}
 \phi_i(\cost;\ratioVecequ) = \cost-\pplequ{\cost}{i} + (1-b(\ratioVecequ_i;\gCorrelationVec_i))\frac{F_i(c)}{f_i(c)},
 \end{equation*}
we can rewrite 
\begin{align*}
&\sum_i q_i \int_{\cost_{\min}}^{\costthreshold{i}} A_i(c) P_i(c) f_i(c) d\cost
\\=& \sum_{i}\groupProb{i}\int_{\cost_{\min}}^{\costthreshold{i}}\phi_{i}(\cost;\ratioVecequ)\selectionProb{c}{i} \costD_i(\cost)d\cost+\sum_{i}\groupProb{i} \ratioequ_i(\pplequ{\costthreshold{i}}{i}-\gpplequ{\costthreshold{i}}{i}-w(\ratioequ)).
\end{align*}



Therefore, the equivalent optimization program is given by:
\begin{align}\label{equ_equivalentproblem}
\begin{split}
\min_{A, \ratioVecequ} \ \ &T(A,\ratioVecequ)\\
s.t.\ \  &\agentNum \left(\sum_{i}\groupProb{i}\int_{\cost_{\min}}^{\costthreshold{i}}\phi_{i}(\cost;\ratioVecequ)\selectionProb{c}{i} \costD_i(\cost)d\cost+\sum_{i}\groupProb{i} \ratioequ_i(\pplequ{\costthreshold{i}}{i}-\gpplequ{\costthreshold{i}}{i}-w(\ratio))\right) \le B,\\
&  \selectionProb{c}{i} \in [0,1]\ \text{is a non-increasing function}~\forall i \in [I],\\
& 0< \ratioequg_i\le 1~\forall i \in [I].
\end{split}
\end{align}

Now, we fix the equilibrium participation profile $\ratioVecequ$. We focus on finding the optimal selection rule given the participation profile. This is given by the following optimization program, plugging back the expression for $T(A,\ratioVecequ)$:
\begin{equation}
\begin{aligned}
\min_{A_i(c)}  \max_{p_i(c)}~~&\frac{\weight}{\agentNum\left(\ratioequ\right)^2}\left(\sum_{i} \groupProb{i}\int_{\cost_{\min}}^{\costthreshold{i}}\frac{p_i(c)}{\selectionProb{c}{i}}f_i(\cost)d\cost-\frac{1}{\ratioequ}\left(\sum_{i} \groupProb{i}\int_{\cost_{\min}}^{\costthreshold{i}}p_{i}(\cost)f_i(\cost)d\cost\right)^2\right) \\ &+(1-\weight)(1-\ratioequ)  \left(1-\frac{1}{\ratioequ}\sum_{i} \groupProb{i}\int_{\cost_{\min}}^{\costthreshold{i}}p_{i}(\cost)f_i(\cost)d\cost\right) \\
s.t.~~&\sum_{i}\groupProb{i}\int_{\cost_{\min}}^{\costthreshold{i}}\phi_{i}(\cost;\ratioVecequ)\selectionProb{c}{i}\costD_i(\cost)d\cost+\sum_{i} \groupProb{i}\ratio_i(\pplequ{\costthreshold{i}}{i}-\gpplequ{\costthreshold{i}}{i}-w(\ratioequ)) \le \frac{B}{\agentNum},
\\& 0 \leq A_i(c), \leq 1~\forall i,c,
\\& 0 \leq p_i(c), \leq 1~\forall i,c.
\end{aligned}
\end{equation}
Note that we do not require the monotonicity constraints on $A_i$ and $p_i$. We will later show that the solutions $A_i^*$ and $p_i$ to this relaxed optimization problem are indeed monotone, hence relaxing the constraint is without loss of generality. 

\paragraph{Solving for the discrete cost case} 
The above formulation is presented for continuous costs. Before we find the solution to the continuous cost problem, let us first focus on the case of discrete costs and find the corresponding optimization program and solution. We will later show how to transform the discrete solution into an optimal solution to the continuous optimization problem above. 

We first introduce the notations in the discrete case. Suppose the cost $\cost$ in the population is from a discrete and finite set $\{c_1,c_2,...,c_J\}$ with size $J$. Recall that there are $I$ groups of agents parameterized by data correlation strength. We write $\pi_{ij}$ as the probability of an agent belonging to group $i$ and having cost $c_j$. Naturally we have $\sum_{1\le i \le I,1\le j\le J} \pi_{ij}= 1$. Recall that agents' decisions have threshold structure, i.e. only agents with costs below the cost threshold choose to participate. We use $t(i)$ to denote the index of threshold cost in group $i$. That is,  if the cost $\cost\le c_{t(i)}$ for an agent in group $i$, he would like to participate.  Let $\phi_{ij}$ be the virtual cost of an agent with cost $c_j$ in group $i$ as follows
\begin{equation}
\phi_{ij}=  \begin{cases}c_1-\ppl{c_1}{i},  &\ \ \text{if} \ \ j = 1,\\ c_j -\ppl{c_j}{i}+(1-b(\ratioVec_i;\gCorrelationVec_i))(c_j-c_{j-1})\frac{\sum_{t=1}^{j-1}\pi_{it}}{\pi_{ij}}, &\ \ \text{if} \ \ j> 1.
\end{cases}
\end{equation}
Since $c_1-\ppl{c_1}{i}>0$, $\phi_{ij}>0$. Recall that $p_{i}(\cost) = Pr[x = 1|c,\correlationVec_i]$ in continuous case is the probability of the data being one. We use $p_{ij}$ to denote the probability of the data being one for an agent in group $i$ having cost $c_j$. We are trying to find the optimal selection probability for discrete cost in each group. Let $A_{ij}$ be the selection probability for agents with cost $c_j$ in group $i$. The discrete version of the the min-max optimization problem is as follows:

\begin{equation}\label{equ_minmax_pro}
\begin{aligned}
\min_{A}  \max_{p}\ \  &\frac{\weight}{\agentNum\left(\ratioequ\right)^2}\left(\sum_{1\le i\le I, 1\le j\le t(i)}\pi_{ij}\cdot\frac{p_{ij}}{A_{ij}}-\frac{1}{\ratioequ}\left(\sum_{1\le i\le I, 1\le j\le t(i)}\pi_{ij} p_{ij}\right)^2\right) 
\\ &+(1-\weight)(1-\ratioequ)\left(1-\frac{1}{\ratioequ}\sum_{1\le i\le I, 1\le j\le t(i)}\pi_{ij} p_{ij}\right) \\
s.t.\ \  &\sum_{1\le i\le I, 1\le j\le t(i)}\pi_{ij}\phi_{ij}A_{ij}+\sum_{i} q_i \ratio_i(\pplequ{c_{t(i)}}{i}-\gpplequ{c_{t(i)}}{i}-w(\ratioequ)) \le \frac{B}{\agentNum},\\
&0\le A_{ij}\le 1,~\forall i \in [I],~j \in [t(i)],\\
&0\le p_{ij}\le 1,~\forall i \in [I],~j \in [t(i)].
\end{aligned}    
\end{equation}

We write $A=[A_{ij}]_{1\le i\le I, 1\le j\le t(i)}$ and $p=[p_{ij}]_{1\le i\le I, 1\le j\le t(i)}$. Here, there is an implicit constraint on the mechanism to make the solution meaningful: \begin{equation}
    \frac{B}{\agentNum}-\sum_{i} q_i \ratio_i(\pplequ{c_{t(i)}}{i}-\gpplequ{c_{t(i)}}{i}-w(\ratioequ))>0.
\end{equation}
If it does not hold, otherwise, it means  $\sum_{i,  j}\pi_{ij}\phi_{ij}A_{ij}\le 0$  and  $A_{ij}=0, \forall i,j$ is the solution. Recall that $\phi_{ij}>0, \forall i,j$. In this case, no data is collected, which is meaningless. 
The constraint means the budget should be higher than a threshold related to participation rate $\ratioVecequ$ so as to generate positive selection probability. Meanwhile, without loss of generality, we assume 
\begin{equation}
    \sum_{1\le i\le I, 1\le j\le t(i)}\pi_{ij}\phi_{ij}+\sum_{i} q_i \ratio_i(\pplequ{c_{t(i)}}{i}-\gpplequ{c_{t(i)}}{i}-w(\ratioequ)) > \frac{B}{\agentNum}.
\end{equation}
to avoid trivial solution of $A_{ij}=1, \forall i,j$. If the above inequality does not hold, it is optimal to select all agents with probability one. This assumption means the analyst would not aggressively select all the participants with abundant budget.

We denote the objective function inside minimax problem as $U(A,p)$ for simplification. Notice that $U(A,p)$ is a convex in $A$ and concave in $p$. To see this, let $\pi=[\pi_{ij}]_{1\le i\le I, 1\le j\le t(i)}$. We can write $U(A,p)$ as 
$$
U(A,p)
=\frac{\weight}{\agentNum\left(\ratioequ\right)^2}\left(\langle\pi, p./ A\rangle-\frac{1}{\ratioequ}\langle\pi, p\rangle^2\right)
+(1-\weight)(1-\ratioequ)  \left(1-\frac{1}{\ratioequ}\langle\pi, p\rangle \right).
$$ 
Here with abuse of notation we use $./$ to denote component-wise product (division) operation for vectors, and use $\langle,\rangle$ to denote inner product operation. Then we can obviously figure out $U(A,p)$ is convex in $A$ and concave in $p$, similarly for each problem. In turn, the optimization program defines a convex-concave zero-sum game, in which player 1 (the analyst) is choosing $A$ to minimize $U(A,p)$ given $p$ and player 2 (an adversary) is choosing $q$ to maximize $U(A,p)$ given $A$. The solution of minimax problem corresponds to the equilibrium $(A^*,p^*)$ of the zero-sum game  satisfying the constraint $U(A^*,p^*)=\min\limits_{A} U(A,p^*)=\max\limits_{p} U(A^*,p)$. Thus, we can find the equilibrium of the game to derive the optimal allocation rule.


Next, we find the equilibrium by characterizing the best responses of two players.
\begin{Lem}\label{lem_q_br}
Given $A$, $p_{ij}$ for each $i,j$ is the best response of maximizing player  if and only if one of the following holds: 
\begin{enumerate}
    \item $p_{ij}=1$ and $\frac{\weight}{\agentNum\left(\ratioequ\right)^2}\left(\frac{1}{A_{ij}}-\frac{2}{\ratioequ}\left(\sum_{i,j}\pi_{ij}p_{ij}\right)\right)-\frac{1}{\ratioequ}(1-\ratioequ)(1-\weight)>0$;
    \item $p_{ij}=0$ and $\frac{\weight}{\agentNum\left(\ratioequ\right)^2}\left(\frac{1}{A_{ij}}-\frac{2}{\ratioequ}\left(\sum_{i,j}\pi_{ij}p_{ij}\right)\right)-\frac{1}{\ratioequ}(1-\ratioequ)(1-\weight)<0$;
    \item $0 \leq p_{ij} \leq 1$ and $\frac{\weight}{\agentNum\left(\ratioequ\right)^2}\left(\frac{1}{A_{ij}}-\frac{2}{\ratioequ}\left(\sum_{i,j}\pi_{ij}p_{ij}\right)\right)-\frac{1}{\ratioequ}(1-\ratioequ)(1-\weight)=0$.
\end{enumerate}

\end{Lem}

\begin{proof}
Since the objective function of the maximizer is convex and differentiable, the KKT conditions are necessary and sufficient for optimality, and the optimal solutions are such that there exist dual variables under which the KKT conditions are satisfied. Note that the Lagrangian of the maximization problem is as follows: 
\begin{equation}
L(p,\lambda^+_{ij},\lambda^-_{ij},\lambda)=U(A,p)+\sum_{i,j}\lambda^+_{ij}(1-p_{ij})+\sum_{i,j}\lambda^-_{ij}p_{ij}.
\end{equation}

Hence, the KKT conditions yield 
\begin{equation}
    \frac{\partial L}{\partial p_{ij}} =\frac{\weight\pi_{ij}}{\agentNum\left(\ratioequ\right)^2}\left(\frac{1}{A_{ij}}-\frac{2}{\ratioequ}\left(\sum_{i,j}\pi_{ij}p_{ij}\right)\right)-\frac{1}{\ratioequ}(1-\ratioequ)(1-\weight)-\lambda^+_{ij}+\lambda^-_{ij}=0
\end{equation}

\begin{equation}
    \lambda^+_{ij}(1-p_{ij})=0, \ \ \lambda^-_{ij}p_{ij}=0, \ \  \lambda^+_{ij}\ge 0,  \ \ \lambda^-_{ij}\ge 0, \ \  \lambda\ge 0.
\end{equation}

\begin{itemize}
    \item If $\frac{\weight}{\agentNum\left(\ratioequ\right)^2}\left(\frac{1}{A_{ij}}-\frac{2}{\ratioequ}\left(\sum_{i,j}\pi_{ij}p_{ij}\right)\right)-\frac{1}{\ratioequ}(1-\ratioequ)(1-\weight)>0$, then $\lambda^+_{ij}> 0$, hence a best response must satisfy $p_{ij} = 1$ by complementary slackness. Further, when taking $\lambda^+_{ij}> 0$, $\lambda^-_{ij} = 0$, and $p_{ij}=1$, the KKT conditions hold; hence, $p_{ij} = 1$ is indeed a (the unique) best response. 
\item  If $\frac{\weight}{\agentNum\left(\ratioequ\right)^2}\left(\frac{1}{A_{ij}}-\frac{2}{\ratioequ}\left(\sum_{i,j}\pi_{ij}p_{ij}\right)\right)-\frac{1}{\ratioequ}(1-\ratioequ)(1-\weight)<0$, it must be that $\lambda^-_{ij} > 0$, which in turns implies $p_{ij} = 0$. Further, the KKT conditions hold with $\lambda^-_{ij}> 0$, $\lambda^+_{ij} = 0$, $p_{ij}=0$, hence $p_{ij} = 0$ is indeed a (the unique) best response. 
\item  Finally, if $\frac{\weight}{\agentNum\left(\ratioequ\right)^2}\left(\frac{1}{A_{ij}}-\frac{2}{\ratioequ}\left(\sum_{i,j}\pi_{ij}p_{ij}\right)\right)-\frac{1}{\ratioequ}(1-\ratioequ)(1-\weight)=0$, the KKT conditions hold so long as $0 \leq p_{ij} \leq 1$ with $\lambda^-_{ij} = \lambda^+_{ij} = 0$. Therefore, any $p_{ij} \in [0,1]$ with $$\frac{\weight}{\agentNum\left(\ratioequ\right)^2}\left(\frac{1}{A_{ij}}-\frac{2}{\ratioequ}(\sum_{i,j}\pi_{ij}p_{ij})\right) - \frac{1}{\ratioequ}(1-\ratioequ)(1-\weight) = 0.$$ is a best response for the maximizing player.
\end{itemize}
\end{proof}

Next, the best response of the minimizing player (the analyst) is given by the following lemma: 
\begin{Lem}\label{lem: min_best_resp}
Given $p$, 
\begin{itemize}
    \item if $p_{ij}= 0, \forall i,j $, any $A_{ij}\in [0,1]$ that satisfies budget constraint is best response of the minimizing player;
    \item if $p_{ij}\neq 0, \forall i,j $, the best response of the minimizing player is 
\begin{equation}
    A^*_{ij} = \min\left\{1,\sqrt{\frac{\weight p_{ij}}{\agentNum \left(\ratioequ\right)^2 \lambda^* \phi_{ij}}}\right\},
\end{equation}
where $\lambda^*$ is such that 
\begin{equation}
\sum_{i,j: p_{ij}\neq 0 }\pi_{ij}\phi_{ij}\cdot  \min\left\{1,\sqrt{\frac{\weight p_{ij}}{\agentNum \left(\ratioequ\right)^2 \lambda^* \phi_{ij}}}\right\}+\sum_{i}\groupProb{i} \ratio_i(\pplequ{c_{t(i)}}{i}-\gpplequ{c_{t(i)}}{i}-w(\ratio)) = B/\agentNum.
\end{equation}

\end{itemize}

\end{Lem}

\begin{proof}
If $p_{ij}= 0, \forall i,j $, as $A$ does not appear in the objective function.  Thus, any $A_{ij}\in [0,1]$ that satisfies budget constraint is best response of the minimizing player.

Next, we focus on the case of  $p_{ij}\neq 0, \forall i,j $. We drop the constraint that $A_{ij} \geq 0$ for all $i,j$ (we will see that this is without loss of generality, as we will recover a positive solution). The Lagrangian of the minimization problem is as follows: 
\begin{equation}
\begin{aligned}
L(A,\lambda,\lambda_{ij}) = & \ \ U(A,p)+\sum_{i,j}\lambda_{ij}(A_{ij}-1)\\&+\lambda\cdot \left( \sum_{i,j}\pi_{ij}\phi_{ij}A_{ij}+\sum_{i} \groupProb{i} \ratio_i(\pplequ{c_{t(i)}}{i}-\gpplequ{c_{t(i)}}{i}-w(\ratioequ)) - B/\agentNum\right).
\end{aligned}
\end{equation}
From the KKT condition, we have the optimal $A^*_{ij}$ and optimal dual $\lambda^*$, $\lambda^*_{ij}$ must satisfy
\begin{equation}\label{equ_kkt1}
    \frac{\partial L}{ \partial A_{ij}} = -\frac{\weight}{\agentNum\left(\ratioequ\right)^2}\cdot\frac{\pi_{ij}p_{ij}}{A^{*2}_{ij}}+\lambda^*\pi_{ij}\phi_{ij}+\lambda^*_{ij}=0;
\end{equation}
\begin{equation}
    \lambda^*_{ij}(A^*_{ij}-1)=0,
\end{equation}
\begin{equation}
\lambda^*\cdot \left( \sum_{i,j}\pi_{ij}\phi_{ij}A_{ij}+\sum_{i} \groupProb{i} \ratio_i(\pplequ{c_{t(i)}}{i}-\gpplequ{c_{t(i)}}{i}-w(\ratioequ)) - B/\agentNum\right)=0.
\end{equation}
Thus, we have i) $A^*_{ij}=1$; or 
ii) $A^*_{ij}<1$ and  $ \lambda^*_{ij}=0$. In the second case, we obtain that $A^*_{ij} = \sqrt{\frac{\weight p_{ij}}{\agentNum \left(\ratioequ\right)^2 \lambda^* \phi_{ij}}}$  for some $\lambda^*>0$, according to (\ref{equ_kkt1}). A higher value of selection probability (which indicates higher budget) would always reduce variance and thus, the objective function.   Thus, the optimal allocation rule is such that the budget constraint is binding. 
In conclusion, we have $ A^*_{ij} = \min\left\{1,\sqrt{\frac{\weight p_{ij}}{\agentNum \left(\ratioequ\right)^2 \lambda^* \phi_{ij}}}\right\}$ where $\lambda^*$ is such that the budget constraint is binding. 

\end{proof}

Before we present the intersection of both players' best responses, we change  some  indexes for simplicity. We denote the set of virtual cost of all groups, $\Phi \triangleq \{\phi_{ij}: 1\le i\le I, 1\le j \le t(i)\}$. Suppose there are $K\triangleq \sum_{1\le i\le I}t(i)$ number of elements in the set. We sort them in an non-decreasing order, and replace index $ij$ with index $k$, which indicates the corresponding position in the sorted set, i.e., $\Phi = \{\phi_k:1\le k\le K, \phi_1<\phi_2<...<\phi_K\}$. Similarly, we replace index $ij$ of $A_{ij}$, $p_{ij}$, $\pi_{ij}$ with index $k$, and obtain   $A_k$, $p_{k}$, $\pi_{k}$, which are associated with  virtual cost $\phi_k$. For simplicity, we define $l$ as follows
\begin{equation}
 l\triangleq \sum_{i} q_i \ratioequg_i(\pplequ{\costthreshold{i}}{i}-\gpplequ{\costthreshold{i}}{i}-w(\ratioequ)).
\end{equation}
We begin with some necessary notations as follows. We define 
    \begin{equation}\label{equ_Q}
        Q(m,z) \triangleq \sum\limits_{k=1 }^{m}\pi_{k}\phi_{k}+\sqrt{\frac{\phi_m}{z}}\cdot\sum\limits_{k=m+1}^{K}\pi_{k}\sqrt{\phi_k}, m=1,...,K.
    \end{equation}
    \begin{equation}\label{equ_R}
        R(m,z) \triangleq 2\weight\left(\frac{z}{\phi_{m}} \cdot \sum\limits_{k=1 }^m\pi_{k}\phi_{k}+\sum\limits_{k= m+1}^K\pi_{k}\right)+\left(\ratioequ\right)^2(1-\ratioequ)(1-\weight)\agentNum, m=1,...,K.
    \end{equation}
Notice that $R(m,z)>0$. 

\begin{Cla}\label{cla_monotonic}
$Q(m,1)$ is increasing in $m$,  $R(m,1)$  is  decreasing in $m$ and thus, $\frac{Q(m,1)}{R(m,1)}$  is increasing in $m$. 
\end{Cla}

\begin{proof}
We can see that 
    
    \begin{align*} 
     Q(m+1,1)-Q(m,1) 
     &= \sum_{k=1}^{m+1} \pi_k \phi_k - \sum_{k=1}^m \pi_k \phi_k 
     + \sum_{k=m+2}^K \pi_{k}\sqrt{\phi_{k}\phi_{m+1}} - \sum_{k=m+1}^K \pi_{k}\sqrt{\phi_{k}\phi_{m}}
     \\&= \pi_{m+1} \phi_{m+1} + \sum_{k=m+2}^K \pi_{k}\sqrt{\phi_{k}} \left(\sqrt{\phi_{m+1}} - \sqrt{\phi_m}\right) - \pi_{m+1} \sqrt{\phi_m \phi_{m+1}}
     \\&= \sum_{k=m+2}^K \pi_{k}\sqrt{\phi_{k}} \left(\sqrt{\phi_{m+1}} - \sqrt{\phi_m}\right) + \pi_{m+1} \sqrt{\phi_{m+1}} \left(\sqrt{\phi_{m+1}} - \sqrt{\phi_{m}}\right)>0.
    \end{align*}
    The inequality is due to increasing virtual cost, i.e., $\phi_k$ is increasing in $k$. And we have 
    
    \begin{align*}
     R(m+1,1)-R(m,1)
     &=2\weight\left(\sum\limits_{k=1 }^{m+1}\pi_{k}\phi_{k}\frac{1}{\phi_{m+1}}-\sum\limits_{k=1 }^m\pi_{k}\phi_{k}\frac{1}{\phi_{m}}+\sum\limits_{k= m+2}^K\pi_{k}-\sum\limits_{k= m+1}^K\pi_{k}\right)\\
     &=2\weight\left(\sum\limits_{k=1 }^{m}\pi_{k}\phi_{k}\left(\frac{1}{\phi_{m+1}}-\frac{1}{\phi_{m}}\right)+\pi_{m+1}\phi_{m+1}\frac{1}{\phi_{m+1}}-\pi_{m+1}\right)\\
     &= 2\weight\sum\limits_{k=1}^m\pi_k\phi_k \left(\frac{1}{\phi_{m+1}}-\frac{1}{\phi_{m}}\right)<0
    \end{align*}
    The inequality is also due to increasing virtual cost. Thus $Q(m,1)$ is increasing in $m$, and $R(m,1)$  is  decreasing in $m$. Notice that $R(m,1)>0$. As such, we have that $\frac{Q(m,1)}{R(m,1)}$  is increasing in $m$.  

\end{proof}

\begin{Cla}\label{cla_equality}
For $m=1,...,K-1$,
\begin{equation}
Q(m+1,1)= Q\left(m,\frac{\phi_m}{\phi_{m+1}}\right), \ \ R(m+1,1)= R\left(m,\frac{\phi_m}{\phi_{m+1}}\right).
\end{equation}
\end{Cla}

\begin{proof}
\begin{align*}
    Q(m+1,1)&= \sum\limits_{k=1 }^{m+1}\pi_{k}\phi_{k}+\sqrt{\phi_{m+1}}\cdot\sum\limits_{k=m+2}^{K}\pi_{k}\sqrt{\phi_k}\\
    &=\sum_{k=1}^{m}\pi_{k}\phi_{k}+ \pi_{m+1}\phi_{m+1}+\sqrt{\phi_{m+1}}\cdot\sum\limits_{k=m+2}^{K}\pi_{k}\sqrt{\phi_k} \\
    &= \sum_{k=1}^{m}\pi_{k}\phi_{k}+ \sqrt{\phi_{m+1}}\cdot \pi_{m+1}\sqrt{\phi_{m+1}}+\sqrt{\phi_{m+1}}\cdot\sum\limits_{k=m+2}^{K}\pi_{k}\sqrt{\phi_k}\\
    & = \sum_{k=1}^{m}\pi_{k}\phi_{k}+\sqrt{\phi_{m+1}}\cdot\sum\limits_{k=m+1}^{K}\pi_{k}\sqrt{\phi_k}\\
    & = \sum_{k=1}^{m}\pi_{k}\phi_{k}+\sqrt{\phi_{m}}\cdot \sqrt{\frac{\phi_{m+1}}{\phi_m}}\cdot\sum\limits_{k=m+1}^{K}\pi_{k}\sqrt{\phi_k}\\
    &=Q\left(m,\frac{\phi_m}{\phi_{m+1}}\right).
\end{align*}

To prove $R(m+1,1)= R\left(m,\frac{\phi_m}{\phi_{m+1}}\right)$, it suffices to show $$\frac{1}{\phi_{m+1}} \cdot \sum\limits_{k=1 }^{m+1}\pi_{k}\phi_{k}+\sum\limits_{k= m+2}^K\pi_{k}= \frac{\phi_m}{\phi_{m+1}} \cdot\frac{1}{\phi_{m}} \cdot \sum\limits_{k=1 }^{m}\pi_{k}\phi_{k}+\sum\limits_{k= m+1}^K\pi_{k}.$$ To this end, we can check
\begin{align*}
    \frac{1}{\phi_{m+1}} \cdot \sum\limits_{k=1 }^{m+1}\pi_{k}\phi_{k}+\sum\limits_{k= m+2}^K\pi_{k}&=\frac{1}{\phi_{m+1}} \cdot \sum\limits_{k=1 }^{m}\pi_{k}\phi_{k}+ \frac{1}{\phi_{m+1}}\cdot\pi_{m+1}\phi_{m+1} +\sum\limits_{k= m+2}^K\pi_{k}\\
    &=\frac{1}{\phi_{m+1}} \cdot \sum\limits_{k=1 }^{m}\pi_{k}\phi_{k}+ \pi_{m+1} +\sum\limits_{k= m+2}^K\pi_{k}\\
    &=\frac{\phi_m}{\phi_{m+1}} \cdot\frac{1}{\phi_{m}} \cdot \sum\limits_{k=1 }^{m}\pi_{k}\phi_{k}+\sum\limits_{k= m+1}^K\pi_{k}.
\end{align*}
So we have $R(m+1,1)= R\left(m,\frac{\phi_m}{\phi_{m+1}}\right)$.

\end{proof}

Next, we characterize the intersection of both players' best responses. The minimizing player's strategy corresponds to the solution of the optimization problem in \eqref{equ_minmax_pro}. We will show that the minimizing player's strategy has the following form:
\begin{equation}\label{equ_solution_discrete}
        A_{k}=\begin{cases}\chi, & \text{if} \ \  k \le \hat{k},\\
    \frac{1}{\sqrt{\phi_{k}}} \cdot\frac{B/s-l-\chi\sum\limits_{k=1}^{\hat{k}} \pi_{k}\phi_{k}}{\sum\limits_{k= \hat{k}+1}^K\pi_{k}\sqrt{\phi_{k}}} , & \text{if} \ \  k > \hat{k}.
    \end{cases}
\end{equation}
Here, the constants $\chi$ and $\hat{k}$ are defined as follows:
\begin{itemize}
    \item If $\frac{B/s-l}{\weight\ratioequ}<\frac{Q(1,1)}{R(1,1)}$, then $\chi=0$, $\hat{k}=0$.
    \item If $\frac{Q(1,1)}{R(1,1)}\le\frac{B/s-l}{\weight\ratioequ}< \frac{Q(K,1)}{R(K,1)}$, then $m^*\in\{1,...,K-1\}$ and $z^*\in(0,1]$ be such that  $\frac{Q(m^*,z^*)}{R(m^*,z^*)}= \frac{B/s-l}{\weight\ratioequ}$ (we prove its existence later). 
    \begin{itemize}
        \item If $\frac{B/s-l}{Q(m^*,z^*)}\le 1$, then $\chi = \frac{B/s-l}{Q(m^*,z^*)}$ and $\hat{k}=m^*$.
        \item If $\frac{B/s-l}{Q(m^*,z^*)}> 1$, then $\chi = 1$ and $\hat{k}=\max\{k:Q(k,1)<B/s-l\}.$. 
    \end{itemize}
    \item If $\frac{B/s-l}{\weight\ratioequ}\ge \frac{Q(K,1)}{R(K,1)}$, then $\chi = \frac{B/s-l}{\sum_{k=1}^K\phi_k}$ and $\hat{k}=K$. 
\end{itemize}

Now we begin to present how to obtain this solution by deriving the intersection of both players' best responses.
\begin{itemize}

    \item Case 1: $\frac{B/s-l}{\weight\ratioequ}<\frac{Q(1,1)}{R(1,1)}$, i.e., 
    \begin{equation}
        (B/s-l)(2\weight+\ratioequ(1-\ratioequ)(1-\weight)\agentNum)<\weight\sqrt{\phi_{1}} \cdot \sum\limits_{k=1}^K\pi_{k}\sqrt{\phi_{k}}.
    \end{equation}
    Notice that $\sum_{k=1}^K \pi_k = \ratioequ$. Then,
    

    \begin{equation}
     A_{k} = \frac{1}{\sqrt{\phi_{k}}} \cdot\frac{B/s-l}{\sum\limits_{1\le k\le K}\pi_{k}\sqrt{\phi_{k}}}, \ \ \forall k,
     \end{equation}
     and $p_{k}=1, \forall k$ constitute the mutual best response. To see this, on the one hand, $A_k$ is a best response to $p_k$. Let 
    \[
        \lambda^* = \frac{\weight\left(\sum\limits_{k=1}^{K}\pi_k\sqrt{\phi_k}\right)^2}{\agentNum(\ratioequ)^2\left(\budget/\agentNum-l\right)^2}. 
    \] 
    Then, we have
    \begin{align*}
        \sqrt{\frac{\weight p_k}{\agentNum \left(\ratioequ\right)^2 \lambda^* \phi_k}}=&\sqrt{\frac{\weight}{\agentNum(\ratioequ)^2}\times \frac{\agentNum(\ratioequ)^2\left(\budget/\agentNum-l\right)^2}{\weight\left(\sum\limits_{k=1}^{K}\pi_k\sqrt{\phi_k}\right)^2}}\times \sqrt{\frac{1}{\phi_k}}=\frac{1}{\sqrt{\phi_{k}}} \cdot\frac{B/s-l}{\sum\limits_{k=1}^K\pi_{k}\sqrt{\phi_{k}}}=A_k.
        \end{align*}
    Meanwhile,
    \[
    A_k=\frac{1}{\sqrt{\phi_{k}}} \cdot\frac{B/s-l}{\sum\limits_{k=1}^K\pi_{k}\sqrt{\phi_{k}}}\le \frac{B/s-l}{\sqrt{\phi_{1}}\sum\limits_{k=1}^K\pi_{k}\sqrt{\phi_{k}}}<\frac{\weight}{2\weight+\ratioequ(1-\ratioequ)(1-\weight)\agentNum}<\frac{1}{2}<1.
    \]
    The budget constraint is binding, as 
    \[
    \sum\limits_{k=1}^K \pi_k \phi_k A_k
    =\sum\limits_{k=1}^K \pi_k \phi_k \frac{1}{\sqrt{\phi_{k}}} \cdot\frac{B/s-l}{\sum\limits_{k=1}^K\pi_{k}\sqrt{\phi_{k}}}
    =\frac{B/s-l}{\sum\limits_{k=1}^K\pi_{k}\sqrt{\phi_{k}}}\cdot\sum\limits_{k=1}^K \pi_k \sqrt{\phi_k}
    =B/s-l.
    \]
    Thus, we have $A_k = \min\left\{1,\sqrt{\frac{\weight p_{k}}{\agentNum \left(\ratioequ\right)^2 \lambda^* \phi_{k}}}\right\}$, for all $k$. So $A$ is indeed a best response to $p$, as per Lemma \ref{lem: min_best_resp}.   And it is easy to check that $A_k$ is decreasing in $k$ as virtual cost $\phi_k$ is increasing.
    
    On the other hand, $q$ is a best response to $A$ in P2. We have for all $k$, $p_k=1$, and
    \begin{align*}
    \frac{\weight}{\agentNum}\left(\frac{1}{A_{k}}-\frac{2}{\ratioequ}(\sum\limits_{k=1}^K\pi_{k}p_{k})\right)&=\frac{\weight}{\agentNum}\left(\frac{\sqrt{\phi_k}\sum\limits_{k=1}^K\pi_k\sqrt{\phi_k}}{B/\agentNum-l}-2\right)\\
    &\ge \frac{\weight}{\agentNum}\left(\frac{\sqrt{\phi_1}\sum\limits_{k=1}^K\pi_k\sqrt{\phi_k}}{B/\agentNum-l}-2\right)
    \\&>\frac{\weight}{\agentNum}\left(2+\frac{\ratioequ(1-\ratioequ)(1-\weight)\agentNum}{\weight}-2\right)
    \\&=\ratioequ(1-\ratioequ)(1-\weight).
    \end{align*}
    Here, the  equality in the first line follows from $\sum_{k=1}^K p_k = \ratioequ$. The inequality  in the second line follows from $\phi_k \geq \phi_1$. The inequality in the third line follows from  $\frac{B/s-l}{\weight\ratioequ}<\frac{Q(1,1)}{R(1,1)}$. Thus, $\frac{\weight}{\agentNum}\left(\frac{1}{A_{k}}-\frac{2}{\ratioequ}(\sum\limits_{k=1}^K\pi_{k}p_{k})\right)-\ratioequ(1-\ratioequ)(1-\weight)>0$.  
    So $p$ is indeed a best response to $A$, as per Lemma \ref{lem_q_br}. Meanwhile, we can see that $p$ is indeed non-decreasing.

    \item Case 2:   $\frac{Q(1,1)}{R(1,1)}\le\frac{B/s-l}{\weight\ratioequ}< \frac{Q(K,1)}{R(K,1)}$. 
    
    \begin{Cla}\label{cla_existequality}
    There exists $m^*\in \{1,...,K-1\}$ and $z^*\in(\frac{\phi_{m^*}}{\phi_{m^*+1}},1]$ such that 
   
        \begin{equation}\label{equ_QRrelation+}
         \frac{Q(m^*,z^*)}{R(m^*,z^*)}= \frac{B/s-l}{\weight\ratioequ}.
    \end{equation}
    \end{Cla}
    \begin{proof}
    Recall  that $\frac{Q(m,1)}{R(m,1)}$  is increasing in $m$ according to Claim \ref{cla_monotonic}. As  $\frac{Q(1,1)}{R(1,1)}\le\frac{B/s-l}{\weight\ratioequ}< \frac{Q(K,1)}{R(K,1)}$, there exists a unique $m^*\in \{1,...,K-1\}$ such that 
    
    \[
        \frac{Q(m^*,1)}{R(m^*,1)}\le \frac{B/s-l}{\weight\ratioequ} <\frac{Q(m^*+1,1)}{R(m^*+1,1)}.
    \]
    Recall that $Q(m+1,1)=Q\left(m,\frac{\phi_m}{\phi_{m+1}}\right)$ and $R(m+1,1)=R\left(m,\frac{\phi_m}{\phi_{m+1}}\right)$ in Claim \ref{cla_equality}. 
    We have 
     \[
     \frac{Q(m^*,1)}{R(m^*,1)}\le \frac{B/s-l}{\weight\ratioequ} <\frac{Q(m^*+1,1)}{R(m^*+1,1)}=\frac{Q\left(m^*,\frac{\phi_{m^*}}{\phi_{m^*+1}}\right)}{R\left(m^*,\frac{\phi_{m^*}}{\phi_{m^*+1}}\right)}.
     \]
     Notice that $Q(m,z)$ is continuously decreasing in $z$ and $R(m,z)$ is continuously increasing in $z$. So we have $\frac{Q(m,z)}{R(m,z)}$  continuously decreasing in $z$. Thus, there is a unique $z^*\in(\frac{\phi_{m^*}}{\phi_{m^*+1}},1]$ such that 
     \[
     \frac{Q(m^*,z^*)}{R(m^*,z^*)}= \frac{B/s-l}{\weight\ratioequ}.
     \]
     \end{proof}

    There are two sub cases depending on the value of $\frac{B/s-l}{Q(m^*,z^*)}$. 
    \begin{itemize}
        \item Case 2(a): 
        $\frac{B/s-l}{Q(m^*,z^*)}\le 1$.  Then
        \begin{equation}\label{equ_q_1}
        p_{k}=\begin{cases}\frac{z^*}{\phi_{m^*}}\cdot \phi_{k}, &\text{if} \ \  k \le m^*,\\
        1, &\text{if} \ \ k>m^*,
        \end{cases}
        \end{equation} and
        \begin{equation}\label{equ_A_2}
        A_{k}=\begin{cases}\chi\triangleq\frac{B/s-l}{Q(m^*,z^*)}, &\text{if} \ \  k \le m^*,\\
    \frac{1}{\sqrt{\phi_{k}}} \cdot\frac{B/s-l-\chi\sum\limits_{k=1}^{ m^*} \pi_{k}\phi_{k}}{\sum\limits_{k= m^*+1}^K\pi_{k}\sqrt{\phi_{k}}} , &\text{if}  \ \ k > m^*,
    \end{cases}
        \end{equation}
         constitute the mutual best response.
        Indeed, on the one hand, $A$ is a best response to $p$. Let 
        \[
        \lambda^* = \frac{\weight\left(\sum\limits_{k=1}^{K}\pi_k\sqrt{p_k\phi_k}\right)^2}{\agentNum(\ratioequ)^2\left(\budget/\agentNum-l\right)^2}. 
        \] 
        For $k\le m^*$, we have
        \begin{align*}
        \sqrt{\frac{\weight p_k}{\agentNum \left(\ratioequ\right)^2 \lambda^* \phi_k}}=&\sqrt{\frac{\weight}{\agentNum(\ratioequ)^2}\times \frac{\agentNum(\ratioequ)^2\left(\budget/\agentNum-l\right)^2}{\weight\left(\sum\limits_{k=1}^{K}\pi_k\sqrt{p_k\phi_k}\right)^2}}\times \sqrt{\frac{p_k}{\phi_{k}}}\\
        =&\frac{\budget/\agentNum-l}{\sum\limits_{k=1}^{m^*}\pi_k\sqrt{\frac{z^*}{\phi_{m^*}}\cdot \phi_{k}\cdot \phi_k}+\sum\limits_{k=m^*+1}^K\pi_k \sqrt{\phi_k}} \times \sqrt{\frac{z^*}{\phi_{m^*}}}\\
        =&\frac{B/s-l}{Q(m^*,z^*)}=A_k\le 1.
        \end{align*}
        For $k>m^*$, by a similar derivation, we have 
        \begin{align*}
        \sqrt{\frac{\weight p_k}{\agentNum \left(\ratioequ\right)^2 \lambda^* \phi_k}}&=\frac{\budget/\agentNum-l}{Q(m^*,z^*)}\times\sqrt{\frac{\phi_{m^*}}{z^*}}\times \sqrt{\frac{1}{\phi_k}}\\
        &=\frac{\budget/\agentNum-l}{Q(m^*,z^*)} \times\frac{\sum\limits_{k=1 }^{m^*}\pi_{k}\phi_{k}+\sum\limits_{k=m^*+1}^{K}\pi_{k}\sqrt{\frac{\phi_{k}\phi_{m}}{z^*}}-\sum\limits_{k=1}^{m^*} \pi_{k}\phi_{k}}{\sum\limits_{k=m^*+1}^K\pi_{k}\sqrt{\phi_{k}}}\times \sqrt{\frac{1}{\phi_k}}\\
        &=\frac{B/s-l}{Q(m^*,z^*)}\times\frac{Q(m^*,z^*)-\sum\limits_{k=1}^{m^*} \pi_{k}\phi_{k}}{\sum\limits_{k= m^*+1}^K\pi_{k}\sqrt{\phi_{k}}}\times \sqrt{\frac{1}{\phi_k}}\\
        &= \frac{B/s-l-\frac{B/s-l}{Q(m^*,z^*)}\cdot\sum\limits_{k=1}^{m^*} \pi_{k}\phi_{k}}{\sum\limits_{k= m^*+1}^K\pi_{k}\sqrt{\phi_{k}}}\times \sqrt{\frac{1}{\phi_k}}\\
        &=\frac{B/s-l-\chi\cdot\sum\limits_{k=1}^{m^*} \pi_{k}\phi_{k}}{\sum\limits_{k= m^*+1}^K\pi_{k}\sqrt{\phi_{k}}}\times \sqrt{\frac{1}{\phi_k}}= A_k.
        \end{align*}
        
        As the virtual cost $\phi_k$ is increasing, $A_k$ is decreasing for $k>m^*$. Notice that $A_{m^*}>A_{m^*+1}$. This is because $z^*>\frac{\phi_{m^*}}{\phi_{m^*+1}}$, and naturally
        \[
        A_{m^*} = \sqrt{\frac{\weight p_k}{\agentNum \left(\ratioequ\right)^2 \lambda^* \phi_k}}= \sqrt{\frac{\weight }{\agentNum \left(\ratioequ\right)^2 \lambda^*}}\times \sqrt{\frac{z^*}{\phi_{m^*}}}>\sqrt{\frac{\weight }{\agentNum \left(\ratioequ\right)^2 \lambda^*}}\times \sqrt{\frac{1}{\phi_{m^*+1}}}=A_{m^*+1}.
        \]
        From this, we can see that $A_k$ given by (\ref{equ_A_2}) is non-increasing, and hence,  $A_k \leq 1$ for all $k$. 
        Thus, we have $A_k = \min\left\{1,\sqrt{\frac{\weight p_{k}}{\agentNum \left(\ratioequ\right)^2 \lambda^* \phi_{k}}}\right\}$, for all $k$.  And the budget constraint is binding, as
        \begin{align*}
            \sum\limits_{k=1}^K \pi_k \phi_k A_k&= \sum\limits_{k=1}^{m^*} \pi_k \phi_k \chi+ \sum\limits_{k=m^*+1}^{K} \pi_k \phi_k \frac{1}{\sqrt{\phi_{k}}} \cdot\frac{B/s-l-\chi\sum\limits_{k=1}^{ m^*} \pi_{k}\phi_{k}}{\sum\limits_{k= m^*+1}^K\pi_{k}\sqrt{\phi_{k}}}\\
            &=\chi\cdot\sum\limits_{k=1}^{m^*} \pi_k \phi_k +\frac{B/s-l-\chi\sum\limits_{k=1}^{ m^*} \pi_{k}\phi_{k}}{\sum\limits_{k= m^*+1}^K\pi_{k}\sqrt{\phi_{k}}} \cdot  \sum\limits_{k=m^*+1}^{K} \pi_k \sqrt{\phi_k}\\
            &=\chi\cdot\sum\limits_{k=1}^{m^*} \pi_k \phi_k+\budget/\agentNum-l-\chi\cdot\sum\limits_{k=1}^{m^*} \pi_k \phi_k=\budget/\agentNum-l.
        \end{align*}
        So $A$ is indeed best response to $p$, as per Lemma \ref{lem: min_best_resp}.   
        
        On the other hand, $p$ is a best response to $A$. Recall that $\frac{Q(m^*,z^*)}{R(m^*,z^*)}= \frac{B/s-l}{\weight\ratioequ}.$ For $k\le m^*$, recall that  $0\le p_k\le 1$ and $A_k = \frac{\budget/\agentNum-l}{Q(m^*,z^*)}$. And we have  
        \begin{align*}
        \frac{\weight}{\agentNum(\ratioequ)^2}\left(\frac{1}{A_{k}}-\frac{2}{\ratioequ}(\sum_{k=1}^{K}\pi_{k}p_{k})\right)&=\frac{\weight}{\agentNum(\ratioequ)^2}\left(\frac{Q(m^*,z^*)}{\budget/\agentNum-l}-\frac{2}{\ratioequ}\left(\sum\limits_{k=1 }^{m^*}\pi_{k}\phi_{k}\frac{z^*}{\phi_{m}}+\sum\limits_{k= m^*+1}^K\pi_{k}\right)\right)\\
        &=\frac{\weight}{\agentNum(\ratioequ)^2}\left(\frac{R(m^*,z^*)}{\weight\ratioequ}-\frac{2}{\ratioequ}\left(\sum\limits_{k=1 }^{m^*_+}\pi_{k}\phi_{k}\frac{z^*}{\phi_{m}}+\sum\limits_{k= m^*+1}^K\pi_{k}\right)\right)\\
        &=\frac{1}{\agentNum(\ratioequ)^3}\left(R(m^*,z^*)-2\weight\left(\sum\limits_{k=1 }^{m^*}\pi_{k}\phi_{k}\frac{z^*}{\phi_{m}}+\sum\limits_{k= m^*+1}^K\pi_{k}\right)\right)\\
        &=\frac{1}{\agentNum(\ratioequ)^3}\times\left(\ratioequ\right)^2(1-\ratioequ)(1-\weight)\agentNum\\
        &=\frac{1}{\ratioequ}(1-\ratioequ)(1-\weight),
        \end{align*}
        i.e.,$\frac{\weight}{\agentNum(\ratioequ)^2}\left(\frac{1}{A_{k}}-\frac{2}{\ratioequ}\left(\sum\limits_{k=1}^K\pi_{k}p_{k}\right)\right)-\frac{1}{\ratioequ}(1-\ratioequ)(1-\weight)=0$. With this equality,  we do have a best response for $k\le m^*$.
        
        For $k>m^*$, we have $p_k=1$.  Remember that $A_{m^*+1}>A_{m^*}$ and that $A_k$ is decreasing for $k>m^*$. We obtain $A_k<A_{m^*}$ for $k>m^*$. 
        Then we have for $k>m^*_+$, 
        \begin{align*}
            &\ \ \frac{\weight}{\agentNum(\ratioequ)^2}\left(\frac{1}{A_{k}}-\frac{2}{\ratioequ}(\sum\limits_{k=1}^K\pi_{k}p_{k})\right)-\frac{1}{\ratioequ}(1-\ratioequ)(1-\weight)
            \\>&\ \ \frac{\weight}{\agentNum(\ratioequ)^2}\left(\frac{1}{A_{m^*}}-\frac{2}{\ratioequ}(\sum\limits_{k=1}^K\pi_{k}p_{k})\right)-\frac{1}{\ratioequ}(1-\ratioequ)(1-\weight)
            \\=& \ \ 0
        \end{align*}
        
         So $p$ is indeed best response to $A$, as per Lemma \ref{lem_q_br}.  And we can see that $p_k$ is indeed non-decreasing. 
      
        \item Case 2(b):  $\frac{B/s-l}{Q(m^*,z^*)}> 1$. 
        
        \begin{Cla}\label{clam_exist_R_equality}
        There exists $k'\in \{1,...,m^*\}$ and $z'\in (\frac{\phi_{k'}}{\phi{k' + 1}},1]$
        such that 
        \begin{equation}
            R(k',z')= \weight\ratioequ.
        \end{equation}
        \end{Cla}
        \begin{proof}
        Recall that $R(m^*,z^*)=\frac{Q(m^*,z^*)}{B/s-l}\cdot\weight\ratioequ$ and  $\frac{B/s-l}{Q(m^*,z^*)}> 1$. We obtain $R(m^*,z^*)<\weight\ratioequ$. Recall that $R(m,z)$  is  decreasing in $m$, and increasing in $z$, and $R(m^*+1,1) = R\left(m^*,\frac{\phi_{m^*}}{\phi_{m^*+1}}\right)$ (Claim \ref{cla_equality}). We have 
        \[
        R(m^*+1,1) = R\left(m^*,\frac{\phi_{m^*}}{\phi_{m^*+1}}\right) <
        R(m^*,z^*) < \weight\ratioequ.
        \]
        Here, the first inequality holds as $\frac{\phi_{m^*}}{\phi_{m^*+1}}< z^*$, according to Claim \ref{cla_existequality}. Meanwhile,  $R(1,1)=2\weight\ratioequ+\left(\ratioequ\right)^2(1-\ratioequ)(1-\weight)\agentNum>\weight\ratioequ$. Thus, there exists $k'\le m^*$ such that 
        \[
        R(k'+1,1) <\weight\ratioequ\le  R(k',1).
        \]
        Recall that $R\left(k',\frac{\phi_{k'}}{\phi_{k'+1}}\right)=R(k'+1,1)$. We have 
        
        \[
        R\left(k',\frac{\phi_{k'}}{\phi_{k'+1}}\right)=R(k'+1,1) <\weight\ratioequ\le  R(k',1).
        \]
        
        Since $R(m,z)$  is continuously increasing in $z$, there exists $z'\in (\frac{\phi_{k'}}{\phi_{k'+1}},1]$ such that 
        \[
         R(k',z')= \weight\ratioequ.
        \] 
        \end{proof}

        Define \begin{equation}\label{equ_k+}
            k^* = \max\{k:Q(k,1)<B/s-l\}.
        \end{equation} 
        Then 
        \begin{equation}\label{equ_q_2}
        p_{k}=\begin{cases}\frac{z'}{\phi_{k'}}\cdot \phi_{k}, &\text{if} \ \  k \le k',\\
        1, &\text{if} \ \  k>k',
        \end{cases}
        \end{equation}
        and

        \begin{equation}
        A_{k}=\begin{cases}\chi\triangleq1, &\text{if} \ \ k \le k^*,\\
    \frac{1}{\sqrt{\phi_{k}}} \cdot\frac{B/s-l-\sum\limits_{k=1}^{ k^*} \pi_{k}\phi_{k}}{\sum\limits_{k= k^*+1}^K\pi_{k}\sqrt{\phi_{k}}} , &\text{if}  \ \ k>k^*,
    \end{cases}
        \end{equation}
        constitute the mutual best response. To see this, on the one hand, $A$ is best response to $p$. Recall that $Q(m,z)$ is increasing in $m$ and decreasing in $z$. As $k' \le m^*$ (Claim \ref{clam_exist_R_equality}) and $z^*\le 1$ (Claim \ref{cla_existequality}),  we have 
        \[
        Q(k',1) \le Q(m^*,1) \le Q(m^*,z^*) <  B/s-l,
        \]
        i.e., $B/s-l< Q(k',1)$. By definition
        $k^* = \max\{k:Q(k,1)<B/s-l\}$, we have  $k^*\ge k'$.

       Meanwhile, we have $Q(k^*+1,1)\ge B/s-l$,
       i.e., 
    \begin{align*}
        &\sum\limits_{k=1 }^{k^*+1}\pi_{k}\phi_{k}+\sqrt{\phi_{k^*+1}}\sum\limits_{k=k^*+2}^{K}\pi_{k}\sqrt{\phi_{k}}\\
        =&\sum\limits_{k=1 }^{k^*}\pi_{k}\phi_{k}+\sqrt{\phi_{k^*+1}}\cdot \pi_{k^*+1}\sqrt{\phi_{k^*+1}}+\sqrt{\phi_{k^*+1}}\cdot\sum\limits_{k=k^*+2}^{K}\pi_{k}\sqrt{\phi_{k}}\\
        =&\sum\limits_{k=1 }^{k^*}\pi_{k}\phi_{k}+\sqrt{\phi_{k^*+1}}\cdot\sum\limits_{k=k^*+1}^{K}\pi_{k}\sqrt{\phi_{k}}\\
        \ge&B/s-l.
    \end{align*}
         
        Thus, 
        \[
        \sum\limits_{k=1 }^{k^*}\pi_{k}\phi_{k}+\sqrt{\phi_{k^*+1}}\cdot\sum\limits_{k=k^*+1}^{K}\pi_{k}\sqrt{\phi_{k}}\ge B/s-l,
        \]
        which is actually
        \[
        \frac{1}{\sqrt{\phi_{k^*+1}}} \cdot\frac{B/s-l-\sum\limits_{k=1}^{ k^*} \pi_{k}\phi_{k}}{\sum\limits_{k= k^*+1}^K\pi_{k}\sqrt{\phi_{k}}}=A_{k^*+1}
        \le 1.
        \]
        Let
        \[
        \lambda^* = \frac{\weight\left(\sum\limits_{k=k^*+1}^{K}\pi_k\sqrt{\phi_k}\right)^2}{\agentNum(\ratioequ)^2\left(\budget/\agentNum-l-\sum\limits_{k=1}^{k^*}\pi_k\phi_k\right)^2}. 
        \] 
        For $k>k^*$, as $ k^*\ge k'$, we have $p_k=1$  
        and 
        \begin{align*}
        \sqrt{\frac{\weight p_k}{\agentNum \left(\ratioequ\right)^2 \lambda^* \phi_k}}
        &=\sqrt{\frac{\weight}{\agentNum(\ratioequ)^2}\times \frac{\agentNum(\ratioequ)^2\left(\budget/\agentNum-l-\sum\limits_{k=1}^{k^*}\pi_k\phi_k\right)^2}{\weight\left(\sum\limits_{k=k^*+1}^{K}\pi_k\sqrt{\phi_k}\right)^2}}\times \sqrt{\frac{1}{\phi_{k}}}
        \\&= \frac{1}{\sqrt{\phi_{k}}} \cdot\frac{B-l-\sum\limits_{k=1}^{k^*} \pi_{k}\phi_{k}}{\sum\limits_{k= k^*+1}^K\pi_{k}\sqrt{\phi_{k}}}
        \\&=A_k
        \end{align*}
        Meanwhile, due to the increasing virtual costs in $k$, we have $A_{k}\le  A_{k^*+1} \le 1$, for all $k \ge k^*+1$. So we have $\min\left\{1,\sqrt{\frac{\weight p_k}{\agentNum \left(\ratioequ\right)^2 \lambda^* \phi_k}}\right\}=A_k$ for all $k>k^*$.
        
        Next, in the regime $k\le k^*$, we consider two possibilities of $k^*$: $k^*>k'$ and $k^*=k'$. Recall that $k^*\ge k'$. 
        \begin{itemize}
            \item Firstly, we focus on $k^*>k'$. For $k\in [k'+1,k^*]$, we have $p_k =1$ and
        \begin{align*}
        \sqrt{\frac{\weight p_k}{\agentNum \left(\ratioequ\right)^2 \lambda^* \phi_k}}=\sqrt{\frac{\weight}{\agentNum \left(\ratioequ\right)^2 \lambda^* \phi_k}}\ge\sqrt{\frac{\weight}{\agentNum \left(\ratioequ\right)^2 \lambda^* \phi_{k^*}}}>1.
        \end{align*}
        The last inequality holds due to $Q(k^*,1)<B/s-l$, i.e., 
        \[
        \sum\limits_{k=1 }^{k^*}\pi_{k}\phi_{k}+\sum\limits_{k=k^*+1}^{K}\pi_{k}\sqrt{\phi_{k}\phi_{k^*}}< B/s-l,
        \]
        which leads to (by plugging the value of $\lambda^*$)
        \begin{align*}
        \sqrt{\frac{\weight}{\agentNum \left(\ratioequ\right)^2 \lambda^* \phi_{k^*}}}&=\sqrt{\frac{\weight}{\agentNum(\ratioequ)^2}\times \frac{\agentNum(\ratioequ)^2\left(\budget/\agentNum-l-\sum\limits_{k=1}^{k^*}\pi_k\phi_k\right)^2}{\weight\left(\sum\limits_{k=k^*+1}^{K}\pi_k\sqrt{\phi_k}\right)^2}}\cdot \frac{1}{\sqrt{\phi_{k^*}}}\\
        &=\frac{B/s-l-\sum\limits_{k=1}^{ k^*} \pi_{k}\phi_{k}}{\sum\limits_{k= k^*+1}^K\pi_{k}\sqrt{\phi_{k}}}\cdot \frac{1}{\sqrt{\phi_{k^*}}}\\
        &> 1.
        \end{align*}
        Thus, $\min\left\{1,\sqrt{\frac{\weight p_k}{\agentNum \left(\ratioequ\right)^2 \lambda^* \phi_k}}\right\}=1=A_k$.
        
        For $k\le k'$,  since $z'\ge \frac{\phi_{k'}}{\phi_{k'+1}}$ (Claim \ref{clam_exist_R_equality}) and $\phi_{k^*}\ge \phi_{k'}$, we have 
        \begin{align*}
            \sqrt{\frac{\weight p_k}{\agentNum \left(\ratioequ\right)^2 \lambda^* \phi_k}}
            &=\sqrt{\frac{\weight}{\agentNum \left(\ratioequ\right)^2 \lambda^*}}\times\sqrt{\frac{z'}{\phi_{k'}}}\\&\ge \sqrt{\frac{\weight}{\agentNum \left(\ratioequ\right)^2 \lambda^*}}\times \sqrt{\frac{1}{\phi_{k'+1}}}\\
            &\ge \sqrt{\frac{\weight}{\agentNum \left(\ratioequ\right)^2 \lambda^*}}\times \sqrt{\frac{1}{\phi_{k^*}}}\\
            &>1.
        \end{align*}
        Thus, we have $\min\left\{1,\sqrt{\frac{\weight p_k}{\agentNum \left(\ratioequ\right)^2 \lambda^* \phi_k}}\right\}=1=A_k$. 
        
        In summary, we have $\min\left\{1,\sqrt{\frac{\weight p_k}{\agentNum \left(\ratioequ\right)^2 \lambda^* \phi_k}}\right\}=A_k$ for all $k\le k^*$ when $k^*>k'$.

            \item Secondly, we focus on $k^*=k'$.
            On the one hand, we have  $k' \leq m^*$ according to  Claim \ref{clam_exist_R_equality},  and therefore,  $k^*\le m^*$. On the other hand, $Q(m^*,1) \leq Q(m^*,z^*) < B/s - l$, so we have that $k^* \geq m^*$ by definition of $k^*=\max\{k: Q(k,1)<B/s-l\}$. So, it must be that $k' = k^* = m^*$. 
            
            Recall in the proof of Claim \ref{clam_exist_R_equality} that $R(m^*,z^*) < \gamma \ratioequ = R(k',z') = R(m^*,z')$, i.e., $R(m^*,z^*) < R(m^*,z')$.  And since $R(m,z)$ is increasing in $z$, we have $z^* < z'$. Since $Q(m,z)$ is decreasing in $z$, it must be that 
        \[
        Q(m^*,z') < Q(m^*,z^*) < B/s - l.
        \]
        That is 
        \begin{align*}
            Q(m^*,z') = \sum\limits_{k=1}^{m^*} \pi_{k}\phi_{k}+ \sqrt{\frac{\phi_{m^*}}{z'}}\sum\limits_{k= m^*+1}^K\pi_{k}\sqrt{\phi_{k}}<B/s - l.
        \end{align*}
        For $k\le k^*=k'$, we have $p_k=\frac{z'}{\phi_{k'}}\cdot \phi_{k}$, and 
         \begin{align*}
            \sqrt{\frac{\weight p_k}{\agentNum \left(\ratioequ\right)^2 \lambda^* \phi_k}}
            &=\sqrt{\frac{\weight}{\agentNum \left(\ratioequ\right)^2 \lambda^*}}\times\sqrt{\frac{z'}{\phi_{k'}}}
            \\&= \sqrt{\frac{z'}{\phi_{m^*}}}\cdot \frac{B-l-\sum\limits_{k=1}^{m^*} \pi_{k}\phi_{k}}{\sum\limits_{k= m^*+1}^K\pi_{k}\sqrt{\phi_{k}}}\\
            &>1.
        \end{align*}
        Thus, we have $\min\left\{1,\sqrt{\frac{\weight p_k}{\agentNum \left(\ratioequ\right)^2 \lambda^* \phi_k}}\right\}=1=A_k$. 
        
        In summary, we have $\min\left\{1,\sqrt{\frac{\weight p_k}{\agentNum \left(\ratioequ\right)^2 \lambda^* \phi_k}}\right\}=A_k$ for all $k\le k^*$ when $k^*=k'$. 
        \end{itemize}


         In conclusion, we have $A_k = \min\left\{1,\sqrt{\frac{\weight p_{k}}{\agentNum \left(\ratioequ\right)^2 \lambda^* \phi_{k}}}\right\}$, for all $k$, and $A_k$ is non-increasing in $k$. And the budget constraint is binding through similar calculation. So $A$ is indeed best response to $p$. 
        
        On the other hand, $p$ is best response to $A$. For $k\le k^*$, we have $0<p_k\le 1$ and $A_k=1$.  From $R(k',z')=\weight\ratioequ$, we obtain
        
        \begin{align*}
            R(k',z')=2\weight\left(\sum\limits_{k=1 }^{k'}\pi_{k}\phi_{k}\frac{z'}{\phi_{k'}}+\sum\limits_{k= k'+1}^K\pi_{k}\right)+\left(\ratioequ\right)^2(1-\ratioequ)(1-\weight)\agentNum=\weight\ratioequ,
        \end{align*}
        
        i.e., 
         \begin{align*}
            2\weight\left(\sum\limits_{k=1 }^{K}\pi_{k}p_k\right)+\left(\ratioequ\right)^2(1-\ratioequ)(1-\weight)\agentNum=\weight\ratioequ.
        \end{align*}
        So
        \begin{align*}
            &\frac{\weight}{\agentNum(\ratioequ)^2}\left(\frac{1}{A_{k}}-\frac{2}{\ratioequ}(\sum\limits_{k=1}^K\pi_{k}p_{k})\right)-\frac{1}{\ratioequ}(1-\ratioequ)(1-\weight)\\=&\frac{\weight}{\agentNum(\ratioequ)^2}\left(1-\frac{2}{\ratioequ}(\sum\limits_{k=1}^K\pi_{k}p_{k})\right)-\frac{1}{\ratioequ}(1-\ratioequ)(1-\weight)\\=&0. 
        \end{align*}
       With the above equality, any $p_k\in[0,1]$ is a best response according to Lemma \ref{lem_q_br}.
       For $k> k^*$, we have $p_k=1$, and since $A_{k^*}<1$, we have 
        \begin{align*}
            &\frac{\weight}{\agentNum(\ratioequ)^2}\left(\frac{1}{A_{k}}-\frac{2}{\ratioequ}(\sum\limits_{k=1}^K\pi_{k}p_{k})\right)-\frac{1}{\ratioequ}(1-\ratioequ)(1-\weight)\\>&\frac{\weight}{\agentNum(\ratioequ)^2}\left(1-\frac{2}{\ratioequ}(\sum\limits_{k=1}^K\pi_{k}p_{k})\right)-\frac{1}{\ratioequ}(1-\ratioequ)(1-\weight)\\=&0.
        \end{align*}

       So $p$ is indeed a best response to $A$, as per Lemma \ref{lem_q_br}. And we can see that $p_k$ is indeed non-decreasing. 
       
        \item Case 3:   $\frac{B/s-l}{\weight\ratioequ}\ge \frac{Q(K,1)}{R(K,1)}$. Recall  that $Q(m,z)$ is continuously decreasing in $z$ and $R(m,z)$ is continuously increasing in $z$. Thus, $\frac{Q(m,z)}{R(m,z)}$ is decreasing in $z$.  Notice that when $z=0$,  we have
        \begin{equation*}
        \frac{Q(K,0)}{R(K,0)}=\frac{\sum\limits_{k=1 }^{K}\pi_{k}\phi_{k}}{2\weight\left(\frac{0}{\phi_{m}} \cdot \sum\limits_{k=1 }^K\pi_{k}\phi_{k}\right)+\left(\ratioequ\right)^2(1-\ratioequ)(1-\weight)\agentNum} =\frac{\sum\limits_{k=1 }^{K}\pi_{k}\phi_{k}}{\left(\ratioequ\right)^2(1-\ratioequ)(1-\weight)\agentNum}.
        \end{equation*}
        
        
        If  $ \frac{\sum\limits_{k=1 }^{K}\pi_{k}\phi_{k}}{\left(\ratioequ\right)^2(1-\ratioequ)(1-\weight)\agentNum}>\frac{B/s-l}{\weight\ratioequ}$, let  $\tilde{z}\in (0,1]$ be such that $ \frac{Q(K,\tilde{z})}{R(K,\tilde{z})}=\frac{B/s-l}{\weight\ratioequ}$. 
        Such a $\tilde{z}$ exists because $\frac{Q(K,1)}{R(K,1)}\le \frac{B/s-l}{\weight\ratioequ}<\frac{Q(K,0)}{R(K,0)}$. Then \begin{equation}
            p_k=\tilde{z}\cdot\frac{\phi_k}{\phi_K}, \forall k,
        \end{equation} and 
        \begin{equation}
            A_k = \frac{\budget/\agentNum-l}{\sum\limits_{k=1 }^{K}\pi_{k}\phi_{k}}, \forall k
        \end{equation}
        constitute the mutual best response. To see this, on the one hand, $A$ is best response to $p$. Let 
        \[
        \lambda^* = \frac{\weight\left(\sum\limits_{k=1}^{K}\pi_k\sqrt{p_k\phi_k}\right)^2}{\agentNum(\ratioequ)^2\left(\budget/\agentNum-l\right)^2}. 
        \] 
        Then 
        \begin{align*}
        \sqrt{\frac{\weight p_k}{\agentNum \left(\ratioequ\right)^2 \lambda^* \phi_k}}=&\sqrt{\frac{\weight}{\agentNum(\ratioequ)^2}\times \frac{\agentNum(\ratioequ)^2\left(\budget/\agentNum-l\right)^2}{\weight\left(\sum\limits_{k=1}^{K}\pi_k\sqrt{p_k\phi_k}\right)^2}}\times \sqrt{\frac{p_k}{\phi_{k}}}\\
        =&\frac{\budget/\agentNum-l}{\sum\limits_{k=1}^{K}\pi_k\sqrt{\frac{\tilde{z}}{\phi_{K}}\cdot \phi_{k}\cdot \phi_k}} \times \sqrt{\frac{\tilde{z}}{\phi_{K}}}\\
        =&\frac{B/s-l}{\sum\limits_{k=1 }^{K}\pi_{k}\phi_{k}}=A_k.
        \end{align*}
        And $\frac{B/s-l}{\sum\limits_{k=1 }^{K}\pi_{k}\phi_{k}}<1$, as we assume  $\sum\limits_{k=1 }^{K}\pi_{k}\phi_{k}>B/s-l$ without loss of generality to avoid trivial solution of any selection probability being one.  Thus, we $\min\left\{1,\sqrt{\frac{\weight p_k}{\agentNum \left(\ratioequ\right)^2 \lambda^* \phi_k}}\right\}=A_k$, for all $k$. And the budget constraint is binding as $$\sum_{k=1}^K\pi_k \phi_k A_k =\frac{B/s-l}{\sum\limits_{k=1 }^{K}\pi_{k}\phi_{k}} \cdot \sum_{k=1}^K\pi_k \phi_k=B/s-l.$$ So $A$ is indeed best response to $p$, as per Lemma \ref{lem: min_best_resp}. 
        
        On the other hand, $p$ is a best response to $A$. Notice that 
        \begin{equation*}
            A_k = \frac{B/s-l}{\sum\limits_{k=1 }^{K}\pi_{k}\phi_{k}} = \frac{B/s-l}{Q(K,\tilde{z})}.
        \end{equation*}
        From $ \frac{Q(K,\tilde{z})}{R(K,\tilde{z})}=\frac{B/s-l}{\weight\ratioequ}$, we have 
        
        \begin{align*}
        \frac{\weight}{\agentNum(\ratioequ)^2}\left(\frac{1}{A_{k}}-\frac{2}{\ratioequ}(\sum_{k=1}^{K}\pi_{k}p_{k})\right)&=\frac{\weight}{\agentNum(\ratioequ)^2}\left(\frac{Q(K,\tilde{z})}{\budget/\agentNum-l}-\frac{2}{\ratioequ}\left(\sum\limits_{k=1 }^{K}\pi_{k}\tilde{z}\frac{\phi_{k}}{\phi_{K}}\right)\right)\\
        &=\frac{\weight}{\agentNum(\ratioequ)^2}\left(\frac{R(K,\tilde{z})}{\weight\ratioequ}-\frac{2}{\ratioequ}\left(\sum\limits_{k=1 }^{K}\pi_{k}\tilde{z}\frac{\phi_k}{\phi_{K}}\right)\right)\\
        &=\frac{1}{\agentNum(\ratioequ)^3}\left(R(K,\tilde{z})-2\weight\left(\sum\limits_{k=1 }^{K}\pi_{k}\tilde{z}\frac{\phi_k}{\phi_{m}}\right)\right)\\
        &=\frac{1}{\agentNum(\ratioequ)^3}\times\left(\ratioequ\right)^2(1-\ratioequ)(1-\weight)\agentNum\\
        &=\frac{1}{\ratioequ}(1-\ratioequ)(1-\weight),
        \end{align*}
        i.e.,$\frac{\weight}{\agentNum(\ratioequ)^2}\left(\frac{1}{A_{k}}-\frac{2}{\ratioequ}\left(\sum\limits_{k=1}^K\pi_{k}p_{k}\right)\right)-\frac{1}{\ratioequ}(1-\ratioequ)(1-\weight)=0$. 
         So $pk=\tilde{z}\cdot\frac{\phi_k}{\phi_K}$, for all $k$, is indeed best response to $A$, as per Lemma \ref{lem_q_br}. And we can see that $p_k$ is indeed non-decreasing.  

         If $ \frac{\sum\limits_{k=1 }^{K}\pi_{k}\phi_{k}}{\left(\ratioequ\right)^2(1-\ratioequ)(1-\weight)\agentNum}\le\frac{B/s-l}{\weight\ratioequ}$, then $p_k=0$ for all $k$ and
         
         \begin{equation}
            A_k = \frac{\budget/\agentNum-l}{\sum\limits_{k=1 }^{K}\pi_{k}\phi_{k}}, \forall k,
         \end{equation}
         constitute one mutual best response. To see this, on one hand, $A$ is one best response to $p$. This is because when $p_k=0$ for all $k$, any $A_k\in [0,1]$ that satisfies budget constraint is a best response, according to Lemma \ref{lem: min_best_resp}. And $A_k<1$ indeed satisfies budget constraint. Actually, there are other formats of $A$ that satisfy budget constraint being best response. We present the constant solution, which is non-increasing, for simplicity.
         
         On the other hand, $p$ is best response to $A$.  We have
        \begin{align*}
          \frac{\weight}{\agentNum(\ratioequ)^2}\cdot\frac{1}{A_k}-\frac{1}{\ratioequ}(1-\ratioequ)(1-\weight)= \frac{\weight}{\agentNum(\ratioequ)^2}\left(\frac{\sum\limits_{k=1 }^{K}\pi_{k}\phi_{k}}{B/s-l}\right)-\frac{1}{\ratioequ}(1-\ratioequ)(1-\weight)\le 0.
        \end{align*}
        So $p_k=0$ is indeed best response to $A$, according to Lemma \ref{lem_q_br}.

\end{itemize}

    \end{itemize}

\paragraph{From discrete to continuous costs} So far we present the intersections of both players' best responses, i.e., an equilibrium of the min-max game. And $A_k$ of the equilibrium corresponds to the solution of the minimax optimization problem. We then convert the solution under discrete cost to the solution under continuous cost. The continuous cost  can be considered as a special case of discrete cost  by setting $K$ to be infinity. 

We begin with some key notations under continuous case. We denote the distribution of virtual cost in group $i$ as $\virtualcost{i}$. We can obtain it based on the connection between cost and virtual cost.  Recall that the probability of group $i$ as $q_i$. Thus, we have the distribution of virtual cost in all groups as  $ \sum_{i}q_i \virtualcost{i}(\phi)\triangleq \virtualcostall(\phi)$. As non-participants whose costs are higher than the cost threshold would not get payments, we only need to  focus on participants and their virtual costs. Let $\phi_{i\min}$ and $\phi_{i\max}$ be the maximum and minimum value of virtual costs of participants in group $i$, respectively. Let $\phi_{\min}$ and $\phi_{\max}$ be the maximum and minimum value of virtual costs of participants, respectively. Naturally, $\phi_{\min}=\min_{i}\phi_{i\min}$ and $\phi_{\max}=\max_{i}\phi_{i\max}$.

We divide the interval $[\phi_{\min},\phi_{\max}]$ into $K$ numbers of sub-intervals with equal length $(\phi_{\max}-\phi_{\min})/K$. When $K$ is very large, we can approximate the  continuous distribution by $K$  discrete probabilities $$\pi_k = \int_{\phi_{\min}\frac{K-k+1}{K}+\phi_{\max}\frac{k-1}{K}}^{\phi_{\min}\frac{K-k}{K}+ \phi_{\max}\frac{k}{K}}\virtualcostall(\phi)d\phi, \ \  k=1,...,K,$$
which are integrals of $\virtualcostall(\phi)$ in sub-intervals. 
Consider the functions $Q(m,z)$ and $R(m,z)$ defined in \eqref{equ_Q} and \eqref{equ_R}.  We define their continuous versions by replacing summations with integrals as follows:

\begin{equation}
Q_c(x) = \int_{\phi_{\min}}^{x}\phi \virtualcostall(\phi)d\phi + \sqrt{x}\int_{x}^{\phi_{\max}}\sqrt{\phi}\virtualcostall(\phi)d\phi
\end{equation}

\begin{equation}
R_c(x) = 2\weight\left(\frac{1}{x}\int_{\phi_{\min}}^{x}\phi \virtualcostall(\phi)d\phi + \int_{x}^{\phi_{\max}}\virtualcostall(\phi)d\phi\right)+\left(\ratioequ\right)^2(1-\ratioequ)(1-\weight)\agentNum.
\end{equation}
We can see that $Q\left(\frac{x-\phi_{\min}}{\phi_{\max}-\phi_{\min}},1\right)\rightarrow Q_c(x)$ and $R\left(\frac{x-\phi_{\min}}{\phi_{\max}-\phi_{\min}},1\right)\rightarrow R_c(x)$ as $K\rightarrow \infty$. As $\phi_m \rightarrow \phi_{m+1}$, $m=1,.., K$ in the limit $K\rightarrow \infty$, we have $Q(m,z)\rightarrow Q(m+1,z)$ and $R(m,z)\rightarrow R(m+1,z)$.  Since $Q(m+1,1)= Q\left(m,\frac{\phi_m}{\phi_{m+1}}\right), R(m+1,1)= R\left(m,\frac{\phi_m}{\phi_{m+1}}\right)$ according to Claim \ref{cla_equality}, we have $Q\left(m,\frac{\phi_m}{\phi_{m+1}}\right) \rightarrow Q(m,1)$ and $R\left(m,\frac{\phi_m}{\phi_{m+1}}\right) \rightarrow R(m,1)$. Thus, it suffices to consider $z=1$ and leverage the limit $Q\left(\frac{x-\phi_{\min}}{\phi_{\max}-\phi_{\min}},1\right)\rightarrow Q_c(x)$ and $R\left(\frac{x-\phi_{\min}}{\phi_{\max}-\phi_{\min}},1\right)\rightarrow R_c(x)$ when characterizing the connection between $Q(m,z)$, $R(m,z)$ and $Q_c(x)$, $R_c(x)$.

Based on the connection between discrete $Q$ and $R$ and continuous $Q_c$ and $R_c$, we adapt the solutions of discrete case to that of continuous case. 

\begin{itemize}
    \item Case 1:  $\frac{B/s-l}{\weight\ratioequ}<\frac{Q_c(\phi_{\min})}{R_c(\phi_{\min})}$, i.e., 
    $(B/s-l)(2\weight+\ratioequ(1-\ratioequ)(1-\weight)\agentNum)<\weight\sqrt{\phi_{\min}} \cdot \int_{\phi_{\min}}^{\phi_{\max}}\sqrt{\phi}\virtualcostall(\phi)d\phi,$. We have the optimal allocation rule for virtual cost $\phi$
    \begin{equation}
    A(\phi) = \frac{1}{\sqrt{\phi}}\cdot\frac{B/s-l}{\int_{\phi_{\min}}^{\phi_{\max}}\sqrt{\phi}\virtualcostall(\phi)d\phi}.
    \end{equation}
    We can equivalently present the solution in terms of cost $\cost$:
    \begin{equation}
      A_i(\cost)=  \frac{1}{\sqrt{\phi_{i}(\cost)}}\cdot\frac{B/s-l}{\sum\limits_{i}q_i\int_{\cost_{\min}}^{\costthreshold{i}}\sqrt{\phi_i(\cost)}f_i(c)dc}.
    \end{equation}
    \item Case 2: $\frac{Q_c(\phi_{\min})}{R_c(\phi_{\min})}\le \frac{B/s-l}{\weight\ratioequ}<\frac{Q_c(\phi_{\max})}{R_c(\phi_{\max})}$. Let $\phi' \in [\phi_{\min},\phi_{\max})$ such that $\frac{Q_c(\phi')}{R_c(\phi')}=\frac{B/s-l}{\weight\ratioequ}$. 
    \begin{itemize}
        \item Case 2(a): $\frac{B/s-l}{Q_c(\phi')}\le 1$. We have 
        \begin{equation}
        A(\phi)=\begin{cases}\chi\triangleq\frac{B/s-l}{Q_c(\phi')}, &\text{if} \ \ \phi \le \phi',\\
    \frac{1}{\sqrt{\phi}} \cdot\frac{B/s-l-\chi\int_{\phi_{\min}}^{\hat{\phi}}\phi\virtualcostall(\phi)d\phi}{\int_{\phi'}^{\phi_{\max}}\phi\virtualcostall(\phi)d\phi} , & \text{if}\ \ \phi < \phi'.
    \end{cases}
        \end{equation}
    The solution in terms of cost $\cost$ is 
    
    \begin{equation}
	A_i(\cost) = \begin{cases}\chi\triangleq\frac{B/s-l}{Q_c(\phi')}, &\text{if}\ \ \phi_i(\cost) \le \phi', \\
	\frac{1}{\sqrt{\phi_{i}(\cost)}}\cdot\frac{B/s-l-\chi\sum\limits_{i}q_i\int_{\phi_{i\min}}^{\hat{\phi}}\phi_i(c)\virtualcost{i}(\phi_i)d\phi_i}{\sum\limits_{i}q_i\int_{\phi'}^{\phi_{i\max}}\sqrt{\phi_i(\cost)}\virtualcost{i}(\phi_i)d\phi_i}, &\text{if}\ \ \phi_i(\cost) > \phi'.
	\end{cases}
	\end{equation}
        \item Case 2(b): $\frac{B/s-l}{Q_c(\phi')}>1$. Let $\phi^*$ be such that $Q_c(\phi^*)=B/s-l$. Then we have 
        \begin{equation}
        A(\phi)=\begin{cases} 1, &\text{if}\ \ \phi \le \phi^*,\\
    \frac{1}{\sqrt{\phi}} \cdot\frac{B/s-l-\int_{\phi_{\min}}^{\phi^*}\phi\virtualcostall(\phi)d\phi}{\int_{\phi^*}^{\phi_{\max}}\phi\virtualcostall(\phi)d\phi} , &\text{if} \ \ \phi < \phi^*.
    \end{cases}
        \end{equation}
    The solution in terms of cost $\cost$ is 
    
    \begin{equation}
	A_i(\cost) = \begin{cases}1, &\text{if}\ \ \phi_i(\cost) \le \phi^*, \\
	\frac{1}{\sqrt{\phi_{i}(\cost)}}\cdot\frac{B/s-l-\sum\limits_{i}q_i\int_{\phi_{i\min}}^{\hat{\phi}}\phi_i(c)\virtualcost{i}(\phi_i)d\phi_i}{\sum\limits_{i}q_i\int_{\hat{\phi}}^{\phi_{i\max}}\sqrt{\phi_i(\cost)}\virtualcost{i}(\phi_i)d\phi_i}, &\text{if}\ \ \phi_i(\cost) > \phi^*.
	\end{cases}
	\end{equation}
    
    \end{itemize}
    \item  Case 3:  $\frac{B/s-l}{\weight\ratioequ}\ge\frac{Q_c(\phi_{\max})}{R_c(\phi_{\max})}$. We have 
    \begin{equation}
    A(\phi) =\chi= \frac{B/s-l}{\int_{\phi_{\min}}^{\phi_{\max}}\phi\virtualcostall(\phi)d\phi}.
    \end{equation}
     The solution in terms of cost $\cost$ is 
    \begin{equation}
      A_i(\cost) = \chi=  \frac{B/s-l}{\sum\limits_{i}q_i\int_{\phi_{i\min}}^{\phi_{i\max}}\phi_i(\cost)\virtualcost{i}(\phi_i)d\phi_i}.
    \end{equation}
\end{itemize}

In summary, we can present the solution of continuous case for cost $c\le \costthreshold{i}$ in group $i$ as follows:
\begin{equation}\label{equ_continuousA}
	A_i(\cost) = \begin{cases}\chi, &\text{if} \ \  \phi_{i}(\cost) \le \hat{\phi}, \\
	\frac{1}{\sqrt{\phi_{i}(\cost)}}\cdot\frac{B/s-l-\chi\sum\limits_{i}q_i\int_{\phi_{i\min}}^{\hat{\phi}}\phi_i(c)\virtualcost{i}(\phi_i)d\phi_i}{\sum\limits_{i}q_i\int_{\hat{\phi}}^{\phi_{i\max}}\sqrt{\phi_i(\cost)}\virtualcost{i}(\phi_i)d\phi_i}, &\text{if} \ \  \phi_{i}(\cost) > \hat{\phi}.
	\end{cases}
	\end{equation}
 Here, the constants $\chi$ and $\hat{\phi}$ are defined as follows:
 \begin{itemize}
     \item If $\frac{B/s-l}{\weight\ratioequ}<\frac{Q_c(\phi_{\min})}{R_c(\phi_{\min})}$, then  $\chi=0$ and $\hat{\phi}<\phi_{i\min}$ for all $i$.
     \item If $\frac{Q_c(\phi_{\min})}{R_c(\phi_{\min})}\le \frac{B/s-l}{\weight\ratioequ}<\frac{Q_c(\phi_{\max})}{R_c(\phi_{\max})}$, then   $\chi=\min\left\{1,\frac{B/s-l}{Q_c(\phi')}\right\}$ where $\phi'$ satisfies $\frac{Q_c(\phi')}{R_c(\phi')}=\frac{B/s-l}{\weight\ratioequ}$, and $\hat{\phi}$ satisfies $\frac{Q_c(\hat{\phi})}{\max\left\{1,R_c(\hat{\phi})/\weight\ratioequ\right\}}=B/s-l$.
     \item    If $ \frac{B/s-l}{\weight\ratioequ}\ge\frac{Q_c(\phi_{\max})}{R_c(\phi_{\max})}$, then $\chi = \frac{B/s-l}{\sum\limits_{i}q_i\int_{\phi_{i\min}}^{\phi_{i\max}}\phi_i(\cost)\virtualcost{i}(\phi_i)d\phi_i}$ and $\hat{\phi} = \phi_{\max}$.
 \end{itemize}

$\hfill\square$

\section{Proof of Corollary \ref{cor_opt_pro}}

Corollary \ref{cor_opt_pro} is a special case of Theorem \ref{the_opt_pro}. The proof is in Case 1 in the proof of  Theorem \ref{the_opt_pro}. $\hfill\square$

\section{Proof of Proposition \ref{Pro_fullpart_benefit}}\label{appendixsec_fullpart_benefit}


Under condition 1) of  Proposition \ref{Pro_fullpart_benefit},  the solution falls into Case 1 in the proof of Theorem \ref{the_opt_pro}. Plugging the solution of Case 1 in the proof of Theorem \ref{the_opt_pro} to the objective function, we have the following optimization problem over participation profile $\ratioVecequ$:
    \begin{equation}
    \min_{\ratioVecequ} \ \  T^*(\ratioVecequ) = \frac{\weight}{\agentNum}\left(U(\ratioVecequ)-\frac{1}{\ratioequ}\right),
    \end{equation}
    where \begin{equation}
        U(\ratioVecequ) \triangleq  \frac{1}{\left(\ratioequ\right)^2}\cdot \frac{\left(\sum_{i}\groupProb{i}\int_{\cost_{\min}}^{\costthreshold{i}}\sqrt{\phi_{i}(\cost;\ratioVecequ)}\costD_i(\cost)d\cost\right)^2}{B/s-l(\ratioVecequ)}.
    \end{equation}
    
    Notice that $\ratioequ = \sum_{i} q_i \ratioequg_i$ by the definition of average participation rate. Thus, $\frac{\partial \ratioequ}{\partial \ratioequg_i}=q_i$. Define
    
    \begin{equation}
    r(\ratioVecequ) \triangleq \left(\sum_{i}\groupProb{i}\int_{\cost_{\min}}^{\costthreshold{i}}\sqrt{\phi_{i}(\cost;\ratioVecequ)}\costD_i(\cost)d\cost\right)^2.
    \end{equation}
    And recall that 
    \begin{equation}
 l(\ratioVecequ)\triangleq \sum_{j} q_j \ratioequg_i(\pplequ{\costthreshold{j}}{j}-\gpplequ{\costthreshold{j}}{j}-w(\ratioequ)),
\end{equation}
and 
    \begin{equation}
        \Delta_i \triangleq \ppl{\costthreshold{i}}{i}- \gppl{\costthreshold{i}}{i}.
    \end{equation}
    We can write  $U(\ratioVecequ)$ as follows:
    
    \begin{equation}
    U(\ratioVecequ) =\frac{1}{\left(\ratioequ\right)^2}\cdot \frac{r(\ratioVecequ)}{B/s-l(\ratioVecequ)} = \frac{1}{\left(\ratioequ\right)^2}\cdot \frac{r(\ratioVecequ)}{B/s-\sum_{i}q_i\ratioequ_i(\Delta_i-w(\ratioequ))}.
    \end{equation}
    Before presenting the derivative of the objective function with respect to group $i$'s participation ratio $\ratioequg_i$, i.e.,  $\frac{\partial T^{*}(\ratioVecequ)}{\partial \ratioequg_i}$, we give the the derivative of $l(\ratioVecequ)$ with respect to  $\ratioequg_i$, which appears in $\frac{\partial T^{*}(\ratioVecequ)}{\partial \ratioequg_i}$. Note that $\frac{\partial l(\ratioVec)}{\partial \ratioequ_i}=\frac{\partial \sum_{j}q_j \ratioequg_j(\Delta_j-w(\ratioequ))}{\partial \ratioequ_j}$. That is,
    \begin{equation}
    \begin{aligned}
    &\ \ \frac{\partial \sum_{j}q_j \ratioequg_j(\Delta_j-w(\ratioequ))}{\partial \ratioequ_i}\\
    =&\ \  \frac{\partial q_i \ratioequ_i (\Delta_i -w(\ratioequ))}{\partial \ratioequ_i} +   \frac{\partial \sum_{j\neq i}q_j \ratioequ_j (\Delta_j -w(\ratioequ))}{\partial \ratioequ_i}\\
     =&\ \  q_i(\Delta_i - w(\ratioequ)) +q_i\ratioequ_i\left(\frac{\partial\Delta_i }{\partial \ratioequ_i}-q_iw'(\ratioequ)\right)+ \sum_{j\neq i}q_j\ratioequ_j\left(\frac{\partial\Delta_j }{\partial \ratioequ_i}-q_iw'(\ratioequ)\right)\\
    =&\ \   q_i(\Delta_i - w(\ratioequ))+q_i\ratioequ_i\frac{\partial \Delta_i}{\partial \ratioequ_i}+\sum_{j\neq i}q_j\ratioequ_j\frac{\partial \Delta_j}{\partial \ratioequ_i} - q_iw'(\ratioequ)\sum_{j}q_j\ratioequ_j\\
   = &\ \   q_i(\Delta_i - w(\ratioequ))+q_i\ratioequ_i\frac{\partial \Delta_i}{\partial \ratioequ_i}+\sum_{j\neq i}q_j\ratioequ_j\frac{\partial \Delta_j}{\partial \ratioequ_i} -  q_iw'(\ratioequ)\ratioequ\ \ \ \ \ \text{using $\sum_{j}q_j\ratioequ_j=\ratioequ$}\\
    =&\ \  q_i \left(\ratioequ_i\frac{\partial \Delta_i}{\partial \ratioequ_i}+\Delta_i - w(\ratioequ)-\ratioequ w'(\ratioequ)\right)+\sum_{j\neq i}q_j\ratioequ_j\frac{\partial \Delta_j}{\partial \ratioequ_i}.
    \end{aligned}
    \end{equation}
    
    The derivative of the objective function with respect to group $i$'s participation ratio $\ratioequg_i$ is 
    \begin{equation} 
       \frac{\partial T^{*}(\ratioVecequ)}{\partial \ratioequg_i} =   \frac{\weight}{\agentNum}\frac{\partial U(\ratioVecequ)}{\partial \ratioequg_i}+\frac{\weight\groupProb{i}}{\agentNum\left(\ratioequ\right)^2}.
    \end{equation}
    
    where

    
    \begin{equation}
    \begin{aligned}
    \frac{\partial U(\ratioVecequ)}{\partial \ratioequg_i} =&\ \  \frac{\frac{\partial r(\ratioVecequ)}{\partial \ratioequg_i}}{(\ratioequ)^2(B/s-l(\ratioVecequ))}-\frac{2\groupProb{i}r(\ratioVecequ)}{(\ratioequ)^3(B/s-l(\ratioVecequ))}\\
    &\ \ +\frac{q_i r(\ratioVecequ)(\ratioequ_i\frac{\partial \Delta_i}{\partial \ratioequ_i}+\Delta_i-w(\ratioequ)-\ratioequ w'(\ratioequ))}{(\ratioequ)^2(B/s-l(\ratioVecequ))^2}+\frac{r(\ratioVecequ)\sum\limits_{j\neq i}q_j\ratioequ_j\frac{\partial \Delta_j}{\partial \ratioequ_i}}{(\ratioequ)^2(B/s-l(\ratioVecequ))^2}.
    \end{aligned}
    \end{equation}

    We wish to characterize the conditions under which  $\frac{\partial T^*(\ratioVecequ)}{\partial \ratioequg_i}\le  0, \forall i$, i.e., 
    \begin{equation}
        \frac{\partial T^*(\ratioVecequ)}{\partial \ratioequg_i}= \frac{\weight}{\agentNum}\frac{\partial U(\ratioVecequ)}{\partial \ratioequg_i}+\frac{\weight\groupProb{i}}{\agentNum\left(\ratioequ\right)^2}\le 0.
    \end{equation}
    This is equivalent to 
    \begin{equation}
        \frac{\partial T^*(\ratioVecequ)}{\partial \ratioequg_i}= \frac{\partial U(\ratioVecequ)}{\partial \ratioequg_i}+\frac{\groupProb{i}}{\left(\ratioequ\right)^2}\le 0,
    \end{equation}
    by removing $\weight/\agentNum$, which is greater than zero. Plugging in $\frac{\partial U(\ratioVecequ)}{\partial \ratioequg_i}$, we have 
    
    \begin{equation}
    \begin{aligned}
    &\frac{\frac{\partial r(\ratioVecequ)}{\partial \ratioequg_i}}{(\ratioequ)^2(B/s-l(\ratioVecequ))}-\frac{2\groupProb{i}r(\ratioVecequ)}{(\ratioequ)^3(B/s-l(\ratioVecequ))}+\frac{q_i r(\ratioVecequ)(\ratioequ_i\frac{\partial \Delta_i}{\partial \ratioequ_i}
    +\Delta_i-w(\ratioequ)-\ratioequ w'(\ratioequ))}{(\ratioequ)^2(B/s-l(\ratioVecequ))^2}\\
    &+\frac{r(\ratioVecequ)\sum\limits_{j\neq i}q_j\ratioequ_j\frac{\partial \Delta_j}{\partial \ratioequ_i}}{(\ratioequ)^2(B/s-l(\ratioVecequ))^2}+\frac{\groupProb{i}}{\left(\ratioequ\right)^2}\le 0.
    \end{aligned}
    \end{equation}
    Multiplying $(\ratioequ)^2(B/s-l(\ratioVecequ))^2$, which is greater than zero (recall that $(B/s-l(\ratioVecequ)>0$ implicitly. Otherwise, there is no positive solution of allocation rule) at both sides, we have
\begin{equation}\label{equ_proof1}
        \begin{aligned}
     &\frac{\partial r(\ratioVecequ)}{\partial \ratioequg_i}(B/s-l(\ratioVecequ))-\frac{2\groupProb{i}r(\ratioVecequ)(B/s-l(\ratioVecequ))}{\ratioequ}+q_i r(\ratioVecequ)(\ratioequ_i\frac{\partial \Delta_i}{\partial \ratioequ_i}+\Delta_i-w(\ratioequ)-\ratioequ w'(\ratioequ))\\
     &+r(\ratioVecequ)\sum\limits_{j\neq i}q_j\ratioequ_j\frac{\partial \Delta_j}{\partial \ratioequ_i}+q_i(B/s-l(\ratioVecequ))^2\le 0.
    \end{aligned}
\end{equation}
So, to ensure $\frac{\partial T^*(\ratioVecequ)}{\partial \ratioequg_i}\le  0, \forall i$, we need to have the inequality in (\ref{equ_proof1}) holds for all $i$.     Define


       

    \begin{equation}\label{equ_D}
        \begin{aligned}
        D_i(\ratioVecequ,\budget,\weight,\agentNum,\groupProbVec) \triangleq &\ \ \frac{(\ratioequ_i \Delta'_i(\ratioequ_i)+\Delta_i(\ratioequ_i)-w(\ratioequ))}{\ratioequ}+\frac{\frac{\partial r(\ratioVecequ)}{\partial\ratioequg_i}(B/\agentNum-l(\ratioVecequ))}{\ratioequ q_i r(\ratioVecequ)}\\
        &\ \ -\frac{2(B/s-l(\ratioVecequ))}{(\ratioequ)^2}+ \frac{(B/s-l(\ratioVec))^2}{\ratioequ r(\ratioVecequ)}+\frac{\sum\limits_{j\neq i}q_j\ratioequ_j\frac{\partial \Delta_j}{\partial \ratioequ_i}}{q_i\ratioequ}.
    \end{aligned}
    \end{equation}

           
       If $w'(\ratioequ)\ge  D_i(\ratioVecequ,,\budget,\weight,\agentNum,\groupProbVec)$, for all $\ratioVec$ in which $\ratioequ_i\ge \ratio_{\min}>0$, the inequality in (\ref{equ_proof1}) holds, i.e., 
        we can obtain $\frac{\partial T^*(\ratioVecequ)}{\partial \ratioequg_i}\le  0, \forall i.$ Thus, the objective function is decreasing in the group participation ratio, and the optimal participation ratio for group $i$ is one, for all groups $i$. $\hfill\square$
    
\section{Proof of Proposition \ref{pro_fullpart_loss}}\label{appendixsec_fullpart_loss}
The idea of the proof is similar to the proof of Proposition \ref{Pro_fullpart_benefit}. We can similarly arrive at  (\ref{equ_proof1}). Define

        
        \begin{equation}\label{equ_delta}
            \begin{aligned}
            \delta_i (\ratioVecequ,\budget,\weight,\agentNum,\groupProbVec) \triangleq  
            &\ \ \frac{w(\ratioequ)+\ratioequ w'(\ratioequ)-\Delta_i(\ratioequ_i)}{\ratioequ_i}-\frac{\frac{\partial r(\ratioVecequ)}{\partial\ratioequg_i}(B/\agentNum-l(\ratioVecequ))}{\ratioequ_i q_i r(\ratioVecequ)}\\
        &\ \ +\frac{2(B/s-l(\ratioVecequ))}{\ratioequ_i \ratioequ}- \frac{(B/s-l(\ratioVec))^2}{\ratioequ_i r(\ratioVecequ)}-\frac{\sum\limits_{j\neq i}q_j\ratioequ_j\frac{\partial \Delta_j}{\partial \ratioequ_i}}{q_i\ratioequ_i}.
        \end{aligned}
        \end{equation}
          
        If $\Delta'_i(\ratioequ_i)\le   \delta_i(\ratioVec, \budget,\weight,\agentNum,\groupProbVec),$ for all $\ratioVec$ in which $\ratioequ_i\ge \ratio_{\min}>0$, the inequality in (\ref{equ_proof1}) holds, i.e., 
        we can obtain $\frac{\partial T^*(\ratioVecequ)}{\partial \ratioequg_i}\le  0, \forall i.$ Thus, the objective function is decreasing in the group participation rate, and the optimal participation rate for group $i$ is one, for all groups $i$. $\hfill\square$

\end{document}